\documentclass[11pt]{article}
\usepackage[margin=1in]{geometry}
\usepackage{amsmath,amssymb,amsthm,booktabs,mathtools}
\usepackage{array}
\usepackage{graphicx,float}
\usepackage{microtype}
\usepackage{xr-hyper}
\usepackage[colorlinks=true,linkcolor={[rgb]{0.10,0.30,0.80}},citecolor={[rgb]{0.10,0.30,0.80}},urlcolor={[rgb]{0.10,0.30,0.80}}]{hyperref}
\newtheorem{theorem}{Theorem}
\newtheorem{proposition}{Proposition}
\newtheorem{corollary}{Corollary}
\theoremstyle{remark}
\newtheorem{remark}{Remark}
\newtheorem{openproblem}{Open problem}
\newcommand{\R}{\mathbb{R}}
\newcommand{\E}{\mathbb{E}}
\newcommand{\Var}{\operatorname{Var}}
\newcommand{\Prob}{\mathbb{P}}
\newcommand{\Xbar}{\bar X}
\newcommand{\ybar}{\bar y}
\newcommand{\ii}{\mathrm{i}}
\newcommand{\dd}{\,\mathrm{d}}
\newcommand{\Pacc}{P_{\mathrm{acc}}}
\newcommand{\mest}{m_{\mathrm{est}}}
\newcommand{\ntest}{n_{\mathrm{test}}}
\newif\ifblind\blindfalse %
\newif\ifsubmission\submissionfalse
\ifdefined\BLINDBUILD\blindtrue\submissiontrue\fi
\ifdefined\SUBMISSIONBUILD\submissiontrue\fi
\usepackage{setspace}
\ifsubmission\fi
\title{Precision follows the decision: how well an acceptance\\
probability is known for mean--spread release rules}
\ifblind
\author{}\date{}
\else
\IfFileExists{_p3_authorblock.tex}{\author{Remus Osan\\[2pt]
\normalsize Gilead Sciences, Foster City, CA 94404, USA\\
\normalsize \texttt{remus.osan@gilead.com}}
\date{}
}{\author{}\date{}}
\fi

\begin{document}
\maketitle

\ifsubmission
\begin{abstract}
Many release rules accept a lot when a statistic of the sample mean and standard deviation stays
within a limit, and the probability of passing can be computed exactly for a broad class of parent
distribution. We ask how well that probability is known once the parent has been estimated.
Under a reference posterior the answer is an interval, narrow where a lot clearly passes or fails
and wide where it is contested. In ninety-eight published content-uniformity lots that position,
not sample size, sets the width: among forty-eight lots assayed at ten units, and so equally
informed about the spread, the $95\%$ interval runs from essentially zero width to $0.96$. As
test and estimation samples grow together the interval narrows only slowly, because the
operating characteristic steepens as the posterior contracts; $1207$ occupational exposure
assessments under a half-plane rule agree. For a pure-spread rule under a Normal parent the
interval has an exact chi-square form that contains Bergum's lower bound as one endpoint. A
strongly right-skewed parent lowers the pessimistic endpoint by up to a fifth and widens the
interval by up to $0.22$. At the content-uniformity operating point studied, under a Normal parent, the usual symmetric
error bar puts the lower endpoint $0.21$ too high and returns an upper endpoint above one.
\end{abstract}
\else
\begin{abstract}
Many release rules accept a lot when a statistic built from the sample mean and standard
deviation stays within a limit, and the probability that a process passes such a rule can be
computed exactly for a broad class of parent distributions. This paper asks how well that probability is itself known once
the parent has been estimated. Under a stated reference posterior the answer is an interval, narrow where a lot clearly
passes or fails and wide where it is contested. In ninety-eight published content-uniformity lots
that position, not the sample size, sets the width: among the forty-eight lots assayed at the same
ten units, and so equally informed about the spread, the $95\%$ interval runs from essentially
zero width to $0.96$. With the operating point held fixed while test and estimation sample
grow together, the interval narrows only slowly, the contraction of the posterior being offset by
the steepening of the operating characteristic; a calibrated study at the sample sizes of $1207$
occupational exposure assessments shows the same under a half-plane rule. For a pure-spread rule
under a Normal parent the interval has an exact chi-square form in the plug-in value and the two
sample sizes, extending to the mean--spread wedge with a measured residual and containing
Bergum's lower bound as one endpoint. With the parent's posterior formed from the same published
summaries through the evaluator's own likelihood, the qualitative pattern survives but shape
moves both endpoints and width: on the same lots a strong right skew lowers the pessimistic
endpoint by up to a fifth and widens the interval by up to $0.22$. At the content-uniformity operating point studied, under a Normal parent, the symmetric error
bar in standard use puts the lower endpoint $0.21$ too high and returns an upper endpoint above one,
even with the full posterior covariance.
\end{abstract}
\fi

\vspace{0.4em}

\section{Introduction}\label{sec:intro}

\paragraph{The rules.} A large class of regulated decisions shares one form: measure $n$ units,
compute the sample mean $\Xbar$ and the sample spread $s$, and release the lot only if the pair
lands in a fixed region $A$ of the $(\Xbar,s)$ plane. The region is the standard's entire
notion of acceptable: how far the mean may drift, how much scatter is tolerable, and the rate
at which one is traded against the other. USP $\langle905\rangle$ releases a batch of medicine this way: its region, a wedge, charges deviation of the mean and dispersion jointly against a single limit. EN\,689 declares a workplace's air compliant through a half-plane that charges only the high side. A process capability index admits a supplier's part when the sample $C_{pk}$ clears a contracted floor. A percent-within-limits specification sets acceptance and pay for a highway pavement lot (\S\ref{sec:setup}). Attached to any such rule and any
hypothesized process is one operational number: the probability that a lot from that process
passes. This paper is about what that number is worth.

\paragraph{The gap this closes.} Recent work \cite{paper1}, cited below as the exact-evaluation paper, shows that for this family
the acceptance probability
\begin{equation}\label{eq:P}
\Pacc(F)=\Prob\big[(\Xbar,s)\in A\big],\qquad X_1,\dots,X_n\ \text{i.i.d.}\ \sim F,
\end{equation}
is computable exactly and deterministically, for any parent $F$ meeting its integrability
condition, at a cost that for a fixed grid does not grow with $n$. That answers ``what is the pass probability if the process looks like
\emph{this}?'' It does not answer the question a release decision poses, ``what is the pass
probability, given that we do not know exactly what the process looks like?'' In practice the parent is estimated: from a development batch, a calibration run, a handful of historical lots, or, in \S\ref{sec:cehd}, from eight air samples. Simulation-based content-uniformity tools built on measured drug-substance properties report a point value computed at the estimate $\hat\theta$ \cite{morrison2021,chu2025}. That value silently carries the uncertainty of its input.

\paragraph{Why exactness changes the shape of the answer.} The usual route to an interval on a simulated quantity is nested: draw parameters, then simulate a lot at each draw. That route has two entangled error sources, and its cost means that in practice it is skipped. An exact evaluator removes one source entirely. Write $g(\theta):=\Pacc(F_\theta)$. Because this map is deterministic, uncertainty in $\Pacc$ is nothing but uncertainty in $\theta$ carried through $g$. With $\theta$ drawn from whatever law $\pi$ the analyst holds, the interval is read off the law of the single random variable $g(\theta)$. That induced law is the \emph{pushforward} $g_\#\pi$;
\S\ref{sec:pushforward} defines it, and the rest of the paper measures it.

\paragraph{What exactness is worth, in the exact-evaluation paper's own currency.} The exact-evaluation
paper \cite{paper1} measures the \emph{approximation error}: the gap between the exact acceptance
probability and the incumbent Normal-parent value for the same lot, a defect of the arithmetic
that would be present even if the process were known exactly. At the operating points examined
there it is of the order of a tenth and changes sign. Monte Carlo, the other incumbent, carries no
such bias, only sampling error and cost. This paper measures the \emph{interval width} $W$ of
\eqref{eq:width}: how much of the unit interval the acceptance probability is still free to occupy
once one admits that the parent was estimated. That is a limit of the evidence, not of the
arithmetic. Among ninety-eight published content-uniformity lots, a curated convenience sample,
the contested ones, whose \emph{plug-in} acceptance probability (the value computed by treating
the estimated mean and spread as if they were the true ones, which is the number practice
reports) lies between $0.20$ and $0.99$, have a
median width of $W=0.77$ (\S\ref{sec:reallots}), several times the approximation error at its
largest. The two are a sequence rather than rivals: making the evaluation exact drives the
approximation error to zero, so the width that remains becomes visible and attributable.

\paragraph{Where the evidence comes from.} Section~\ref{sec:ladder} presents three studies, each
more realistic and less analytically controlled than the one before. The first is an analytic
anchor, a toy version of the USP $\langle905\rangle$ content-uniformity rule: the unclipped
wedge \eqref{eq:wedge} with three units, $k=1$ and a Uniform parent, a non-Normal case for which $g$ is known in
closed form. The second is pharmaceutical: the acceptance-value condition of $\langle905\rangle$, the rule
for releasing drug-product lots, under a Normal parent. It is evaluated first at one calibrated
operating point on the Stage-2 constants, then on ninety-eight published lots at Stage~1. Stage~1 is evaluated as written; the Stage-2
evaluations omit that stage's individual-unit limits, and the two-stage sequence is left to
Appendix~\ref{app:twostage}. The
third is occupational hygiene: the rule that accepts when the upper tolerance limit of workplace
exposure sits below the exposure limit. Its parents are lognormal with median geometric standard
deviation $2.20$ (GSD, the factor $e^{\sigma}$ of the log-scale spread), and its sample sizes,
median $8$, are those of $1207$ exposure assessments extracted from a public regulatory archive
of over a million measurements \cite{cehd}; only the sample sizes are empirical. They form a ladder:
the anchor checks the evaluator the lot study uses against a closed form, and the lot study's construction is
repeated verbatim on the exposure study. Apart from the anchor, every study is run first under a
Normal parent (lognormal for exposure), where the posterior is classical and the comparison with
the existing, Normal-based methods is like for like. Only once the Normal case is exhausted does
\S\ref{sec:skewed} replace it, repeating the lot study under skewed parents, a skew-Normal
calibrated on three hundred published tablets and a right-skewed Beta stress case, each with its
own likelihood; that is where the exact evaluator's reach beyond the Normal is used.

\paragraph{What we find.} One structure explains all three studies. Where the spread is the coordinate that matters, the interval on $\Pacc$ is the posterior range of $\sigma$ mapped through the operating characteristic. Its width is therefore the product of two factors. The first is how far apart the posterior quantiles of $\sigma$ sit, which depends on the sample size and nothing else. The second is how steep the operating characteristic is between them, which depends on where the lot sits on the rule. Section~\ref{sec:closed} makes this exact: the endpoints are two evaluations of $g$, one at each $\sigma$-quantile. For a Normal parent they have a closed form, a chi-square formula in the plug-in value and the two sample sizes, which gives Bergum's lower bound in one line together with the interval that bound belongs to. Section~\ref{sec:unify} turns it into a law. The two factors are not comparable in reach.
\emph{At a fixed sample size, position alone moves the width across almost the whole unit
interval; at a fixed position away from an exact decision boundary, more units narrow it only as $n^{-1/2}$.} A lot deep inside the
acceptance region is pinned down at ten units; a lot near the knee is not pinned down at
fifty-five. Whether a given sample size suffices is therefore not a property of the sample
size. Parent shape, the subject of the exact-evaluation paper, moves both the endpoints and the
width: under the skew seen in segregating runs the Normal lower endpoint sits too high by up to a
fifth, on lots that appear to pass comfortably (\S\ref{sec:skewed}).

Two of the ninety-eight lots settle this without further machinery. Both were assayed at $n=10$,
so the data pin down $\sigma$ equally well in relative terms: each posterior range spans the same
$2.65$-fold factor, $[1.22,\,3.25]$ for the first and $[3.33,\,8.85]$ for the second, because that
factor depends only on the sample size. The lots are alike in nothing else. Their intervals on $\Pacc$ are $[0.993,\,1.000]$ and
$[0.066,\,1.000]$, $0.007$ and $0.934$ wide, a factor of $137$ apart, because at a given mean the
$\langle905\rangle$ rule of \S\ref{sec:setup} is a ceiling on the sample spread, and that
ceiling lies above the first lot's whole $\sigma$-range and inside the second's. An error bar scaled by $1/\sqrt n$, necessarily
the same for both, describes the first and misses the second by most of the unit interval
(Figure~\ref{fig:premise}; Appendix~\ref{app:lots}).

\paragraph{What is new here, and what is not.} Posterior inference for the probability of passing a compendial test, including its percentiles, is established \cite{leblond2014,lewis2016}. The chi-square identity behind the pure-spread case is known from the estimated-parameter control-chart literature \cite{jensen2006,epprecht2015}. Finite-sample decision instability and nonlinear amplification at a fixed boundary are studied in \cite{jiang2026a,jiang2026b}. The exact evaluator is imported from the exact-evaluation paper \cite{paper1}. What this paper adds is
narrower. Two results are proved. When only the spread matters, the acceptance probability falls
steadily as $\sigma$ grows, so the ends of the posterior interval for $\sigma$ map directly to the
ends of the interval for the acceptance probability, and for a Normal parent this becomes a chi-square formula in the plug-in acceptance probability and the two sample sizes
(\S\ref{sec:closed}). And the sample that estimates the process and the sample the test draws are
different quantities; keeping them apart gives three ways for the sample size to grow, which
narrow the interval at very different rates (\S\ref{sec:unify}). Two results are computed, not
proved. A width law $\mathcal A(g)$ predicts the interval width from the plug-in acceptance
probability alone, nearly independently of sample size; it is checked, not proved, on the
exact-evaluation paper's capability and percent-within-limits rules (\S\ref{sec:transfer}). And
under one posterior the exact interval is compared with the two shortcuts in use: the
delta-method error bar, on the probability and on the logit scale, and the worst case over the
corners of a parameter box (\S\S\ref{sec:corner}, \ref{sec:link}). The illustrations are a convenience sample of published lots, a realistic sample-size
distribution rather than $1207$ independent validations, and declared shape scenarios
(\S\ref{sec:skewed}). A focused search found no earlier source with this combination.

A practitioner already asks four questions, and the paper answers each in terms of that two-factor product.

\begin{itemize}\itemsep3pt
\item \emph{How much does a reported acceptance probability actually determine?} It depends
on where that probability falls and not on the sector: on ninety-eight published lots under USP
$\langle905\rangle$ the interval collapses at both ends and opens in the middle, and a calibrated
study at the sample sizes of $1207$ occupational exposure assessments reproduces the shape
(\S\ref{sec:reallots}, \S\ref{sec:cehd}).
\item \emph{Does a larger sample narrow the interval?} Only if the extra units are spent on
estimation while the rule stays as the standard writes it. When the test grows with the
estimation sample at a recalibrated limit, the amplification $|\partial g/\partial\ln\sigma|$
grows as $\sqrt n$ while the posterior spread shrinks at the same rate, and the two cancel
(\S\ref{sec:unify}, Table~\ref{tab:designs}).
\item \emph{Which coordinate, if better known, would narrow it?} The spread, for a rule that
weights spread as $\langle905\rangle$ does. Correlation between units of the same lot, which
the sampling-plan literature around ASTM E2709 and E2810 treats as a central concern
\cite{bergum2015}, matters about as much as the mean (\S\ref{sec:decomp}).
\item \emph{Which of the approximations already in use can be trusted, and when?} One shortcut is the worst-case corner band, Bergum's method as standardized in ASTM E2709 and E2810 \cite{bergum1990,astm2709,astm2810}: put a confidence box around the uncertain process parameters, evaluate the acceptance probability at the corners of the box, and report the smallest and largest values. At the $\langle905\rangle$ point it finds the true minimum but misses the maximum. The rule is symmetric about the target, so the highest acceptance probability occurs at the target mean, inside the box rather than at a corner (\S\ref{sec:corner}). The symmetric delta band is the error bar $\hat g\pm1.96\,\mathrm{sd}$ that the delta method attaches to the plug-in value $\hat g$ defined in \S\ref{sec:setup}. At the $\langle905\rangle$ operating point, under a Normal parent, it puts the lower endpoint $0.21$ too high, even with the full covariance and a converged gradient. In every one of the $1207$ assessments it returns an upper endpoint that is not a probability. Carried on the logit scale it repairs the range and comes within about a hundredth of the exact interval at both operating points examined, but it offers no condition under which that agreement can be checked (\S\ref{sec:link}).
\end{itemize}

\begin{figure}[H]\centering
\includegraphics[width=\linewidth]{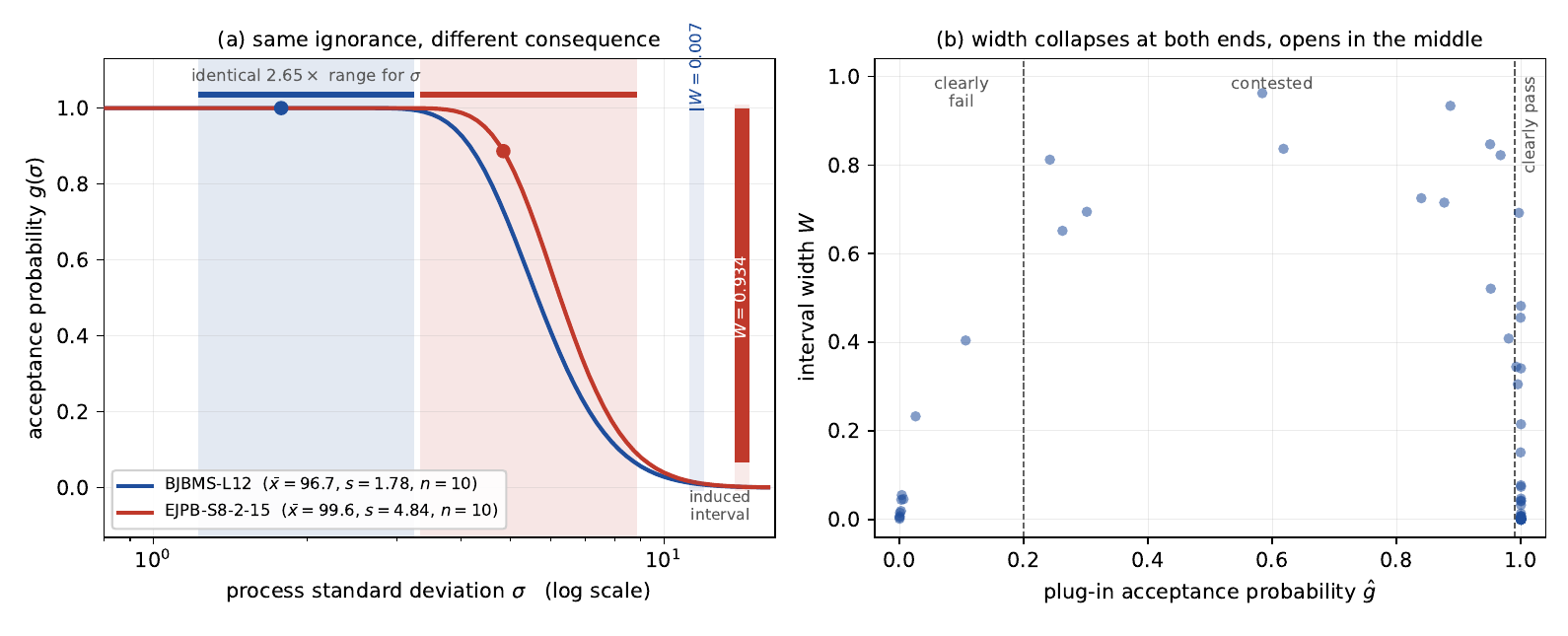}
\caption{The premise, in two views. \emph{(a)} Two of the ninety-eight lots, both at $n=10$, each in one
color throughout. The shaded bands are the values of the process standard deviation $\sigma$
consistent with the observed $s$, equal in width on the log scale because that $2.65$-fold
factor depends on the sample size alone; projecting each band through that lot's own
$g(\sigma)$ gives the induced intervals at the right, a hairline where the rule has saturated
and almost the whole unit interval at the knee. The curves hold $\bar x$ fixed; the bars are the
full posterior intervals over $\mu$ and $\sigma$ jointly, $[0.066,\,1.000]$ against
$[0.090,\,1.000]$ on the $\sigma$-slice for the second lot. \emph{(b)} All ninety-eight lots: the
width collapses wherever the plug-in probability saturates and opens only in between; dashed
lines are the band edges of Table~\ref{tab:reallots}, and the three-regime shape, not the cuts,
is the finding.}
\label{fig:premise}
\end{figure}

\paragraph{Two sectors, and why not three.} Only two of the four rules get a study of their own.
The exact-evaluation paper's other two, capability release and percent-within-limits
acceptance, have no public per-lot data this paper could build a study on (\S\ref{sec:transfer}),
so they serve instead as the test of whether the result carries over to rules it was not built
from. In plain terms, USP $\langle905\rangle$ adds two penalties, the distance of the sample mean
from target and a multiple of the sample standard deviation, and releases the lot when the total
stays under a limit, so its passing set is a \emph{wedge}; EN\,689 Annex~F penalizes only high
exposure, through one linear inequality, so its passing set is a \emph{half-plane}
(\S\ref{sec:setup}, Figure~\ref{fig:geom}).

\begin{figure}[tbp]\centering
\includegraphics[width=\linewidth]{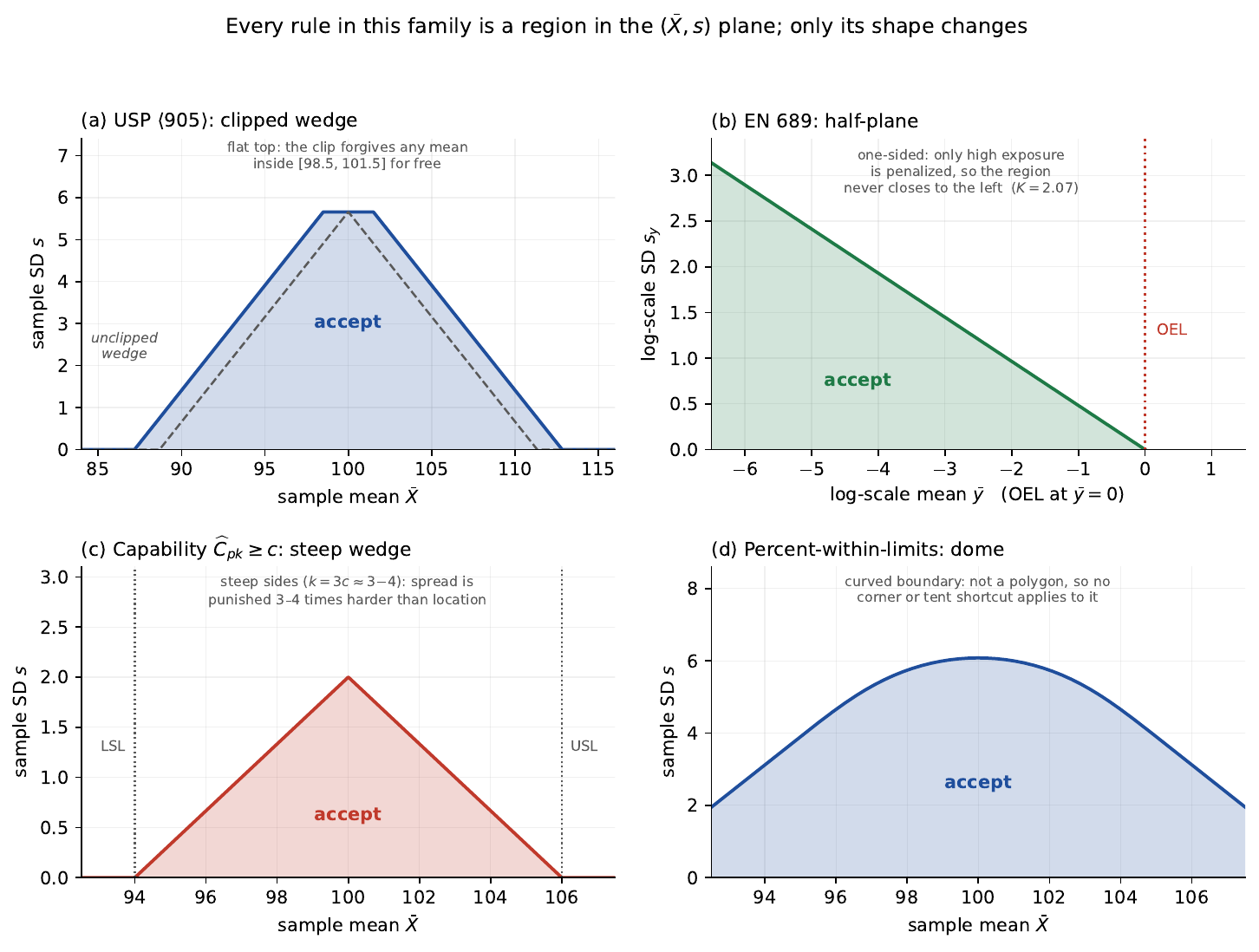}
\caption{The acceptance regions of this paper's two empirical studies (top) and of the exact-evaluation
paper's other two rules (bottom), in the $(\Xbar,s)$ plane. \emph{(a)} USP $\langle905\rangle$:
a wedge, flat-topped because the clip forgives any mean inside $[98.5,101.5]$; dashed, the same
rule without the clip. \emph{(b)} EN\,689 Annex~F: a half-plane of slope $-1/K$, open on the
safe side. \emph{(c)} Capability release: the wedge of (a) with a penalty on spread about one and
a half to two times larger, since $k=3c\approx3.3$--$4.7$ against $2.4$; with spread on the
vertical axis the sides are correspondingly flatter, of slope magnitude $1/k$. \emph{(d)} Percent-within-limits: a curved dome.}
\label{fig:geom}
\end{figure}

\section{Setup}\label{sec:setup}

Let $X_1,\dots,X_n$ be the measured values of $n$ units, $\Xbar=n^{-1}\sum_i X_i$ and
$s^2=(n-1)^{-1}\sum_i(X_i-\Xbar)^2$. A mean--spread rule accepts when $(\Xbar,s)\in A$. Two
geometries recur across the catalog of \cite{paper1}: the \emph{wedge}
\begin{equation}\label{eq:wedge}
|\Xbar-T|+k\,s\le L,
\end{equation}
whose regulated instance is the \emph{acceptance value} of USP General Chapter
$\langle905\rangle$, \emph{Uniformity of Dosage Units} \cite{usp905}, the harmonized
pharmacopeial requirement under which a batch is sampled and the content of each unit assayed
against the label claim. We model the acceptance-value condition alone, one clause of that
chapter; Appendix~\ref{app:twostage} treats the two-stage sequence. Contents are reported as a
percentage of label claim throughout, so a perfectly dosed unit reads $100$, and the batch is
accepted when $\mathrm{AV}=|M-\Xbar|+ks\le15$, where $M=\operatorname{clip}(\Xbar,98.5,101.5)$
forgives any mean inside the window $98.5$ to $101.5\%$. It is the most widely applied member
of the family and serves throughout as the reference rule. The second geometry is the
\emph{half-plane}
\begin{equation}\label{eq:halfplane}
\Xbar+K\,s\le c,
\end{equation}
whose instance in \S\ref{sec:cehd} is the EN\,689 Annex F / AIHA upper tolerance limit
$\ybar+K s_y\le\ln(\mathrm{OEL})$ on log-transformed exposures: a workplace is compliant when
the estimated upper tolerance limit of the exposure distribution sits below the occupational
exposure limit. Figure~\ref{fig:geom}(a,b) draws both; panels (c) and (d) are the exact-evaluation
paper's two other rules, which \S\ref{sec:transfer} prices without data of their own. The rules differ in sector, regulator and units, yet one evaluator serves all of them, and \S\ref{sec:unify} shows that they behave alike.
Throughout, write $F_\theta$ for the parent indexed by $\theta$ (location, spread, shape, and
where relevant a dependence parameter), and
\begin{equation}\label{eq:g}
g(\theta):=\Pacc(F_\theta,n,A)
\end{equation}
for the acceptance probability of rule $A$ at sample size $n$, and $\hat g=g(\hat\theta)$ for
its \emph{plug-in value} at the estimated parameter, the number practice reports; everything
that follows is a statement about this one function.

\paragraph{What computes \texorpdfstring{$g$}{g}, stated here.} The result this paper rests on
is carried by the exact-evaluation paper \cite{paper1}, not yet peer reviewed, so its statement and hypotheses are given
in full; the proof is not reproduced.

\begin{theorem}[Exact representation at fixed dimension; proved in \cite{paper1}]
\label{thm:engine}
Let $X_1,\dots,X_n$ be i.i.d.\ from an absolutely continuous law $F$ with $\E X^4<\infty$, and
let $\psi(t,u)=\int e^{\ii t x+\ii u x^{2}}\dd F(x)$ be the joint characteristic function of
$(X,X^{2})$. If $\iint_{\R^2}|\psi(t,u)|^{\,n}\dd t\dd u<\infty$, then
$(T_1,T_2)=(\sum_i X_i,\sum_i X_i^{2})$ has the bounded continuous joint density
\begin{equation}\label{eq:inv}
f_{T_1,T_2}(\tau_1,\tau_2)=\frac{1}{(2\pi)^{2}}\iint_{\R^{2}}\psi(t,u)^{n}\,
\mathrm{e}^{-\ii(t\tau_1+u\tau_2)}\dd t\dd u ,
\end{equation}
and for every measurable $A$,
\[
\Pacc(F,n,A)=\Prob\bigl[(\Xbar,s)\in A\bigr]
=\iint_{A^{*}}f_{T_1,T_2}(\tau_1,\tau_2)\dd\tau_1\dd\tau_2,
\]
\[
A^{*}=\Bigl\{(\tau_1,\tau_2):\ \tau_2>\tfrac{\tau_1^{2}}{n},\
\bigl(\tfrac{\tau_1}{n},\ \sqrt{\tfrac{\tau_2-\tau_1^{2}/n}{n-1}}\bigr)\in A\Bigr\}.
\]
If $\psi^{n}\in L^{2}$ rather than $L^{1}$, the same conclusion holds with ``bounded
continuous'' weakened to $L^{2}$ and \eqref{eq:inv} read as an $L^{2}$ limit, valid almost
everywhere; the region integral is unchanged. The sample size enters only as the exponent in
$\psi^{n}$, so the inversion dimension and the work at a fixed grid do not grow with $n$; the
grid a prescribed accuracy demands may still depend on $n$, on the parent and on the region.
\end{theorem}

Four hypotheses bound this paper as much as the exact-evaluation paper. Absolute continuity excludes the discrete and censored parents of Open Problem~\ref{op:censor}. The moment condition makes $s$ a well-behaved coordinate. The integrability of $|\psi|^{n}$ licenses the inversion, and it weakens to $L^{2}$ at the small $n$ where the anchors of \S\ref{sec:anchor} live. Independence across units is relaxed only through the shared-latent construction of \S\ref{sec:reallots}, which enters through $\psi$. Nothing here verifies these hypotheses for a parent; every number
reported inherits them, and Appendix~\ref{app:engine} states how \eqref{eq:inv} is discretized
and validated. Read $g$ as a hypothetical: \emph{if the process truly looked like $F_\theta$,
what fraction of its lots would pass this test?} Three properties of $g$ shape what follows. It belongs to the \emph{pair} of process and rule, so enlarging the sample changes the question asked and not only the evidence answering it (\S\ref{sec:unify}). It is deterministic, so any spread in the reported probability can only come from uncertainty about $\theta$. And it is the operating characteristic read at a point, so the paper's question is where on that curve the estimated process sits.

\section{The pushforward}\label{sec:pushforward}

\subsection{The induced law and its width}\label{sec:induced}

Throughout this paper $\pi$ is a Bayesian posterior $\pi(\dd\theta\mid D)$, so the pushforward
$g_\#\pi$ defined below is a posterior law for the process's conditional operating characteristic.
The same operation applies to the sampling law of an estimator or to a measure on a confidence set.
Those produce different objects whose quantiles share no one inferential reading, so they are
reported in separate subsections and never compared endpoint to endpoint.

\begin{proposition}[Uncertainty in $\Pacc$ is uncertainty in $\theta$]\label{prop:push}
Let $\pi$ be a probability law on $\theta$. Since $g$ in \eqref{eq:g} is a deterministic
measurable map, the acceptance probability has an induced law: the \emph{pushforward} of $\pi$
through $g$, written $g_\#\pi$ and defined by
\begin{equation}\label{eq:pushdef}
(g_\#\pi)(B)\;=\;\pi\big(\{\theta:\ g(\theta)\in B\}\big)
\end{equation}
for every set $B$ of probability values.
\end{proposition}
\noindent In words, $g_\#\pi$ is the law of $g(\theta)$ when $\theta\sim\pi$. Its
$\beta$-quantile is
\begin{equation}\label{eq:pushforward}
q_\beta=\inf\big\{v:\ \pi\big(\{\theta:g(\theta)\le v\}\big)\ge\beta\big\},
\end{equation}
and the quantity reported throughout is the \emph{interval width}
\begin{equation}\label{eq:width}
W\;=\;q_{0.975}-q_{0.025},
\end{equation}
the spread of the equal-tailed central $95\%$ of $g_\#\pi$ under the law $\pi$ named in each
application, in probability points, so $0\le W\le1$. Because $W$ summarizes a central interval and
not the support, two readings must be avoided. First, $W=0$ does not mean the probability is
determined. It means only that the central interval is degenerate at the resolution reported, and up
to $5\%$ of the mass is free to sit elsewhere. Second, $W=1$ does not mean every value is admissible:
a law with all its mass at $0$ and $1$ has $W=1$ while excluding every interior value.
This second case is this paper's own mechanism at its limit: on $\langle905\rangle$ as written,
with the limit fixed and the test grown to $\ntest=100$, a ten-unit posterior straddles the step
to which the operating characteristic converges, the reported interval is all of $[0,1]$, and the
limiting law is two atoms (Appendix~\ref{app:blindspot}). Every unlabeled number quoted in these units below is a value of $W$.

\subsection{The width in closed form for a Normal parent}\label{sec:closed}

Where the spread is the only live coordinate, the posterior for $\sigma$ under the flat prior
on $\ln\sigma$ has closed-form quantiles, $\sigma_{p}=s\,c_{p}(n)$ with
$c_{p}(n)=\sqrt{(n-1)/\chi^{2}_{n-1,\,1-p}}$, and $g$ is monotone in $\sigma$. A monotone map
carries quantiles to quantiles, so the central interval on $\Pacc$ is exactly
\begin{equation}\label{eq:exactw}
W\;=\;\bigl|\,g(\sigma_{0.025})-g(\sigma_{0.975})\,\bigr|,
\qquad \sigma_{p}=s\,c_{p}(n).
\end{equation}
The two points sit a factor $c_{0.975}(n)/c_{0.025}(n)$ apart, a function of the sample size
alone ($2.65$ at $n=10$, $1.69$ at $n=30$); where they fall is set by the data through $s$ and by
the rule through $g$. Linearizing between them, the delta method \cite{cramer1946,serfling1980} in
$\ln\sigma$, names the structure (\S\ref{sec:unify}) but is not a formula for the number: at
$n=10$ the two quantiles sit a full unit apart in $\ln\sigma$, and a tangent is off by an order of
magnitude wherever the point estimate is saturated and the posterior is not. Every width reported
here comes from \eqref{eq:exactw} or the full pushforward. For a Normal parent \eqref{eq:exactw} is
a formula. Below, $F_\nu$ is the chi-square distribution function with $\nu$ degrees of freedom
and $\chi^2_{\nu,p}$ its $p$-quantile.

\begin{proposition}[Pure spread rule]\label{prop:chi}
Let the parent be Normal, the rule $k\,s_t\le L'$ on a test sample of $n_t$ units,
$\nu_t=n_t-1$, and the posterior $\sigma^2\mid s\sim s^2\nu_e/\chi^2_{\nu_e}$ from an estimation
sample of $m$ units, $\nu_e=m-1$. With $\hat g=g(\bar x,s)$ the plug-in value and
$\hat q=F_{\nu_t}^{-1}(\hat g)$, the central $95\%$ interval on $\Pacc$ is
$[g_{0.975},\,g_{0.025}]$ with
\begin{equation}\label{eq:chilaw}
g_p=F_{\nu_t}\!\Big(\frac{\chi^2_{\nu_e,\,1-p}}{\nu_e}\,\hat q\Big),
\qquad
W=F_{\nu_t}\!\Big(\frac{\chi^2_{\nu_e,\,0.975}}{\nu_e}\,\hat q\Big)
 -F_{\nu_t}\!\Big(\frac{\chi^2_{\nu_e,\,0.025}}{\nu_e}\,\hat q\Big).
\end{equation}
The rule enters only through $\hat g$; $k$ and $L'$ cancel.
\end{proposition}

\begin{proof}
Since $\nu_t s_t^2/\sigma^2\sim\chi^2_{\nu_t}$, $g(\sigma)=F_{\nu_t}(\nu_t L'^2/(k^2\sigma^2))$,
and $\hat g=g(s)$ defines $\hat q=\nu_t L'^2/(k^2s^2)$. At the posterior quantile
$\sigma_p=s\,c_p$, $c_p^2=\nu_e/\chi^2_{\nu_e,1-p}$, the argument of $F_{\nu_t}$ is
$\hat q\,\chi^2_{\nu_e,1-p}/\nu_e$. As $g$ decreases in $\sigma$, quantiles of $\sigma$ map to
quantiles of $g$.
\end{proof}

At $n_t=3$ the law collapses to a power law in the failure probability, elementary when the
estimation sample is also three, and three calibration units then cannot certify a release at
this rule whatever they report (Corollary~\ref{cor:n3}, Appendix~\ref{app:n3law}).

The identity itself is not new. It is the conditional in-control probability of an $S^2$ control
chart whose variance is estimated from Phase~I data, whose distribution over the Phase~I sample
is the object the estimated-parameters literature designs against
\cite{jensen2006,epprecht2015,faraz2015,does2020}, read here as a posterior. What is new is its
use as an interval on an acceptance probability indexed by the plug-in value alone, the wedge
extension below, and the asymptotics that follow.

The $\langle905\rangle$ wedge adds a mean channel. Under a Normal parent $\Xbar\perp s$, so
Stage~1 is exactly the one-dimensional integral
\begin{equation}\label{eq:g1d}
g(\mu,\sigma)=\int\varphi(x)\,F_{\nu_t}\!\Big(\frac{\nu_t\,[L-D(\mu+\sigma x/\sqrt{n_t})]_+^{2}}
{k^2\sigma^2}\Big)\,dx,\qquad D(x)=(|x-100|-1.5)_+,
\end{equation}
a Gaussian mixture of pure-spread terms with a random effective limit. This is the
practitioner's evaluator for a Normal parent: one-dimensional quadrature, machine precision,
and $1-g$ available without cancellation. Entered at $\hat q=F_{\nu_t}^{-1}(\hat g)$ with $\hat g$ the full wedge plug-in,
Proposition~\ref{prop:chi} reproduces the exact $\sigma$-only width on the ninety-eight lots of
\S\ref{sec:reallots} to a median of $10^{-6}$ and a maximum of $0.081$;
Appendix~\ref{app:meanchannel} separates the mean channel from the calibration term and checks
the evaluator against the FFT engine.

Table~\ref{tab:chilaw} is the proposition at the sample sizes practice meets, and it is read at
the end a release decision consumes: a lot reporting $0.99$ from ten units can pass as rarely as
$0.31$ on the evidence, and needs fifty units before that floor reaches $0.88$.

\begin{table}[htbp]\centering\small
\caption{Lower endpoint of the $95\%$ interval on $\Pacc$ for the \emph{pure-spread} rule of
Proposition~\ref{prop:chi}, test size $n_t=10$, against the plug-in value $\hat g$ and the
estimation sample size $m$. Only the lower endpoint is tabulated because it is the one a release
decision consumes (\S\ref{sec:endpoint}); the upper endpoint exceeds $0.93$ for $\hat g\ge0.8$ at
every $m$ shown, and at $\hat g=0.5$ it is $0.960$ for $m=10$ and $0.712$ for $m=100$
(Table~\ref{tab:chilawhi}). These entries are exact
for the pure-spread rule only. Read as an approximation to the $\langle905\rangle$ wedge they are
optimistic: on EJPB-S8-1-10 the table gives $0.406$ against a full-posterior lower endpoint of
$0.308$, and on BJBMS-L01 $0.615$ against $0.518$, both at $m=10$.}
\label{tab:chilaw}
\begin{tabular}{@{}lccccc@{}}
\toprule
$\hat g$ & $m=10$ & $m=20$ & $m=30$ & $m=50$ & $m=100$\\
\midrule
$0.5$ & $0.019$ & $0.083$ & $0.134$ & $0.199$ & $0.278$\\
$0.8$ & $0.068$ & $0.234$ & $0.339$ & $0.454$ & $0.569$\\
$0.9$ & $0.117$ & $0.351$ & $0.478$ & $0.604$ & $0.716$\\
$0.95$ & $0.172$ & $0.459$ & $0.595$ & $0.717$ & $0.815$\\
$0.99$ & $0.311$ & $0.662$ & $0.786$ & $0.876$ & $0.934$\\
$0.999$ & $0.502$ & $0.840$ & $0.920$ & $0.964$ & $0.986$\\
$0.9999$ & $0.659$ & $0.929$ & $0.972$ & $0.990$ & $0.997$\\
$0.99999$ & $0.775$ & $0.970$ & $0.990$ & $0.997$ & $0.999$\\
\bottomrule
\end{tabular}
\end{table}

Cornish--Fisher on both chi-squares gives the large-sample form,
\begin{equation}\label{eq:probitlaw}
g_p=\Phi\!\Big(z_{\hat g}+r\,z'+\tfrac{\sqrt2}{3}\,r\,(z_{\hat g}z'-r)\,\nu_t^{-1/2}+O(\nu^{-1})\Big),
\qquad r=\sqrt{\nu_t/\nu_e},\quad z'=z_{1-p},
\end{equation}
so that $W\to\Phi(z_{\hat g}+1.96r)-\Phi(z_{\hat g}-1.96r)$: a one-line master curve in which the
rule, its constants and the parent's scale have disappeared, and the sample sizes survive only
through $r$. The leading term is not accurate at the sample sizes practice uses; the $\nu^{-1/2}$ term is,
to within $0.014$ over $z_{\hat g}\in[-3,3]$ and $0.023$ on the ninety-eight lots
(Appendix~\ref{app:asym}).

Parent shape changes the scale of this picture but not its structure (\S\ref{sec:skewed}); the
numerical law above is the Normal case.

\subsection{Which endpoint a decision consumes}\label{sec:endpoint}

$W$ is a summary, and its two halves are not interchangeable. A laboratory defending a release
against the worst the evidence permits reads the \emph{lower} endpoint, equivalently
$1-q_{0.025}$, an upper bound on the probability that the lot fails the test. That reading is
the object Bergum's method returns as a lower bound at stated confidence
\cite{bergum1990,astm2709,astm2810}, and the quantity LeBlond and Mockus compare against a bound
\cite{leblond2014}. What this paper adds is that the endpoint and the width answer different
questions. Reporting one without the other hides the case where a saturated estimate
carries a long lower tail. A manufacturer asking whether a rule is more punishing than the
evidence warrants reads the \emph{upper} endpoint, and compares it with the normal-theory value
rather than with $W$. At the flagship operating point of \cite{paper1} the exact value is $0.464$
against a normal-theory $0.384$, eight probability points more punishing than the evidence
makes it. At a neighboring point the sign reverses ($+0.062$), which is why a decision
should consume an endpoint rather than a margin. $W$ is reported throughout for two reasons: it is
the one scalar comparable across lots, studies and sectors, and the three-regime structure
of \S\ref{sec:reallots} is a statement about $W$ that neither endpoint alone can carry. Where a
decision is at stake, the endpoint is what should be quoted, and \S\ref{sec:reallots} quotes it.

\subsection{What is being estimated, and what the grid does to it}\label{sec:estimand}

\begin{remark}[Two estimands, and which one a decision uses]\label{rem:estimands}
The pushforward $g_\#\pi$ and the number
\begin{equation}\label{eq:predictive}
\Pr(\text{pass}\mid D)\;=\;\int g(\theta)\,\pi(\dd\theta\mid D)\;=\;\E[g(\theta)\mid D]
\end{equation}
answer different questions and this paper reports both. Equation~\eqref{eq:predictive} is the
\emph{posterior predictive} probability that a future sample passes, a single number that
already averages over parameter uncertainty and is what enters an expected-loss calculation.
The law $g_\#\pi$ describes how well the process's \emph{conditional} operating characteristic
$g(\theta)$ is pinned down by the evidence, and its width tells an analyst whether a quoted
operating characteristic can be relied on at all. Neither is a probability that an already
observed sample passes; every statement here concerns a future sample from a process assumed
stationary at $\theta$.
\end{remark}

\noindent Each $g(\theta^{(j)})$ also carries the deterministic bias of the underlying inversion,
which does not average away over $j$ and need not be the same at the two endpoints;
Appendix~\ref{app:boundary} measures its effect on the width at three orders of magnitude below
the widths reported.

\subsection{Three treatments, and what each assumes}\label{sec:three}

None of the three is new as a construction.

\begin{enumerate}\itemsep3pt
\item \textbf{Global (pushforward quantiles).} Evaluate $g$ on a grid or posterior sample and
read \eqref{eq:pushforward}. It is exact up to the $\theta$-grid and captures the asymmetry that
is genuinely present because $g$ saturates. This is the recommended default, computing the object
of Proposition~\ref{prop:push} itself, and the two treatments below are evaluated against it
throughout (Appendix~\ref{app:global}). Its mean is the quantity known in clinical trial design
as \emph{assurance} \cite{ohagan2005}, and the Expostats toolkit performs the corresponding
calculation on exposure data \cite{lavoue2019}.
\item \textbf{Local (delta method).} With $g$ smooth at $\hat\theta$, approximate
$\Var g\approx\nabla g^\top\Sigma_\theta\nabla g$, with the gradient taken by finite differences
of the evaluator, and report $\hat g\pm1.96\,\mathrm{sd}$. It is cheap and symmetric, but it
degrades exactly where release decisions live, near the saturation points $g=0$ and $g=1$, where a
symmetric band is not well posed
(Appendix~\ref{app:delta}). It is the textbook delta method \cite{cramer1946,serfling1980}, the
first-order propagation the GUM prescribes as its law of propagation of uncertainty
\cite{gum2008}, and so the default error bar for any quantity computed from estimated inputs.
\item \textbf{Worst-case (corner) band.} Over a confidence set $\Theta_\alpha$ report $[\min g,\max g]$;
if $g$ is monotone in each coordinate the extremes sit at corners. This is Bergum's method
\cite{bergum1990}, standardized as ASTM E2709 and E2810 \cite{astm2709,astm2810,bergum2015} and
implemented in the CuDAL tooling, which reports the minimum as a lower bound at the stated
confidence. Section~\ref{sec:corner} establishes when this is safe (Appendix~\ref{app:corner}).
\end{enumerate}

\paragraph{The closest precedent.} Treatment 1 has been carried out on this very rule, and the
present paper sits downstream of it. LeBlond and Mockus \cite{leblond2014} place a one-way normal random-effects model on
$\langle905\rangle$, evaluate the acceptance probability at each posterior draw, and compare a
lower percentile of the induced posterior against a bound, in substance the reading of
\S\ref{sec:endpoint}; they also find that the operating characteristics depend on the parameters
almost only through the acceptance probability itself. Lewis and Fan \cite{lewis2016} construct the
same posterior for a single lot under a Normal parent and use its lower quantile to widen the
limits of ASTM E2810, setting a Bayesian posterior against Bergum's confidence-set construction,
the comparison \S\ref{sec:corner} returns to. \emph{The induced law, its percentiles, its use in a
decision, and the observation that position on the rule is what the parameters act through are
therefore not new.} What follows differs in reach. The precedent is Normal throughout, where $\Xbar\perp s$; the
evaluator of \cite{paper1} carries the same statement to non-Normal parents. It concerns a single
compendial test on twelve products; here one evaluator carries a wedge and a half-plane across
ninety-eight published lots and $1207$ exposure assessments. And it asks what a particular batch
is worth, whereas the question here is what governs the \emph{width} across decisions: among the
forty-eight lots assayed at ten units, which carry identical posterior precision on the spread,
the induced widths run from essentially zero to $0.962$, and across all ninety-eight, width
correlates with distance from the contested middle at $-0.75$ and with sample size at only
$-0.15$.

Two recent preprints reach a neighboring conclusion from the frequentist side: a capability
gate on a plug-in $C_{pk}$ accepts with probability tending to one half when the true capability
sits at the threshold \cite{jiang2026a}, and uncertainty that looks stable in index space is
amplified in defect-probability space \cite{jiang2026b}. Their census of 500 manufacturing
dimensions, 151 within $0.10$ of the $1.33$ approval threshold, documents how common contested
operating points are. Their gate is a wedge of the family treated here, its probability
estimated by Monte Carlo and propagated by a two-point delta method. The present treatment
evaluates the inner functional exactly, carries the uncertainty as a posterior pushforward, gives
an explicit width representation whose transfer across geometries is tested, and forms a
posterior for a non-Normal parent from published summaries alone (\S\ref{sec:skewed}).

\section{Three studies, from anchor to archive}\label{sec:ladder}

\subsection{An analytic anchor}\label{sec:anchor}

For a Uniform parent, $n=3$, and the wedge $k=1$, $L=\tfrac12$, the acceptance probability is
the elementary number $6+18\ln\tfrac34=0.821723$ \cite{paper1,n3companion}. The evaluator returns
$0.821489$ at $N=4096$, a grid error of $2.3\times10^{-4}$ consistent with the inversion's
calibrated bias. The pushforward machinery is exercised on top of a $g$ whose value is known in
closed form, so any discrepancy is attributable.

In this study the smearing can be written down rather than measured. For a Uniform parent at
$n=3$ both $\Xbar$ and $s$ scale with the support width $w=\sigma\sqrt{12}$, so the acceptance
probability depends on the single ratio $\lambda=L/w$: $g(\sigma)=G(\lambda)$ for a scalar
function $G$. Uncertainty in $\sigma$ is then a pure change of variables, and the interval is governed by the
\emph{amplification} $|dG/d\ln\sigma|$. This is the elementary derivative that converts a proportional
error in $\sigma$ into probability points on $\Pacc$: a one per cent error in the spread
moves the acceptance probability by $0.01\,|dG/d\ln\sigma|$. The exact-evaluation paper
\cite{paper1} obtains $G$ in closed form on $\lambda\le\sqrt3/4$, where
$G=\tfrac{4\sqrt3\pi}{3}\lambda^{3}$, and piecewise-elementary beyond it; the anchor above is the
value at $\lambda=\tfrac12$. In the original units the cubic range reads
$g(\sigma)=\pi L^{3}/(18\sigma^{3})$ \emph{exactly} --- an inverse-cube law in the standard
deviation, \eqref{eq:gsigma}. Appendix~\ref{app:anchorscale} carries the whole reduction out
step by step. On the cubic range the amplification is then \emph{exactly}
\begin{equation}\label{eq:amp3G}
\Big|\frac{dG}{d\ln\sigma}\Big| = 3\,G ,
\end{equation}
which rises to $1.767$ at the end of the cubic range, $\lambda=\sqrt3/4$. Past that join, on the elementary
piece, the amplification peaks at $1.803$ at $\lambda=0.442$, where the elasticity $\lambda G'/G$ has fallen from $3$ to $2.88$,
and then falls to zero as the rule saturates (Figure~\ref{fig:anchor}(c)). This is the mechanism behind
\S\ref{sec:reallots}. An interval on the acceptance probability has width of order
$|dG/d\ln\sigma|\times\mathrm{sd}(\ln\sigma)$, so it collapses wherever the rule saturates, for lots that clearly pass and for lots that clearly fail, and opens only in between.

\subsection{The \texorpdfstring{$\langle905\rangle$}{<905>} rule, on one operating point
and then on ninety-eight real lots}\label{sec:905}

\subsubsection{A single operating point, in detail}

We take the Stage-2 constants of $\langle905\rangle$ ($n=30$, $k=2.0$, the clipped acceptance value) with a Normal parent at $\hat\mu=100$, $\hat\sigma=5$, and $L$ calibrated so that the
point estimate is exactly $g(\hat\theta)=0.850$. Calibrating the limit rather than fixing it is
essential: without that control a comparison reports which rule is easier to pass rather than
how each responds to uncertainty. The rule is the acceptance-value condition alone; Stage~2's
individual-unit limits, no unit outside $0.75M$ to $1.25M$, are omitted. At this calibrated limit
they are practically implied by the acceptance value: passing forces $s\le L/k=5.66$, and one unit
$25$ points from target already forces $s\ge4.6$, leaving the other twenty-nine almost no room. By
quadrature, the probability that a lot passes the acceptance value yet has a unit outside those
limits is $4\times10^{-9}$ averaged over the posterior below and at most $2\times10^{-8}$ at any
draw, far below the endpoints' standard error.

One more ingredient enters here, and only here. Tablets from one lot are not independent draws
from the process: they come from the same blend and so share that lot's offset from the process
mean, which makes them positively correlated, the concern of the sampling-plan literature around
ASTM E2709 \cite{bergum2015}. Two things follow. Averaging the assayed units does not remove the
offset, so the process mean is less well known than $s^2/m$ suggests; and a new lot draws its own
offset, so its mean differs from this lot's. We carry this through a pairwise correlation $\rho$
between units of a lot, given a range rather than a value because one lot cannot estimate it. The
ninety-eight lots below are evaluated with independent units, and \S\ref{sec:decomp} measures how
much $\rho$ matters.

In symbols, the shared offset is a lot-level factor,
$X_i=\mu+\sqrt{\rho}\,\sigma Z+\sqrt{1-\rho}\,\sigma\varepsilon_i$ with $Z$ common to the lot, and the
parameter law is a posterior under a stated prior, not a law stipulated alongside the model.
Compound symmetry keeps $\Xbar$ and $S^2$ independent with exact laws
(Appendix~\ref{app:cslaw}): writing $v=\sigma^2(1-\rho)$ for the within-lot variance,
$\Xbar\sim N(\mu,\,v[1/m+\rho/(1-\rho)])$ and $(m-1)S^2/v\sim\chi^2_{m-1}$. Under the
reference prior $\pi(\mu,v,\rho)\propto v^{-1}\pi_\rho(\rho)$ with $\rho\sim\mathrm{U}[0,0.3]$,
the posterior is
\begin{equation}\label{eq:cspost}
v\mid D=\frac{(m-1)s^2}{U},\quad U\sim\chi^2_{m-1};\qquad
\mu\mid v,\rho,D\sim N\!\Big(\bar x,\;v\Big[\frac1m+\frac{\rho}{1-\rho}\Big]\Big);\qquad
\rho\mid D\sim\pi_\rho ,
\end{equation}
the last because one lot does not identify $\rho$ at all; the data identify only the within-lot
variance $v=\sigma^2(1-\rho)$. The forecast is for a \emph{new lot} from
the same process, the question a release rule puts to a process; a forecast for further units of
the lot already sampled is reported once below for contrast. The evaluator is the
one-dimensional integral of \S\ref{sec:closed} extended to the shared factor, exact to machine
precision (Appendix~\ref{app:cslaw}).

\begin{table}[h]\centering\small
\caption{Three treatments at the $\langle905\rangle$ Stage-2 operating point, all from the
posterior \eqref{eq:cspost} and the same point estimate; endpoints carry a Sobol' standard error
below $0.0002$. The same-lot row is the contrasting forecast target. The delta band's upper
endpoint is not a probability, and the corner band misses the true maximum, which is attained at
the target mean in the interior of the box.}
\label{tab:905}
\begin{tabular}{@{}lcc@{}}
\toprule
treatment & $95\%$ band & width\\
\midrule
point estimate, no uncertainty propagated & $0.850$ & n/a\\
pushforward, new lot, mean $0.533$ & $[0.035,\ 0.983]$ & $0.948$\\
pushforward, same lot, mean $0.714$ & $[0.094,\ 0.999]$ & $0.905$\\
delta-method band, full covariance & $[0.245,\ 1.455]$ & $1.210$\\
corner band on the $95\%$ box & $[0.018,\ 0.992]$ & $0.973$\\
true range of $g$ over that box & $[0.018,\ 0.999]$ & $0.981$\\
\bottomrule
\end{tabular}
\end{table}

\begin{figure}[ht]\centering
\includegraphics[width=0.92\linewidth]{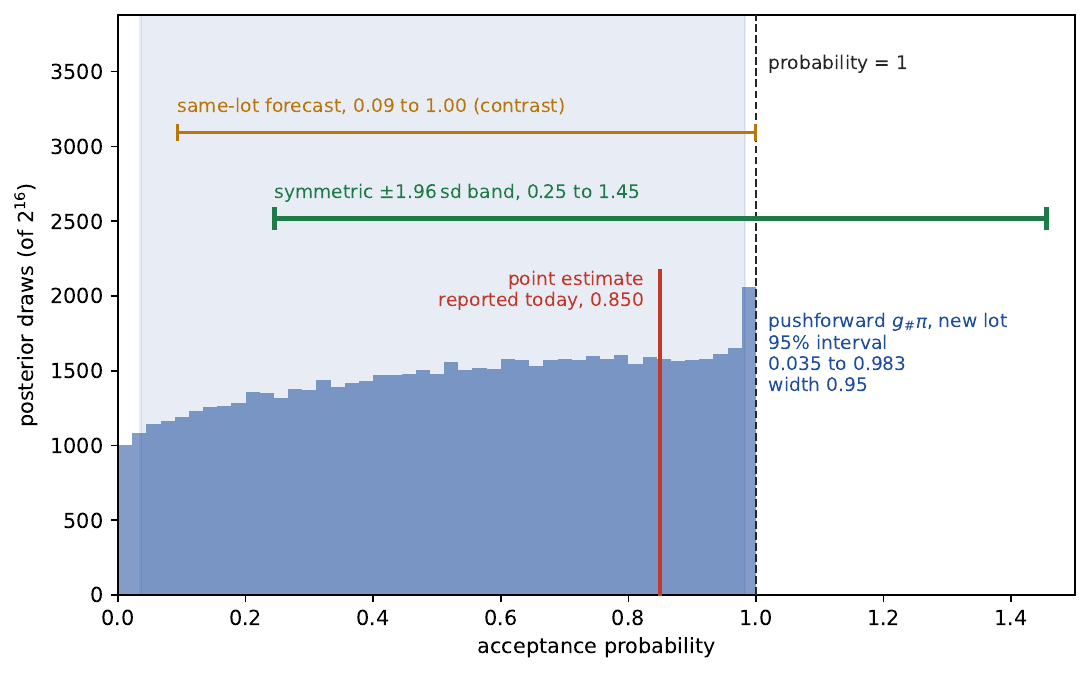}
\caption{The law of the acceptance probability at the $\langle905\rangle$ Stage-2 operating
point, from $2^{16}$ scrambled Sobol' points of the posterior \eqref{eq:cspost}, each evaluated
exactly. The rule is designed so the point estimate is $0.850$ (red); the pushforward is wide,
strongly asymmetric, and piles up against the ceiling. The symmetric $\pm1.96\,$sd band (green),
with the full posterior covariance and a converged gradient, cannot describe such a law: its
upper endpoint is $1.45$, past the dashed line at probability one, and its lower endpoint
$0.25$ sits $0.21$ above the exact $0.035$.}\label{fig:push}
\end{figure}

\begin{remark}[Which digits are earned]\label{rem:conv}
The outer integral over $\theta$ sets the earned digits. The Normal studies integrate it by
replicated scrambled Sobol' sets, four of $2^{16}$ points at this operating point, so every
endpoint carries a standard error, $6\times10^{-5}$ on the lower and $1.3\times10^{-4}$ on the
upper endpoint here. A tail quantile is a tail functional of the evaluated atoms, so a scheme that
settles the mean need not settle the endpoints (Appendix~\ref{app:global}).
\end{remark}

\subsubsection{Ninety-eight published lots}\label{sec:reallots}

The operating point above was chosen by calibration (design~A of Table~\ref{tab:designs}). That makes the
comparison controlled, at a price: the answer is assumed to be interesting rather than found to be.
We therefore repeat the construction on real lots. Each lot's own $(n,\bar x,s)$ serves as the
parameter estimate, and its own posterior is pushed through the same evaluator, exactly as in
\S\ref{sec:cehd}. Each posterior is integrated by two independent scrambled Sobol' sets of $2^{14}$ points
through the one-dimensional evaluator of \S\ref{sec:closed}, so every endpoint carries a
standard error, at most $0.001$ at Stage~1; endpoints are reported to two or three decimals and
counts of lots crossing a threshold remain approximate where a lot sits near it.

\paragraph{Data.} Ninety-eight content-uniformity lots from four published sources, all public,
a curated convenience sample rather than the product of a systematic search:
ten commercial-scale batches at $n=50$--$55$ \cite{goodwin2018}; thirty routine lots at $n=10$
\cite{vranic2008}; eighteen process-development runs at $n=10$, with individual tablet values
published \cite{lyytikainen2024}; and forty batch-versus-continuous formulations at
$n=21$ \cite{ijpx2024}, three tablets from each of seven time-indexed bags treated here as
twenty-one exchangeable units, with the source's standard-deviation convention not stated (a
$\sqrt{21/20}$ rescaling moves the largest width below by $0.03$). A source-to-row inventory
ships with the code. Reported dispersions span $0.69\%$ to $19.5\%$ RSD. Three features of the sample matter. The acceptance
value is published for two sources and computed by us for the other two, so only the former
corroborate our $\langle905\rangle$ implementation. All four sources report summary statistics, so the
parent's shape is not observed; Remark~\ref{rem:scaleinv} shows that this does not matter for the
width. And the two development sources publish failures because failure was their finding,
while the routine sources do not, so the mix is not a failure rate for anything.

\paragraph{What is evaluated, precisely.} Two sample sizes are kept apart. The source lot's
$\mest$ sets the posterior, $\sigma^2\mid\text{data}\sim\text{Inv-}\chi^2(\mest-1,s^2)$ and
$\mu\mid\sigma\sim N(\bar x,\sigma^2/\mest)$; $\ntest$ sets the future test whose probability $g$
reports, $\ntest=10$ with $k=2.4$ for Stage~1. The forty-eight lots published at $\mest=10$ carry a
posterior of the same size as the test, the other fifty, at $\mest=21$ to $55$, a sharper one. At
Stage~1 the acceptance-value condition is the whole rule, so the lots are evaluated against
Stage~1 as written, as a standalone test rather than the first step of the coupled sequence
(Appendix~\ref{app:twostage}).

\begin{table}[h]\centering\small
\caption{Interval width against the plug-in acceptance probability computed here, $98$ published
lots evaluated under the acceptance-value condition alone. Each lot's posterior is formed from
that source's own sample size $\mest$, and the rule is then evaluated at the size the standard
prescribes, $\ntest=10$ with $k=2.4$ (see ``What is evaluated''). The
median collapses at both ends and opens in the middle; the last column shows that the collapse
is a statement about the median and not about every lot.}
\label{tab:reallots}
\begin{tabular}{@{}lrccr@{}}
\toprule
plug-in acceptance probability & lots & median width & largest width & lots with width $>0.10$\\
\midrule
$>0.99$ \quad (clearly passing) & $76$ & $<0.001$ & $0.692$ & $8$\\
$0.80$--$0.99$ & $7$ & $0.725$ & $0.934$ & $7$\\
$0.20$--$0.80$ \quad (contested) & $5$ & $0.812$ & $0.962$ & $5$\\
$<0.20$ \quad (clearly failing) & $10$ & $0.032$ & $0.404$ & $2$\\
\bottomrule
\end{tabular}
\end{table}

Table~\ref{tab:reallots} is the study's result, and it is sharper than the calibrated case
permits. \emph{The point estimate is typically dependable where the decision is not in doubt,
and uninformative where it is}, but the first half of that sentence is about the median and not
about every lot. Among the seventy-six lots with $\hat g>0.99$ the median width is below the
numerical resolution of the table, yet ten exceed $0.05$, eight exceed $0.10$ and one exceeds
$0.50$, the largest being $0.692$ at a plug-in value of $0.997$; among the ten with
$\hat g<0.20$, three exceed $0.05$, two exceed $0.10$ and the largest reaches $0.404$, a lot
reporting $0.107$ against $[0.008,\,0.412]$. A plug-in
value close to a saturation boundary is therefore not sufficient for a narrow interval: the
posterior can carry the process back onto the steep part of the operating characteristic, and
whether it does is a question about the posterior rather than the point estimate.

For lots that look safe, the width is the wrong thing to report, and
Figure~\ref{fig:lotsend}(a) shows what to report instead. When the plug-in is close to one, the
upper end of the interval is stuck at one too, so everything the interval says is in its lower end.
Turned around, that lower end is a worst plausible failure rate: $1-q_{0.025}$, the probability of
failing the test that the evidence still allows at $97.5\%$ credibility. Read that way, the plug-in
can badly understate the risk. Of the seventy-six lots whose plug-in says they fail less than
$1\%$ of the time, fifteen could plausibly fail $1\%$ of the time or more, ten $5\%$ or more and
six $25\%$ or more. The most extreme says it fails $0.33\%$ of the time while the evidence allows
$69\%$. The plug-in alone cannot reveal this; the pushforward reports it at no extra cost. For the twelve lots in between, where
a judgment must actually be made, the median interval is $0.77$ wide: one lot reports $0.886$
against $[0.066,\,1.000]$, another $0.584$ against $[0.015,\,0.978]$.

\begin{figure}[H]\centering
\includegraphics[width=\linewidth]{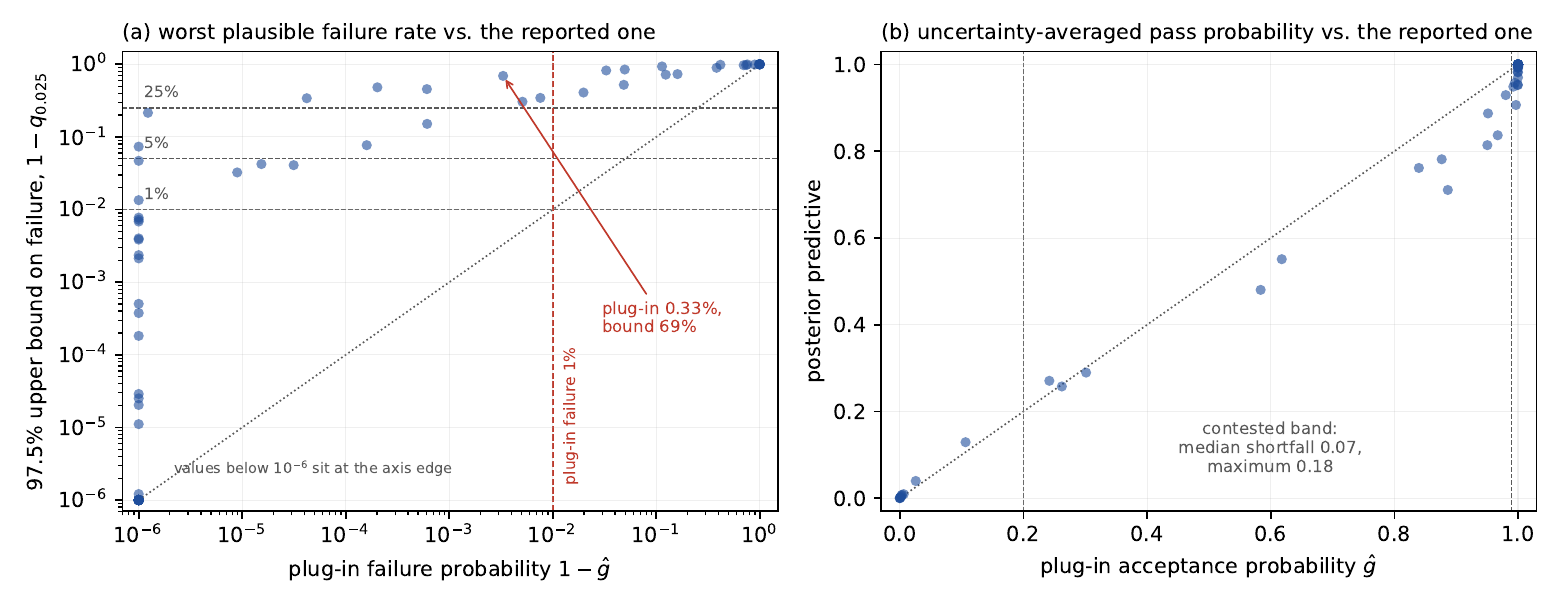}
\caption{What a single reported number hides, lot by lot; each dot is one of the ninety-eight
lots. \emph{(a)} Worst case. Horizontal: the probability of failing Stage~1 that the plug-in
reports. Vertical: the highest failure probability the evidence still allows ($97.5\%$ credible
bound, $1-q_{0.025}$). Both axes are logarithmic, and lots on the dotted diagonal are described
honestly by their plug-in. Left of the red line are the seventy-six lots that appear to fail less
than $1\%$ of the time; fifteen of them sit above the $1\%$ line, ten above $5\%$ and six above
$25\%$. Lots whose plug-in failure is below $10^{-6}$ are stacked at the left edge. \emph{(b)} Best
single number. Horizontal: the plug-in acceptance probability. Vertical: the posterior predictive
(Remark~\ref{rem:estimands}), the chance the next test passes once parameter uncertainty is
averaged in. Clear passes and clear failures sit on the diagonal; in the contested band between the
dashed lines the predictive falls below the plug-in, by a median of $0.07$ and up to $0.18$,
because averaging over plausible parameters pulls in values from the steep side of the acceptance
curve.}
\label{fig:lotsend}
\end{figure}

Panel (b) shows that the posterior predictive, the number an expected-loss calculation would
use, sits below the plug-in exactly in the contested band, by a gap comparable to the approximation
error the exact-evaluation paper removes. That the contested band holds only $12\%$ of these lots
does not weaken the result: manufacturing is not run at the boundary, but the boundary is where
escalation, rejection and investigation decisions are taken.

By source (Appendix~\ref{app:bysource}), the ten commercial batches at $\mest=50$--$55$ are
determined at every lot and the wide cases all sit at $\mest=10$ or $21$; but $\mest$ orders widths
only when position is held roughly fixed, and across four sources it is not.

\paragraph{Escalating the test size does not resolve the contested lots.} Pushing the same
posteriors through the Stage-2 wedge at $\ntest=30$, $k=2.0$ carries seven further lots into the
clearly-passing band, while the median width among the five still contested rises to $0.986$
from $0.769$: escalation changes which lots sit near a boundary, not whether some do
(Appendix~\ref{app:escalate}).

\subsubsection{A skewed parent, with its own likelihood}\label{sec:skewed}

Everything above is under a Normal parent, where the posterior is classical and the evaluator
is a one-dimensional integral; the exact-evaluation paper's point is that the parent is not Normal and
that the level of $g$ moves by a tenth when it is not. Whether the \emph{interval} moves is a
different question, and answering it needs a posterior for a non-Normal parent formed from the
same published summaries, not a box of Normal uncertainty carried over. The evaluator of
\cite{paper1} supplies one. For a location--scale family $X=\mu+\sigma Z$ with $Z$ standardized
and its shape held fixed, the joint density $f_n$ of the standardized summaries $(\bar Z,S_Z)$
follows from the engine's density of $(\sum Z_i,\sum Z_i^2)$ by a change of variables, and the
likelihood of the published pair $(\bar x,s)$ from $m$ units is
\begin{equation}\label{eq:sumlik}
\mathcal{L}(\mu,\sigma)=\sigma^{-2}\,f_m\!\Big(\frac{\bar x-\mu}{\sigma},\ \frac{s}{\sigma}\Big),
\end{equation}
so that under the reference prior $\pi(\mu,\sigma)\propto1/\sigma$, the prior that returns the
inverse-chi-square posterior when $Z$ is Normal, the posterior on $(\mu,\sigma)$ is available on
a grid from the summaries alone. The acceptance probability
$g(\mu,\sigma)=\Prob[(\mu+\sigma\bar Z,\,\sigma S_Z)\in A]$ is the standardized density for the
test size integrated over the transformed region, so one inversion per shape and sample size
serves every $(\mu,\sigma)$: the same object is the likelihood and the operating characteristic.
Appendix~\ref{app:nnlik} records the construction and its checks. Two shapes are carried, with $\mu$ and $\sigma$ learned and no correlation. The first is
calibrated from the three hundred published tablets, in thirty batches, that carry per-unit
values \cite{lyytikainen2024,dyrby2002}: a common skew-Normal shape of skewness $0.74$, $95\%$
interval $[0.38,\,0.99]$. Normality is rejected at $p=0.0002$, but so is a single common shape,
the skew sitting in the failing, segregating runs, and the estimate moves with the inclusion rule
(Appendix~\ref{app:shape}). The shapes are therefore declared scenarios, the second a stress
case, a standardized Beta$(1.5,8)$ of skewness $1.06$, the shape a disturbed process may have.

\begin{figure}[tbp]\centering
\includegraphics[width=0.97\linewidth]{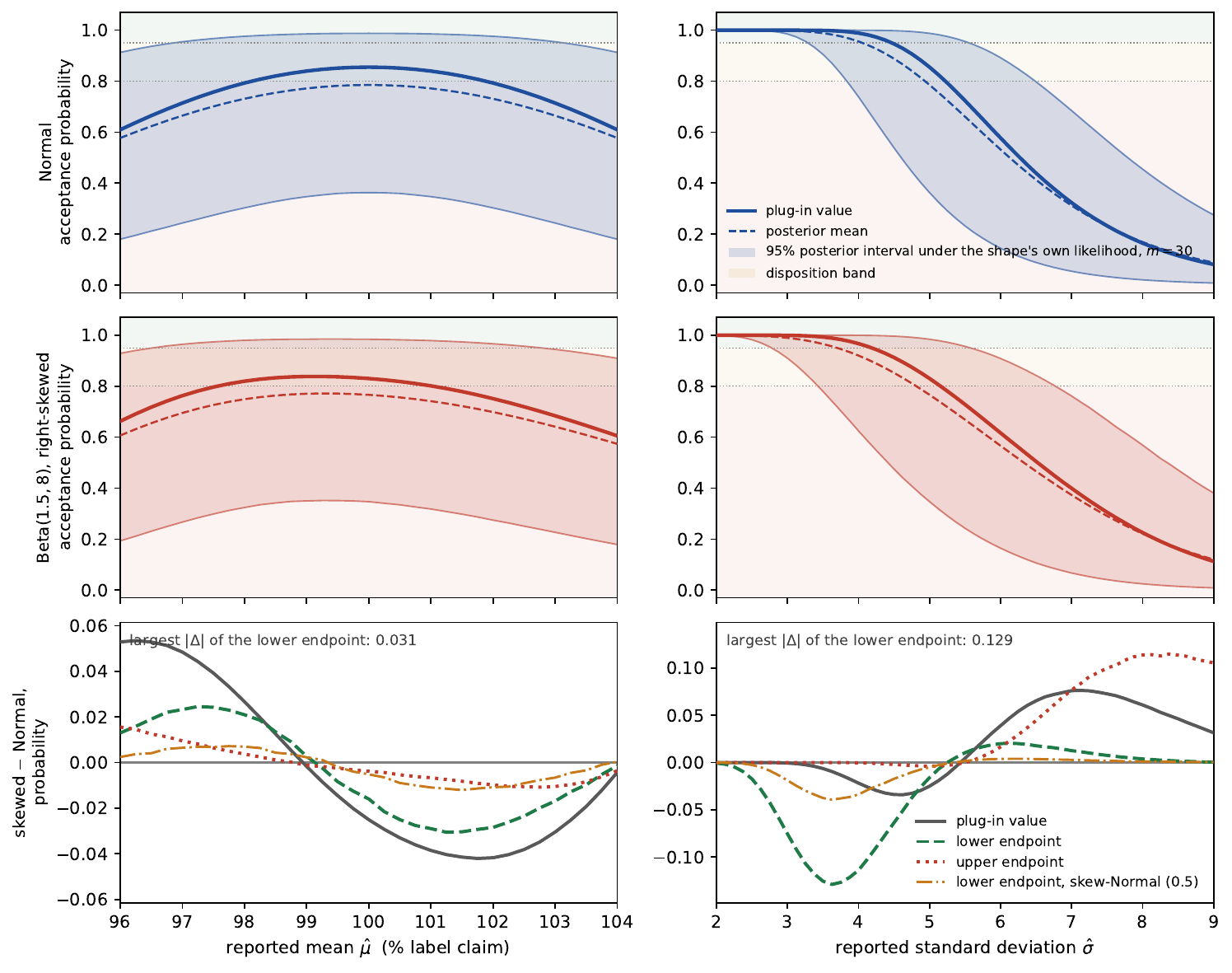}
\caption{Operating characteristics of $\langle905\rangle$ Stage 1 ($n=10$, $k=2.4$, $L=15$,
clipped acceptance value) with posterior intervals. Curve: the plug-in value; dashed: the
posterior mean; tint: the $95\%$ posterior interval under the parent's own likelihood
\eqref{eq:sumlik} for a thirty-unit calibration sample. Left column sweeps the reported mean
through target, right column the reported standard deviation; top row Normal, middle row the
standardized right-skewed Beta$(1.5,8)$, bottom row the difference, skewed minus Normal, of the
plug-in and the two endpoints, with the lower endpoint of a skew-Normal parent of skewness
$0.5$ for scale. Both the lower endpoint and the width move with the parent, by up to $0.13$
each on the spread sweep, where the lot appears to pass comfortably; on the location sweep the
width moves by at most $0.023$.}
\label{fig:oc}
\end{figure}

Figure~\ref{fig:oc} sweeps the Stage-1 rule across its operating characteristic with $m=30$.
The summary width statistics are nearly shape-free, $0.62$ and $0.64$ at target for the Normal
and the stress case, so the three-regime pattern of Table~\ref{tab:reallots} survives the parent.
Compared at the same reported summaries, though, the width moves by up to $0.129$ on the spread
sweep under the stress case and $0.069$ at the calibrated shape, both at $\hat\sigma\approx3.6$,
against at most $0.023$ on the location sweep, and the lower endpoint falls by the same amounts.
The mechanism is the likelihood: a right-skewed sample spread has a heavier right tail, so the
same reported $s$ admits a larger true $\sigma$, the posterior $97.5\%$ quantile at $m=30$ moving
from $6.72$ to $7.22$ under the stress case.

On the ninety-eight lots, at the calibrated shape the $2.5\%$ endpoint sits below the Normal one
in eighty-three lots, by up to $0.130$, with a median move of $5\times10^{-5}$. Under the stress
case, Table~\ref{tab:beta}, it falls by more than $0.10$ in sixteen lots and by up to $0.22$, at a
lot reporting $100.8\pm2.3$ from ten units whose Normal lower endpoint of $0.95$ becomes $0.73$
while its interval widens from $0.047$ to $0.267$; across the lots the paired width increase
reaches $0.221$, with a minimum of $-0.120$, so the width is not shape-free and its move is not
one-signed. The effect is concentrated where the Normal plug-in sits against the ceiling at a
tight reported spread, on lots that appear to pass comfortably. On a disturbed process this
interval effect and the exact-evaluation paper's level error are of the same order ($0.22$
against $0.10$ for the plug-in) and coincide in location; on a routine process both are small. A
margin could repair much of the mostly one-signed endpoint shift, but not the widening, nor the
fact that ten to thirty units do not identify the shape.

\subsection{The exposure rule at the sample sizes of \texorpdfstring{$1207$}{1207} real
occupational assessments}\label{sec:cehd}

The exposure study leaves pharmaceutical release entirely, which is the point: the phenomenon
is a property of estimating a parent, not of one standard. We use OSHA's Chemical Exposure
Health Data \cite{cehd}, a public-domain archive of over one million personal air-sampling
results, and the EN\,689 Annex F / AIHA compliance rule \cite{en689,aiha}, the half-plane of
Figure~\ref{fig:geom}(b) (in standardized units, Figure~\ref{fig:geomstd} of
Appendix~\ref{app:geomstd}):
\begin{equation}\label{eq:en689}
\ybar+K(n)\,s_y\ \le\ \ln(\mathrm{OEL}),\qquad y=\ln(\text{exposure}),
\end{equation}
with $K(n)$ the one-sided tolerance factor for $95\%$ coverage at $70\%$ confidence, computed
exactly as $K(n)=t'_{n-1,0.70}(z_{0.95}\sqrt n)/\sqrt n$.

\paragraph{Extraction.} From the survey years $2015$, $2017$, $2018$ and $2019$ we retain
personal, non-blank samples with positive numeric results, discard the analytical artifacts
recorded in the substance field, group by inspection $\times$ substance $\times$ unit, and keep
groups with $6\le n\le60$ and at most $20\%$ non-detects. Left-censoring is real in this
archive, $60\%$ of personal samples overall are non-detects, and is deliberately outside our
scope. The filter yields $1207$ exposure sets, median $n=8$, median GSD $2.20$. Because the archive
is updated twice a year, the extraction is pinned: the four per-year files as retrieved on
24~September~2026 reproduce the $1207$ sets and every number derived from them (sample sizes and
geometric standard deviations exactly, interval endpoints to floating-point rounding), and their
SHA-256 digests are recorded with the reference implementation.

\paragraph{Method.} For a lognormal parent the acceptance probability of \eqref{eq:en689} is a
one-dimensional integral,
\begin{equation}\label{eq:ghalf}
g(\mu,\sigma)=\E_V\,\Phi\!\Big(\sqrt n\big(\delta-K\sqrt{V/(n-1)}\big)\Big),
\qquad \delta=\frac{c-\mu}{\sigma},\quad V\sim\chi^2_{n-1},
\end{equation}
which is the noncentral-$t$ distribution function, $g(\delta)=T_{n-1}(-K\sqrt n;\,-\delta\sqrt n)$,
validated against quadrature, Monte Carlo at twelve $(n,\delta)$ points and direct simulation
of the rule (Appendix~\ref{app:engine}). For the parameters we use the standard noninformative
posterior, $\sigma^2\mid\text{data}\sim\text{Inv-}\chi^2(n-1,s_y^2)$ and
$\mu\mid\sigma\sim N(\ybar,\sigma^2/n)$. Under it the standardized margin $\delta$ itself has a closed posterior,
$\Prob(\delta\le d\mid\text{data})=1-T_{n-1}(\hat\delta\sqrt n;\,d\sqrt n)$ with
$\hat\delta=(c-\ybar)/s_y$. Because $g$ increases with $\delta$, the ends of the posterior interval for $\delta$ map directly to the ends of the interval for $g$. The endpoints are therefore
$g$ at the two posterior quantiles of $\delta$, found by root-finding with no outer sampling at
all. As in \S\ref{sec:905}, $c$ is calibrated \emph{per
set} so that every one of the $1207$ has point estimate $0.850$; only $n$ and the posterior
differ.

\begin{table}[h]\centering\small
\caption{Pushforward intervals on the compliance probability at the sample sizes of $1207$ real
occupational exposure assessments, each calibrated to a point estimate of $0.850$ so that the
limit, not the observed exposure, is what varies. The calibrated problem depends on $n$ alone,
so each of the $31$ distinct sample sizes is solved once, exactly, by the noncentral-$t$
root-finding of the Method paragraph, and the counts enter as weights; the endpoints carry no
outer sampling error, and the earlier Sobol' evaluation of the same table was high by up to
$2.7\times10^{-4}$ at the upper endpoint through its inner quadrature. The delta-method band returns an upper endpoint exceeding one in
every case.}
\label{tab:cehd}
\begin{tabular}{@{}rrccc@{}}
\toprule
$n$ & sets & median $95\%$ interval & median width & delta band $>1$\\
\midrule
$6$--$7$ & $567$ & $[0.074,\ 0.999]$ & $0.925$ & $100\%$\\
$8$--$9$ & $295$ & $[0.085,\ 0.999]$ & $0.915$ & $100\%$\\
$10$--$14$ & $246$ & $[0.096,\ 0.999]$ & $0.903$ & $100\%$\\
$15$--$24$ & $82$ & $[0.111,\ 0.999]$ & $0.888$ & $100\%$\\
$25$--$60$ & $17$ & $[0.125,\ 0.999]$ & $0.874$ & $100\%$\\
\midrule
all & $1207$ & $[0.085,\ 0.999]$ & $0.915$ & $100\%$\\
\bottomrule
\end{tabular}
\end{table}

Table~\ref{tab:cehd} is the paper's cleanest evidence, because the calibration strips out
everything except sample size. A compliance probability \emph{held} at $0.85$ carries a $95\%$
posterior interval of $[0.08,\,1.00]$: at the contested value the point estimate conveys almost
nothing. A median assessment estimates $\sigma$ from eight observations, so at $n=8$ the true
value is credibly anywhere in $[0.66\,\hat\sigma,\ 2.04\,\hat\sigma]$, a threefold band, against
$[0.80,\ 1.34]$ at $n=30$. Acceptance depends on the margin in units of $\sigma$, so that band
drags the calibrated $\hat\delta=2.67$ across roughly $[1.31,\,4.04]$. Over that range $g$ runs
from about $0.15$ to $0.999$.

\begin{remark}[What the data does, and does not, contribute]\label{rem:scaleinv}
Because $c$ is calibrated per set, the computation is scale-free in $s_y$: $\hat\delta$ solves
$g(\delta,n,K)=0.85$, which involves $n$ alone. \emph{The widths in Table~\ref{tab:cehd} are
therefore functions of $n$ and not of the observed dispersion}, an algebraic identity rather than
an empirical finding. What the archive supplies is the \emph{empirical distribution of sample
sizes in regulatory practice}: median $8$, with $567$ of $1207$ assessments at $n\le7$; that is
why the table is computed once per distinct $n$. Nor does the archive establish the modeling
assumptions: a median GSD of $2.20$ describes dispersion, not lognormality, and no claim is made
that the groups are homogeneous similar-exposure groups or comparable with exposure limits. This is
a controlled study weighted by real sample sizes, not an analysis of $1207$ compliance decisions
actually taken (Open Problem~\ref{op:limits}).
\end{remark}

Two further readings deserve emphasis. \emph{The delta band fails universally, and
structurally.} Its upper endpoint exceeds one whenever $1.96\,\mathrm{sd}>1-0.85=0.15$, and the
standard deviation implied by the observed widths is about $0.234$, so its failure on $1207$ of
$1207$ assessments is a deterministic consequence of the ceiling rather than an empirical
tally. \emph{Escalation barely helps here, because two
effects almost exactly cancel.} The study recalibrates the limit at every sample size (design~A of
Table~\ref{tab:designs}, compared with the alternatives in \S\ref{sec:unify}). As $n$ grows,
the posterior for $\sigma$ narrows, with $\mathrm{sd}(\ln\sigma)\approx1/\sqrt{2(n-1)}$. But $g$
simultaneously becomes a steeper function of $\sigma$, because a larger sample makes the rule
discriminate more sharply. The interval width scales with the product of the two, and Figure~\ref{fig:cehd} and Table~\ref{tab:cancel} (Appendix~\ref{app:cancelholds}) show the two curves crossing past each other.

\begin{figure}[ht]\centering
\includegraphics[width=\linewidth]{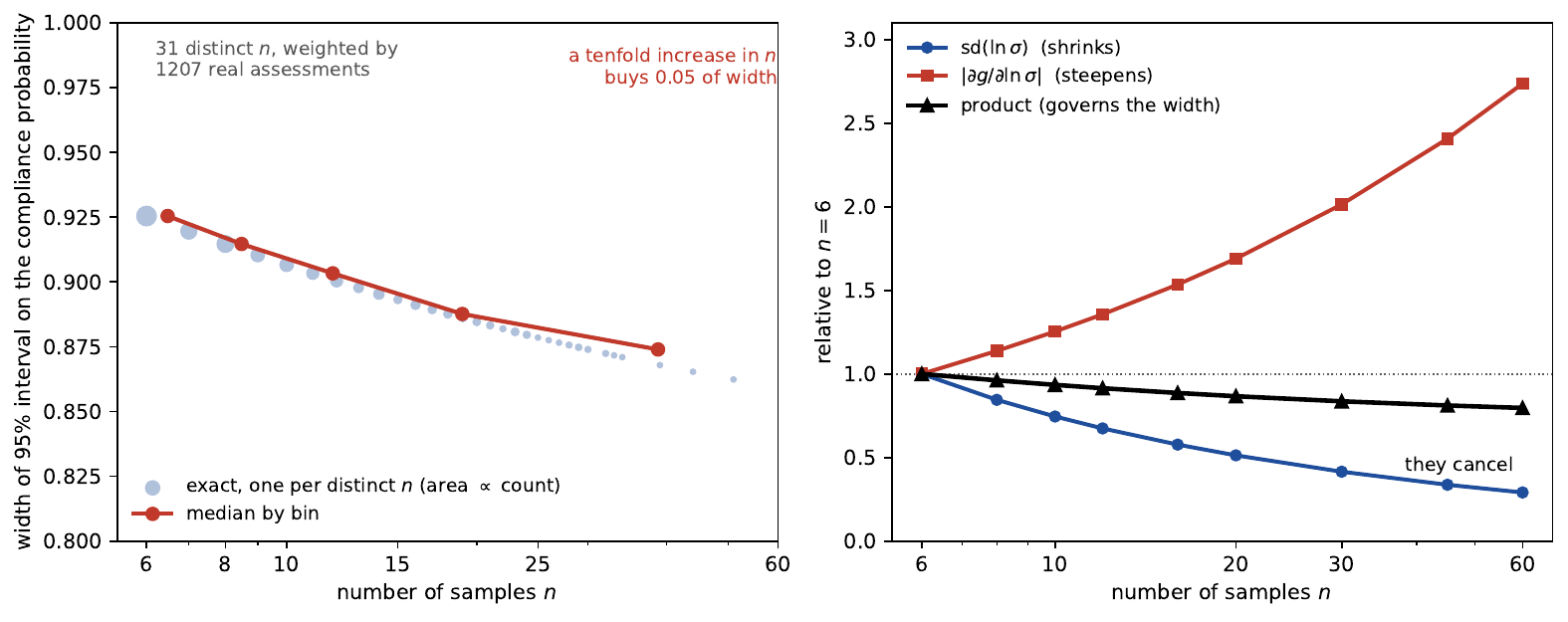}
\caption{Left: width of the $95\%$ interval on the compliance probability against sample size,
one exact point per distinct $n$ with marker area proportional to how many of the $1207$
assessments sit there; a tenfold increase in $n$, from $6$ to $60$, buys $0.066$ of width,
$0.925$ against $0.859$. Right: why. The
posterior spread of $\ln\sigma$ shrinks by a factor $3.4$ between $n=6$ and $n=60$ while the
sensitivity of the acceptance probability to $\ln\sigma$ steepens by a factor $2.7$; their
product falls only to $0.80$ of its starting value.}\label{fig:cehd}
\end{figure}

\noindent Stated plainly: \emph{the information a larger sample gives about $\sigma$ is very
nearly canceled by the leverage $\sigma$ gains over the answer.} This is sharper than
diminishing returns, a near-exact offset, and it is why Table~\ref{tab:cehd} is so close to
flat in $n$. It bears on a live position in the field: a recent assessment of the revised
EN\,689 argues that requiring three measurements rather than one ``will therefore improve the
precision and accuracy of the risk assessment'' \cite{annweh2022}. That is correct for the
exposure \emph{parameter} and very nearly false for the \emph{compliance probability}, and the
distinction is invisible without the calculation.

\section{Why the three studies agree}\label{sec:unify}

The three studies of \S\ref{sec:ladder}, the analytic anchor and the two empirical sectors, are
three views of one structure. Where $\sigma$ is the only uncertain parameter, the
width is exactly \eqref{eq:exactw}: $g$ evaluated at two posterior quantiles of $\sigma$, whose
ratio is fixed by $n$. Linearizing $g$ between those two points writes the width as an
\emph{amplification}, the magnitude $|\partial g/\partial\ln\sigma|$ of the slope of $g$ in $\ln\sigma$, times a posterior spread,
\begin{equation}\label{eq:widthprod}
\text{width}\;\approx\;2\times1.96\;\Big|\frac{\partial g}{\partial\ln\sigma}\Big|\times\mathrm{sd}(\ln\sigma),
\qquad \mathrm{sd}(\ln\sigma)\approx\frac1{\sqrt{2(n-1)}} .
\end{equation}
This is a scaling and not an identity. It holds where the two quantiles fall on a near-linear
stretch of $g$. Saturation, with $g$ near $0$ or $1$, truncates it, and no series taken at a
point reaches across the half unit of $\ln\sigma$ that the posterior spans.
Section~\ref{sec:transfer} measures where it can be trusted: in the contested band, which is
the range this section is about. The first factor, the amplification, grows like $\sqrt n$ with an additive term beside the
leading one, not a different exponent,
\begin{equation}\label{eq:ampscale}
\Big|\frac{\partial g}{\partial\ln\sigma}\Big|\;=\;C\,\sqrt n\;+\;D\;+\;O(n^{-1/2}),
\end{equation}
so a power law fitted over the range a sample-size table covers reports an exponent near $0.44$
with every appearance of precision (Table~\ref{tab:ampscale}, Appendix~\ref{app:constants}).

For the one-sided rule of \S\ref{sec:cehd} both constants are closed, $C=0.250028$ and
$D=0.266992$, and so they are for the pure-spread rule of Proposition~\ref{prop:chi}
(Appendix~\ref{app:constants}).

The second factor is the \emph{shape} in $g$. Writing $\mathcal{A}(g)$ for each curve normalized
by its own peak,
\begin{equation}\label{eq:shape}
\mathcal{A}_n(g)\;=\;\frac{\varphi(z_g)}{\varphi(0)}\Bigl(1+t_1\,z_g\,n^{-1/2}+O(n^{-1})\Bigr),
\qquad z_g=\Phi^{-1}(g),
\end{equation}
with $t_1=0.830$ for the one-sided rule, so $\mathcal{A}$ is $n$-free only to leading order,
and least so toward the ends of the range. 

Substituting \eqref{eq:ampscale} into \eqref{eq:widthprod} gives the result that ties the paper
together:
\begin{equation}\label{eq:cancel}
\text{width}\;\propto\;\frac{C\sqrt n+D+\cdots}{\sqrt{2(n-1)}}\,\mathcal{A}(g)
\;\longrightarrow\;\frac{C}{\sqrt2}\,\mathcal{A}(g),
\qquad \frac{C}{\sqrt2}=0.1768 ,
\end{equation}
\emph{in which $n$ cancels in the limit}. At finite $n$ the residual has relative size
$D/(C\sqrt n)$. The product descends toward $0.1768$, from $0.2525$ at $n=6$ to $0.2011$ at
$n=60$. That $20\%$ fall measures how far $n=60$ still is from the limit. Pure $\sqrt n$ scaling
would predict only $8\%$, and the gap between the two is what reveals the constant $D$. Equation~\eqref{eq:cancel} is a statement about the \emph{rate}, which is why it is written
with a proportionality sign; the width itself is \eqref{eq:exactw}, and the tangent form misstates
both its shape and its constant when the two sample sizes are comparable
(Appendix~\ref{app:cancelholds}). The qualitative claim survives intact: the information a larger
sample supplies about $\sigma$ is very nearly offset by the leverage $\sigma$ thereby acquires.

\paragraph{What the cancellation assumes.} Equation~\eqref{eq:cancel} holds the operating point
fixed while $n$ grows. Recalibrating the limit is one way to arrange that; a process sitting on a
\emph{fixed} boundary is another, and there too the width need not vanish: for the pure-spread
rule with an observed $s$ at the limit and $\mest=\ntest$, the posterior of $\Pacc$ is exactly
uniform and the width is $0.95$ at every sample size. Away from a boundary, a coupled design with a
fixed limit does narrow, by saturation rather than information. The same cancellation follows
with $\sigma$ known, on the location coordinate alone, so what cancels is a \emph{rate}: any
coordinate whose posterior contracts at the rule's own rate cancels in the same way
(Appendix~\ref{app:cancelholds}).

\paragraph{The three designs, side by side.} A sample size and an acceptance constant
\emph{together} fix how hard a test is, so moving one and leaving the other changes the test:
a process at $\mu=100$, $\sigma=5$ passes the Stage-1 constant with probability $0.855$ at ten
units and $0.970$ at thirty. Calibrating each design to report $g=0.85$ with width $0.821$ at
ten units makes any later divergence attributable to the design alone. Design~C is the one a
reader is most likely in: the rule stays at the $\langle905\rangle$ constants $(10,2.4)$ and only
the estimation sample grows. Design~A re-sets the limit at every $n$ with $\mest=\ntest$, a
controlled comparison that holds $\mathcal{A}(g)$ still. Design~B couples the sizes and leaves
the limit where it is. Table~\ref{tab:designs} runs them out.

\begin{table}[h]\centering\small
\caption{Three designs, all calibrated to $g=0.85$ and $W=0.821$ at ten units, in the
location-only half-plane model of Appendix~\ref{app:wsurface}, $\sigma$ known and the mean
unknown, so that each is a slice of \eqref{eq:wsurface}; the entries are not a wedge
calculation, and the wedge analogue of design~C is the fixed Stage-1 question of
Appendix~\ref{app:twostage}. The $g_0$ column is
what the table is for: B reaches a narrow width by moving its own question to an endpoint; only
A and C answer the same question at every size, and only C leaves the rule as the standard
writes it.}
\label{tab:designs}
\begin{tabular}{@{}llcccccc@{}}
\toprule
& units & $10$ & $20$ & $40$ & $80$ & $200$ & $1000$\\
\midrule
C\ \ test fixed ($\ntest=10$)
 & $g_0$ & $0.850$ & $0.850$ & $0.850$ & $0.850$ & $0.850$ & $0.850$\\
 & $W$   & $0.821$ & $0.629$ & $0.456$ & $0.324$ & $0.205$ & $0.091$\\[3pt]
A\ \ coupled, limit recalibrated
 & $g_0$ & $0.850$ & $0.850$ & $0.850$ & $0.850$ & $0.850$ & $0.850$\\
 & $W$   & $0.821$ & $0.821$ & $0.821$ & $0.821$ & $0.821$ & $0.821$\\[3pt]
B\ \ coupled, limit fixed
 & $g_0$ & $0.850$ & $0.929$ & $0.981$ & $0.998$ & $1.000$ & $1.000$\\
 & $W$   & $0.821$ & $0.689$ & $0.455$ & $0.166$ & $0.004$ & $0.000$\\
\bottomrule
\end{tabular}
\end{table}

A moves the goalposts, B changes the examination, and C alone asks the same thing twice. A is
flat because recalibration pins $\mathcal{A}(g)$, which is \eqref{eq:cancel}; B collapses
because the rule stops discriminating, which is saturation; C descends at $\mest^{-1/2}$
because the question is fixed and the evidence accumulates, which is ordinary inference.
\emph{Where a standard fixes the constants, as $\langle905\rangle$ fixes $(10,2.4)$ and
$(30,2.0)$, C is the design that remains, and along it more measurement narrows the
interval at the surface's rate.} Its cost is real: the width falls only as $\sqrt{\mest}$, so pinning an interior
probability to $\pm0.05$ at $\ntest=10$ takes on the order of a thousand units.

\paragraph{Where it fails.} The cancellation is a statement about regular families; where a
parameter is estimable faster than $n^{-1/2}$, as the support boundary of a Uniform parent is, the
width does fall (Appendix~\ref{app:cancelholds}).

The three-regime shape of \S\ref{sec:reallots}, low at both ends and high between, is the
same equation read along $g$: width is proportional
to $\mathcal{A}(g)$, which vanishes as the rule saturates, so an interval collapses for lots that
clearly pass and for lots that clearly fail and opens only between them. And the exposure
assessments of \S\ref{sec:cehd} sit permanently in the open region not because their sector
differs but because $n\approx8$ leaves $\mathrm{sd}(\ln\sigma)$ large while $\mathcal{A}$ is
near its maximum. The equation has a classical relative in trial design: assurance, the mean of this pushforward,
is known not to approach one as the trial grows \cite{ohagan2005}, and \eqref{eq:cancel} is a
width-level analogue with an explicit shape and rate, though a mean that stays away from one does
not by itself fix a width.

\subsection{Predicting the width for a geometry not evaluated}\label{sec:transfer}

If \eqref{eq:cancel} is right, it should price a rule this paper carries no data for. We test it
on the two remaining geometries of Figure~\ref{fig:geom}, panels (c) and (d): capability release at $n=30$ and
percent-within-limits at $n=10$, the exact-evaluation paper's other two applications. Neither has a study of its own here: the published capability data \cite{jiang2026b} are of
undisclosed provenance, and the one public per-lot percent-within-limits record \cite{wsdot517}
was located after this test was fixed and is left as out-of-sample ground. As
predictions they cost nothing: each rule is calibrated to a target operating point, its interval
computed from $1500$ posterior draws, and compared against the width read off the tent master
curve at the same $g$. The transfer errors are single runs without a replicate standard error,
so differences of a few hundredths between the columns of Table~\ref{tab:transfer} are within
their unresolved outer error.

\begin{table}[h]\centering\small
\caption{Predicting the interval width of two unseen geometries from the plain-tent master
curve. The comparison column is the law's own residual in $n$: the spread of tent widths across
$n\in\{3,10,30\}$ at the same $g$. Above $g=0.3$, changing the geometry costs less than changing
the sample size does, at the configurations tested here.}
\label{tab:transfer}
\begin{tabular}{@{}lccccc@{}}
\toprule
& \multicolumn{2}{c}{geometry transfer error} & &
\multicolumn{2}{c}{law's own $n$-residual}\\
\cmidrule(r){2-3}\cmidrule(l){5-6}
operating point & median & max & & median & max\\
\midrule
contested and passing, $g>0.3$ & $0.029$ & $0.071$ & & $0.052$ & $0.107$\\
clearly failing, $g\le0.3$     & $0.094$ & $0.128$ & & $0.089$ & $0.145$\\
\bottomrule
\end{tabular}
\end{table}

Table~\ref{tab:transfer} gives the result. Above $g=0.3$, which covers the contested band and
everything above it and so is where release decisions are actually taken, the prediction
transfers well. Swapping a wedge of slope $k=1$ for a wedge of slope $k\approx4$, or for a
curved dome, moves the predicted width by a median of $0.029$ and never by more than $0.071$.
That is \emph{smaller} than the median error of $0.052$ the curve already incurs because the
sample size does not cancel perfectly. Over this grid, geometry is the lesser nuisance: it moves
the width less than the sample size does. The experiment is exploratory: seven target probabilities, on-target Normal parents,
and single unreplicated runs, which resolves the ordering of the two columns but not a threshold.

Below $g=0.3$ the prediction degrades to a median error of $0.094$, always in the same
direction: both unseen rules carry wider intervals than the tent predicts. Both are more spread-dominated than the $k=1$ tent on which $\mathcal{A}$ was fitted,
capability having $k=3c\approx3.3$ to $4.7$ over this sweep, so $\mathcal{A}$ depends weakly on the
effective spread weight, the axis along which \cite{paper1} already organizes its applications. To
price a new rule from its operating point alone, $g$ sufficed above $0.3$ and $k$ was needed below
it, over seven target probabilities at on-target Normal parents; whether that holds off target is untested.

\subsection{Where the uncertainty lives}\label{sec:decomp}

Making one coordinate of $\theta$ uncertain at a time at the $\langle905\rangle$ operating point
of \S\ref{sec:905}, under its law in \eqref{eq:cspost} with the others at the plug-in, gives
widths of $0.885$ for the within-lot spread $w$, $0.159$ for the mean $\mu$ and $0.289$ for the
correlation $\rho$ (Table~\ref{tab:decomp}, Appendix~\ref{app:decomp}). The ordering
$w>\rho\approx\mu$ says where to spend measurement effort, given uncertainty of the stated size in
each coordinate. It prices those ranges together with the rule's own weighting of spread against
location, so it is a statement about $\langle905\rangle$ and rules that weight spread as it does,
not an intrinsic property of the coordinates; the bands are not an additive decomposition.

\section{When an approximation in current use may replace the pushforward}\label{sec:corner}

Treatment 3 of \S\ref{sec:three} evaluates $g$ at the corners of a box confidence set, which is
valid only if $g$ is monotone in each coordinate. At the $\langle905\rangle$ operating point it is
not monotone in the mean: the rule is symmetric about target, so over a box containing the target
the maximum of $g$ sits at the target mean on the low-spread face, not at a corner. With the box of
Appendix~\ref{app:corner} the eight corners give $[0.018,\,0.992]$ against a true range of
$[0.018,\,0.999]$: the minimum is attained and the maximum is missed by $0.007$, seven hundred
times the evaluator's resolution. The repair is cheap: evaluate the admissible mean closest to
target for the maximum on every face, or search the box. For a Normal parent the corner values have
a closed approximation, Bergum's construction in the spread coordinate with the mean channel
frozen, which agrees with the exact corner values on the ninety-eight lots to a median of
$3\times10^{-5}$ and a maximum of $0.068$ (Appendix~\ref{app:cornerdetail}). The two constructions
answer different questions, a posterior probability against a worst case over a set whose
coverage must be stated separately \cite{lewis2016}.

\subsection{The delta band, given its best case}\label{sec:link}

An objection is owed a hearing. No competent analyst bands a bounded parameter on the raw
scale; one bands on a link and maps back, which cannot leave $[0,1]$. With $\mathrm{sd}$ the
gradient standard deviation used throughout, the logit band is
$\mathrm{expit}\big(\mathrm{logit}\,\hat g\pm z\,\mathrm{sd}/[\hat g(1-\hat g)]\big)$, and the
probit band replaces $\hat g(1-\hat g)$ by $\phi(\Phi^{-1}(\hat g))$.

\begin{table}[h]\centering\small
\caption{The delta band at its best. Exposure study: medians over all $1207$ assessments, each
calibrated to $\hat g=0.85$. $\langle905\rangle$: the Stage-2 operating point under \eqref{eq:cspost}, the delta band with
the full posterior covariance and a converged gradient, centered at the plug-in. The link
repairs the range violation everywhere and comes within about a hundredth of the exact
endpoints in both cases.}
\label{tab:link}
\begin{tabular}{@{}lcccc@{}}
\toprule
& \multicolumn{2}{c}{$1207$ exposure assessments} & \multicolumn{2}{c}{$\langle905\rangle$
Stage 2}\\
\cmidrule(r){2-3}\cmidrule(l){4-5}
band & interval & width & interval & width\\
\midrule
naive $\hat g\pm z\,\mathrm{sd}$ & $[0.319,\,1.381]$ & $1.062$ & $[0.245,\,1.455]$ & $1.210$\\
clipped to $[0,1]$               & $[0.319,\,1.000]$ & $0.681$ & $[0.245,\,1.000]$ & $0.755$\\
logit                            & $[0.081,\,0.997]$ & $0.916$ & $[0.047,\,0.998]$ & $0.952$\\
probit                           & $[0.107,\,0.9995]$ & $0.892$ & $[0.060,\,1.000]$ & $0.940$\\
\emph{exact pushforward}       & $[0.085,\,0.999]$ & $0.915$ &
$[0.035,\,0.983]$ & $0.948$\\
\bottomrule
\end{tabular}
\end{table}

Table~\ref{tab:link} separates the two failures of the naive band. The range violation is gross,
and clipping removes it at the top but leaves the floor at $0.245$ against a true $0.035$ at
$\langle905\rangle$, and at $0.319$ against $0.085$ on the exposure study. The accuracy of the link
repair is the surprise: on the exposure study the logit band reproduces the exact endpoints to a
median of $0.004$ and $0.002$, and at $\langle905\rangle$ it gives $[0.047,\,0.998]$ against an
exact $[0.035,\,0.983]$. That accuracy is an observation about two points with no criterion behind
it, which is a reason to verify the shortcut once per rule, at the cost of the exact evaluation it
was meant to avoid (Appendices~\ref{app:cornerdetail} and \ref{app:delta}).

\section{Cost, reproducibility, and extensions}\label{sec:cost}

An interval costs $J$ evaluations of $g$, each on a grid whose dimension does not grow with
$n$: at a fixed grid a $J$-node interval at $n=30$ costs what it costs at $n=10$ and at
$n=1000$, though the grid a prescribed accuracy demands may depend on $n$, on the parent and on
the region (Appendix~\ref{app:engine}). Outside the exploratory two-stage
appendix no lot is simulated, and the one stochastic element, sampling $\theta$ for the outer
integral, is fixed by a seed, so the endpoints are reproducible under audit. A seeded nested simulation is
reproducible in the same sense; what exactness removes is the inner error source, so that the
width is attributable to the parameter law alone.  The construction generalizes in the direction practice needs: a
posterior on $\theta$ accumulated across lots pushes forward to a posterior on the acceptance
probability by quadrature alone, turning ``what is the release risk for this process, and how
sure are we'' into a computation rather than a simulation study. Two extensions are immediate
and untreated: censored parents, which the exposure archive demands, and sequential two-stage
rules, where the stages couple and whose exact treatment remains open
(Appendix~\ref{app:twostage}).

\section{Conclusion}\label{sec:conclusion}

The exact-evaluation paper \cite{paper1} showed that the acceptance probability of a mean--spread
release rule can be computed exactly, and that the incumbent approximation gets it wrong by an
\emph{approximation error} of the order of a tenth, with a sign that changes across the range.
This paper asks what that exact number is worth once one admits that the parent it conditions
on was estimated, and the answer is that a larger uncertainty sits underneath the first:
across the contested lots among ninety-eight published content-uniformity lots the median
interval width is $W=0.77$, and on $1207$ real exposure assessments, each calibrated to report
$0.850$, that number carries $[0.085,\,0.999]$. The first is an approximation error under a stated parent; the second is a conditional
posterior width, a statement about what the data leave open given the model and prior. They are a
sequence: making the evaluation exact removes the first and so makes the second visible and
attributable. Exact evaluation under an assumed parent does not validate that parent.
What keeps this from being merely discouraging is that the width is not arbitrary. It follows a
single curve, $\mathcal{A}(g)$ through \eqref{eq:cancel}, in which the geometry of the rule very
nearly cancels and the sample size cancels along a design that holds the operating point fixed;
for a pure-spread rule under a Normal parent it is a closed formula in the plug-in value and the
two sample sizes (Proposition~\ref{prop:chi}). On a disturbed process, parent shape moves the
interval as much as it moves the acceptance probability itself; on routine lots both effects are a
few hundredths. The practical consequences are conditional, and are collected below.

\par\medskip\noindent\rule{\linewidth}{0.4pt}
{\small\ifsubmission\singlespacing\fi
\noindent\emph{What to do with a reported acceptance probability}
\begin{enumerate}\itemsep2pt\parskip0pt
\item \textbf{Quote an interval, not a number, and check it even when the point estimate looks
settled.} The median width is below $0.001$ for $\hat g>0.99$ and $0.77$ in the contested band,
yet eight of the seventy-six apparently-passing lots carry a width above $0.10$
(\S\ref{sec:reallots}). For a pure-spread rule under a Normal parent Table~\ref{tab:chilaw} gives
the lower endpoint from the plug-in value alone; for the $\langle905\rangle$ wedge compute it
directly, which the one-dimensional evaluator makes cheap.
\item \textbf{Expect a larger sample to narrow it, unless the test grows with the sample.}
Against a fixed Stage-1 question the width falls from $0.97$ at $\mest=10$ to $0.72$ at
$\mest=30$; along the coupled path at a recalibrated limit it falls only $0.925\to0.859$ from
$n=6$ to $n=60$ (\S\ref{sec:unify}).
\item \textbf{To narrow it, reduce the uncertainty about the spread.} At the $\langle905\rangle$
point the one-at-a-time widths for the within-lot spread, the mean and the correlation are
$0.885$, $0.159$ and $0.289$ (\S\ref{sec:decomp}).
\item \textbf{Integrate over $\theta$ with a scheme that resolves the quantiles, and quote the
endpoint's own error} (Remark~\ref{rem:conv}).
\item \textbf{Check monotonicity before trusting a corner band.} At the $\langle905\rangle$ point
the corners attain the minimum but miss the maximum, which sits at the target mean inside the box
(\S\ref{sec:corner}).
\item \textbf{A delta band on a link came within about a hundredth of the exact endpoints at
both operating points, with no criterion that says when it will}; the symmetric band returns an
upper endpoint above one at $\langle905\rangle$ and in all $1207$ exposure assessments
(\S\ref{sec:link}).
\item \textbf{To predict the width for a rule not yet evaluated, the operating point sufficed
above $g=0.3$} in the configurations tested; below it the effective spread weight $k$ was needed
as well (\S\ref{sec:transfer}).
\item \textbf{Allow for parent shape on both the endpoint and the width.} Under a strong right
skew the $2.5\%$ endpoint falls by up to $0.22$ and the width rises by up to $0.22$, on lots that
appear to pass comfortably; ten to thirty units do not identify the shape (\S\ref{sec:skewed}).
\end{enumerate}\par}
\noindent\rule{\linewidth}{0.4pt}\par\medskip

The through-line is that exactness does not make the answer certain; it makes the remaining
uncertainty attributable. Every statement above names which coordinate of the parent, or which
region of the operating characteristic, is responsible, and none could be made while a
computational error of unknown sign was still in the way.

Five problems are left open and recorded in Appendix~\ref{app:open}: certified tail quantiles
for a non-Normal parent, a coherent parameter law for a nonregular parent, left-censored parents,
real exposure limits in place of calibrated ones, and uncertainty about the parent family.

\ifblind\else
\section*{Acknowledgments}
I thank Kevin Chu and Ron Yu for many discussions of the variance-distribution problems from
which the present work grew, and Kevin Chu, Nicole Tin, Somdatta Bhattacharya and Chiajen Lai
for extensive conversation about the pharmaceutical applications and in-silico tools that
motivated this extension. I thank Mandeep Chauhan, Mike Salem, Ash Baghai, Wei-Ting Liu, Shikhar Mohan,
Adedayo Sanni, Chris Remmers, Caroline Branch, Kiara Cui, Dana Caulder, Joanna Koziara, Victor Rucker, Jennifer Hedborn and
Henry Morrison for motivating discussions at Gilead Sciences.
\fi

\section*{Disclosure of interest}
\ifblind
The author reports no conflict of interest.
\else
The author is employed by Gilead Sciences. The author reports no conflict of interest.
\fi

\section*{Author contributions}
\ifblind
The sole author carried out all of the work.
\else
Remus Osan: Conceptualization, Methodology, Software, Formal analysis, Investigation, Data
curation, Visualization, Writing -- original draft, Writing -- review and editing.
\fi

\section*{Funding}
No funding was obtained for the reported work.

\section*{Declaration of generative AI use}
Claude models (Anthropic) assisted with derivations, code implementation, and manuscript
drafting. All outputs were verified against independent symbolic and numerical computations
(Python, MATLAB).

\section*{Data availability}
Every dataset used is public: the ninety-eight content-uniformity lots are the reported per-lot
summaries of four published studies cited in \S\ref{sec:reallots}, the per-tablet values behind
the parent-shape estimate of \S\ref{sec:skewed} are the supplementary tables of one of them and
the reference contents of the Copenhagen tablet set \cite{dyrby2002}, the exposure records are
the public-domain OSHA Chemical Exposure Health Data archive, and the percent-within-limits
record of a state agency report is cited as out-of-sample ground and is not used in any
calculation here. The Copenhagen archive is read by its batch identifiers, thirty-one batches of
ten tablets, and the pilot batches are the ones carried into the shape analysis; the scale
mapping is given in Appendix~\ref{app:nnlik}. The code that reproduces the results, figures and
tables in the paper and its supplement accompanies the published version of this article\ifsubmission.\else{} and
will be added to a later version of this preprint.\fi

\bigskip
\noindent\textbf{Supplementary material.} Appendices~\ref{app:engine}--\ref{app:open}, after the
references, contain the evaluator and its validation, the extended forms of
Proposition~\ref{prop:chi}, the three treatments carried through a common operating point, the
closed-form $n=3$ anchor, three real lots worked in detail, the two-stage $\langle905\rangle$
sequence, the width law and its constants, and the open problems.

\clearpage
\begin{center}
{\Large\bfseries Supplementary Material}
\end{center}
\vspace{0.5em}
\noindent Lettered sections and S-numbered floats and equations (Appendix~A, Figure~S1,
equation~(S1), \dots) belong to this supplement; plain-numbered references (\S1,
Proposition~1, Table~1, \dots) are to the article above.
\vspace{1em}

\appendix
\setcounter{figure}{0}\renewcommand{\thefigure}{S\arabic{figure}}
\renewcommand{\theHfigure}{S\arabic{figure}}
\setcounter{table}{0}\renewcommand{\thetable}{S\arabic{table}}
\renewcommand{\theHtable}{S\arabic{table}}
\setcounter{equation}{0}\renewcommand{\theequation}{S\arabic{equation}}
\renewcommand{\theHequation}{S\arabic{equation}}

\section{The evaluator: discretization and validation}\label{app:engine}

Theorem~\ref{thm:engine} is an identity. What this paper computes is a discretization of it, and
the two are not the same object: the identity is exact, the values are not. This appendix states
the scheme and reports what the values have been checked against, so that the numerical
foundation can be assessed without \cite{paper1} in hand.

\subsection{The scheme}

The parent is centered, $Y=X-\mu$, and the inversion \eqref{eq:inv} is carried out for
$(T_1,T_2)=(\sum_i Y_i,\sum_i Y_i^{2})$ on a uniform $N\times N$ grid,
$\tau_1\in[\tau_1^{\min},\tau_1^{\max})$ and $\tau_2\in[\tau_2^{\min},\tau_2^{\max})$ with
spacings $\delta_1,\delta_2$. The conjugate frequency grids are
$t_k=(k-N/2)\,\delta t$ and $u_l=(l-N/2)\,\delta u$ with $\delta t\,\delta_1=\delta u\,\delta_2=2\pi/N$,
and the density is recovered by one two-dimensional FFT,
\begin{equation}\label{eq:fft}
f[a,b]=\frac{\delta t\,\delta u}{(2\pi)^{2}}\,(-1)^{a+b}\,
\operatorname{FFT}_2\bigl(\psi^{\,n}\cdot\text{phase}\bigr)[a,b],
\end{equation}
the phase factor carrying the grid origin. The acceptance probability is then the sum of
$f$ over the cells whose $(\Xbar,s)$ image lies in $A$, times $\delta_1\delta_2$. Every reported
value in this paper uses $N=1024$ unless stated otherwise; the anchor of \S\ref{sec:anchor} uses
$N=4096$.

Three approximations separate \eqref{eq:fft} from Theorem~\ref{thm:engine}, and they are
different in kind.
\begin{enumerate}\itemsep2pt
\item \textbf{Truncation.} The grid is finite, so mass outside it would be lost. At the
anchor none is: the support $T_1\in[-1.5,1.5]$, $T_2\in[0,0.75]$ lies inside the box, and the
density summed over the whole grid is one to machine precision. Summed over the admitted cells
it is $0.999771$, so the $2.3\times10^{-4}$ deficit is mass that the inversion places outside
the admissible cone and the cell-admission mask then discards: an inversion and masking effect,
not omitted support.
\item \textbf{Quadrature.} The frequency integral becomes a finite sum, and the density a
midpoint rule. This is the error that shrinks as $N$ grows, and the ladders below measure it.
\item \textbf{Cell admission.} A cell is admitted on a test at its center, so cells straddling
$\partial A$ are counted whole or not at all. \emph{This error changes under refinement too, but
not in the same way as the quadrature error, and a refinement study mixes the two rather than
separating them}, which is why it is measured separately, by
varying the convention at fixed $N$; Appendix~\ref{app:boundary} reports that measurement.
\end{enumerate}

\subsection{What the values were checked against}

Table~\ref{tab:validation} collects the checks. Two are references external to the scheme --- an
analytic constant and a distribution in \texttt{scipy} --- and the rest are internal consistency
of a kind that can only bound an error, not certify it.

\begin{table}[h]\centering\small
\caption{Validation of the discretized evaluator. Rows above the rule compare against a target
computed \emph{outside} the scheme; rows below vary a parameter of the scheme and report the
movement. Every value is reproduced by the check suite of the reference implementation; the count of checks it
reports is not quoted here, being a number that moves with the checker rather than with the
paper.}
\label{tab:validation}
\begin{tabular}{@{}p{0.30\linewidth}p{0.36\linewidth}p{0.28\linewidth}@{}}
\toprule
Check & Reference & Agreement\\
\midrule
$n=3$ Uniform anchor, $N=4096$ & $6+18\ln\tfrac34=0.821723$, closed form; independently
confirmed by $31$-digit quadrature and by symbolic integration & $2.3\times10^{-4}$\\
One-sided rule, Normal parent & noncentral $t$, \texttt{scipy.stats.nct} & $<10^{-8}$\\
Pure-spread rule, $n=3$--$30$ & $2\times10^{6}$-draw Monte Carlo & $\le1.0\times10^{-3}$\\
\midrule
Grid ladder $N=2048\to4096$, widths & --- & $1.1\times10^{-4}$ at $n=3$; $\le3.4\times10^{-5}$ at $n=4,10,30$\\
Grid ladder $N=1024\to8192$, spread & --- & $3.2\times10^{-4}$ at $n=3$; $\le9.3\times10^{-5}$ elsewhere\\
Cell-admission convention at fixed $N$ & --- & observed $8.9\times10^{-5}$ (box), $3.0\times10^{-4}$ (posterior support), sign not stable (App.~\ref{app:boundary})\\
\bottomrule
\end{tabular}
\end{table}

The exposure rule's one-dimensional evaluator \eqref{eq:ghalf} carries its own validation:
against $4\times10^{6}$-draw Monte Carlo at $n\in\{6,8,12,30\}$ and $\delta\in\{1.5,2.5,3.5\}$,
eleven of twelve points agree within four Monte-Carlo standard errors, the twelfth differing by
$2.4\times10^{-4}$ at $g\approx0.9989$ where the quadrature is coarsest in the upper tail; and
the full pushforward, end to end against direct simulation of the rule at $n=8$, gives
$\E[g]=0.71409$ by quadrature against $0.71538$ ($\pm0.00101$) simulated.

Three things should be read off this table rather than around it. \emph{The anchor row is the
only one that measures an absolute error}, because it alone compares against a number known in
closed form. It does so at one parameter value, which bounds nothing at the posterior
endpoints elsewhere. Still, $2.3\times10^{-4}$ at $N=4096$ is the honest scale of the discretization at the
hardest point the paper evaluates. The case $n=3$ is hardest for the reason Theorem~\ref{thm:engine} already
names: it is where $\psi^{n}$ is least integrable, the regime in which the hypothesis weakens
from $L^{1}$ to $L^{2}$. \emph{The ladder rows are convergence diagnostics for the quadrature
error and not error bounds}: a sequence that moves little can still be biased, and small
movement is evidence of convergence only together with a rate. At $n=3$ successive
halvings reduce the movement by a factor of about $2$ rather than $4$, so convergence is first
order in $N$ at that point and the ladder understates what a further refinement would move.
\emph{The convention row is the one a reader would otherwise not think to ask for}: it is the
same size as the quadrature error, it changes sign, and refinement alone does not separate it
from the quadrature error; changing the convention does.

Against the widths this paper reports, all of these are immaterial --- an error of $10^{-4}$
against widths of order $1$ --- and that is the claim the table supports. It does not support a
claim that the shipped values are exact. \emph{``Exact'' in this paper describes
Theorem~\ref{thm:engine}, not the numbers}: what the identity removes is the simulation error of
the incumbent construction, and what remains is a deterministic and measurable discretization
error, reported above rather than assumed away.

\subsection{Boundary-cell convention: what was measured}\label{app:boundary}

\paragraph{Grid bias displaces the interval.} Each $g(\theta^{(j)})$ carries the deterministic bias of the underlying inversion, held below
$5\times10^{-4}$ in \cite{paper1} against an analytic anchor; it does not shrink by averaging
over $j$. If it were constant in $\theta$ it would displace the interval without changing its
width. It is not constant, because the acceptance region moves with $\theta$. The error at the two
endpoints is therefore neither the same number nor necessarily of the same sign, and a variation of
size $\epsilon$ across them moves the width by up to $2\epsilon$. Grid refinement at a fixed
cell-admission rule measures resolution ($3.5\times10^{-5}$ across $N=1024$ to $8192$), not the
effect of the admission convention itself. Appendix~\ref{app:boundary} probes the convention by
changing it at fixed $N$. The largest spread of the difference across the operating box is
$8.9\times10^{-5}$, against a $g$ ranging over $0.83$. Over the region where the posterior puts its
mass the observed ranges are larger, $3.0\times10^{-4}$ and $1.8\times10^{-4}$ on the two branches
(Appendix~\ref{app:boundary}), but still three orders of magnitude below the widths reported. The
width thus enjoys no guaranteed cancellation, but the error this leaves is immaterial at the
precision quoted. That is a statement from measurement rather than assumption.

What follows is the measurement behind that statement.
The question is not how finely the inversion is resolved but which of two admissible
\emph{conventions} for the boundary cells is in force, since a comparison between grids that
share a convention cannot see it.

Each cell's weight is recomputed by $8\times8$ subsampling, so a cell straddling the boundary
carries the fraction of itself that is admitted rather than being counted whole or not at all.
At one subsample the calculation reproduces the shipped evaluator to $7\times10^{-16}$, so the
comparison isolates the convention and nothing else. Evaluated at the corners of the same
$1.96\tau$ box used throughout \S\ref{sec:905}, the center-minus-area difference is

\begin{center}\small
\begin{tabular}{@{}lccc@{}}
\toprule
& range across the box & spread & as \% of the range of $g$\\
\midrule
$\rho=0$ \quad (hard $\Xbar$ mask) & $-2.50\times10^{-5}$ to $+6.39\times10^{-5}$ & $8.89\times10^{-5}$ & $0.011\%$\\
$\rho=0.15$ \quad (smooth weight) & $+1.31\times10^{-8}$ to $+1.77\times10^{-5}$ & $1.77\times10^{-5}$ & $0.002\%$\\
$\rho=0.30$ \quad (smooth weight) & $+2.17\times10^{-8}$ to $+5.17\times10^{-5}$ & $5.17\times10^{-5}$ & $0.009\%$\\
\bottomrule
\end{tabular}
\end{center}

Read that table with one caution in mind. It reports a difference
between two admissible conventions, not a distance from the truth, and calling it a bias would
presume the area-weighted value is the correct one. That presumption is false in general. On an
exactly-known $n=3$ Uniform anchor the exact-evaluation paper \cite{paper1} finds area weighting roughly $1.55$ times
\emph{worse} at every grid. There the leading error is the joint density's own
discontinuity at the cone edge, and weighting the mask does not touch it while the density is
still sampled at cell centers.

For the geometry here the evidence runs the other way, and it can be had without a reference
value. Both conventions converge to the same limit, so whichever moves less as the grid refines
is the closer of the two at any finite $N$. Across $N=1024$ to $8192$ the center-of-cell value
moves by $2.87\times10^{-5}$ at $\rho=0$ and $3.54\times10^{-5}$ at $\rho=0.15$, while the
area-weighted value moves by $3.07\times10^{-6}$ and $2.25\times10^{-6}$: a factor of $9$ and
$16$ steadier. That is what one would expect of a smooth Normal parent whose only hard edges are
the admissibility region and $R>0$, rather than a boundary cutting a discontinuity.

That inference can be checked, because a member of this family has a closed form in exactly the
geometry it concerns. Under a Normal parent the one-sided rule $\Xbar+ks\le L$ has acceptance
probability $T_{n-1}\big(-k\sqrt n;\,-\sqrt n(L-\mu)/\sigma\big)$, the noncentral-$t$ distribution
function, which we confirmed against an $8\times10^{7}$-draw simulation to $1.7$ standard errors
before using it. At $n=10$, $k=1.5$, $L=110$ and a Normal$(100,5)$ parent, area weighting is
closer to the exact value at every grid tried, and the steadier ladder picks the same winner: the
center-of-cell value moves by $7.3\times10^{-5}$ across $N=1024$ to $4096$ against
$1.1\times10^{-6}$ for area weighting, and its largest error over that ladder is
$5.5\times10^{-5}$ against $1.3\times10^{-6}$. Both figures are taken over the ladder rather than
at its finest grid, because the center-of-cell error
changes sign at every refinement: any statistic read off one grid is a draw rather than a
measurement. The area-weighted error, by contrast, keeps one sign throughout. The size of the margin is
accordingly not quotable. The ratio of the two errors runs from a factor of $14$ to a factor
of $380$ across three grids, because it is a quotient of a quantity that oscillates about zero,
so what this establishes is the direction and not the size of the gain. Where both error
sequences are instead one-signed and monotone, as they are on the exact-evaluation paper's Uniform
anchor, the ratio is stable and worth printing; it is that condition, not the geometry, that makes the difference. We therefore read the area-weighted column as the better of the two here, as an
inference from the convergence that an exact reference confirms in the same geometry, rather than
as a property of the scheme. The single finest step is not the statistic
to use: at $\rho=0.15$ it favors center-of-cell, $1.3\times10^{-7}$ against $2.7\times10^{-7}$,
which is one lucky step in a sequence whose sign alternates.

The table supports three readings. The first corrects any hope that the convention error cancels. The difference \emph{does} vary across $\theta$, it is
largest at the wide-$\sigma$ corner in every case, and at $\rho=0$ it changes sign across the
box. There is therefore no cancellation to appeal to: the convention error at the two endpoints
of an interval is not the same number and need not carry the same sign. Second, the magnitude
nonetheless does not threaten any conclusion here. On this box the spread is
$8.9\times10^{-5}$ against a $g$ ranging over $0.83$, about a hundredth of a percent; measured
over the region where the posterior puts its mass it is larger, $3.0\times10^{-4}$ on the
smooth branch ($120$ posterior draws, $\rho\sim\mathrm{U}[0,0.3]$) and $1.8\times10^{-4}$ on the
hard-mask branch ($\rho=0$, $42$ points), $0.03\%$ and $0.02\%$ of the span of $g$ there. These
are observed ranges over finite point sets, not bounds, and the draws that build the flagship interval carry $\rho$ drawn on
$[0,0.3]$, so they sit almost surely in the smoother of the two regimes. Third, the effect
shrinks with $N$ --- at $\hat\theta$ and $\rho=0.15$ the difference runs $-3.05\times10^{-5}$,
$+2.66\times10^{-6}$, $-1.06\times10^{-6}$ at $N=1024,2048,4096$ --- which identifies it as
boundary-cell error rather than anything structural. Its sign is not stable in $N$ at that
point, so no rate should be read off a single $\theta$.

The two $\rho$ regimes are reported separately rather than pooled because they are not the same
computation. The constraint on $\Xbar$ is a hard interval only when $\rho\le0$; for $\rho>0$ the
lot center enters through a smooth normal weight, leaving $\mathrm{valid}$ and $R>0$ as the only
hard edges, and the $\rho=0$ column is both the widest spread and the only sign change. A single
boundary-cell figure quoted across all $\rho$ would be an average of two different geometries.

\section{The \texorpdfstring{$W=1$}{W=1} blind spot, worked}\label{app:blindspot}

The interval width $W$ can read $1$ even when the answer is nearly determined. This is not a contrivance: it is the
paper's own mechanism taken to its limit, and it needs no rule but the one already in hand. Recall that $g$ is a \emph{forecast}: the probability
that a test not yet run will pass, formed from the units already assayed. Take
$\langle905\rangle$ Stage~1 as written, $|M-\Xbar|+2.4\,s\le15$, under a Normal parent on target
at $\mu=100$, with the posterior for $\sigma$ carried by ten units of a real lot
($s=4.845$, giving $\sigma$ a median of $5.0$ and a $95\%$ range $[3.3,\,8.9]$). Hold the limit
where the standard puts it and let the prescribed \emph{test} grow. On target, $\Xbar\to\mu$, and $\mu$
lies inside the clipping window. The mean arm of the wedge therefore retracts to nothing, and what
is left is a ceiling on the spread. The operating characteristic converges to a step at
$\sigma^{*}=15/2.4=6.25$. The ten-unit posterior straddles it, placing
$\Prob(\sigma>\sigma^{*})=0.203$ above, so $g$ is driven to $0$ or to $1$ for almost every draw:
\begin{center}\small
\begin{tabular}{@{}rccc@{}}
\toprule
$\ntest$ & $[q_{0.025},\,q_{0.975}]$ & $W$ & posterior mass with $0.01<g<0.99$\\
\midrule
$10$   & $[0.089,\,1.000]$ & $0.911$ & $85.3\%$\\
$30$   & $[0.009,\,1.000]$ & $0.991$ & $56.1\%$\\
$100$  & $[0.000,\,1.000]$ & $1.000$ & $32.3\%$\\
$300$  & $[0.000,\,1.000]$ & $1.000$ & $18.9\%$\\
$1000$ & $[0.000,\,1.000]$ & $1.000$ & $10.2\%$\\
\bottomrule
\end{tabular}
\end{center}
By $\ntest=100$ the reported interval is the whole of $[0,1]$ --- nominally no information at
all --- while the limiting law is two atoms, $0.797$ at $g=1$ against $0.203$ at $g=0$. What is
unknown there is only the direction, and $W$ cannot say so: a forecaster offered odds on this lot
would take four to one. Two things keep this from being over-read. Running the enlarged test
would itself reveal $\sigma$, so the statement is about what is forecastable \emph{before} a test,
which is when the decision to run one is taken; and the growing test at a fixed limit is a
diagnostic rather than a procedure, visible in full at $\ntest=100$ and already severe at $30$.
What repairs it is measurement rather than testing --- at $\mest=50$ the posterior no longer
reaches the step and $W$ collapses to zero. This is the third path of
\S\ref{sec:unify}, a test growing against a fixed limit, seen from the other side. There the interval narrows by
saturation rather than by information; here it widens without the answer becoming less
determined. Where a stronger
statement is needed --- support, tail mass, or the risk of a particular decision --- it is
reported separately and not read off $W$.

\section{Proposition~\ref{prop:chi}, extended}\label{app:prop1ext}

Three supplements to \S\ref{sec:closed}: the closed form at three units, the mean channel that the wedge adds to the pure-spread rule, and the accuracy of the large-sample form, with the upper endpoints that Table~\ref{tab:chilaw} omits.

\subsection{The three-unit case}\label{app:n3law}

\begin{corollary}[The three-unit case]\label{cor:n3}
At $n_t=3$, so $\nu_t=2$, Proposition~\ref{prop:chi} collapses to a power law in the
\emph{failure} probability,
\begin{equation}\label{eq:n3law}
1-g_p\;=\;(1-\hat g)^{\,\chi^2_{\nu_e,\,1-p}/\nu_e},
\end{equation}
for every estimation size $\nu_e$. If the estimation sample is also three, $\nu_e=2$, the
chi-square quantiles are elementary and \eqref{eq:n3law} becomes, with
$\theta=-\ln(1-\hat g)$,
\begin{equation}\label{eq:n33law}
1-g_p\;=\;p^{\,\theta},\qquad\text{so}\qquad
[\,g_{0.975},\,g_{0.025}\,]=\bigl[\,1-0.975^{\theta},\;1-0.025^{\theta}\,\bigr],\qquad
W=0.975^{\theta}-0.025^{\theta}.
\end{equation}
\end{corollary}

\begin{proof}
$F_2(x)=1-e^{-x/2}$ and $\hat q=F_2^{-1}(\hat g)=-2\ln(1-\hat g)$, so the argument of $F_{\nu_t}$
in \eqref{eq:chilaw} is $-\ln(1-\hat g)\,\chi^2_{\nu_e,1-p}/\nu_e$ and \eqref{eq:n3law} follows.
For $\nu_e=2$ the chi-square with two degrees of freedom is exponential,
$\chi^2_{2,\alpha}=-2\ln(1-\alpha)$, so $\chi^2_{2,1-p}/2=\ln(1/p)$, and
$(1-\hat g)^{\ln(1/p)}=p^{\,\theta}$.
\end{proof}

Two features of \eqref{eq:n33law} are worth stating because they are exact rather than numerical.
The width is largest in the interior, at $\theta^\ast=\ln(b/a)/(b-a)$ with $a=\ln(1/0.975)$ and
$b=\ln(1/0.025)$, which is $\hat g^\ast=0.74328$ and $W=0.95953$; the worst case is not at an
extreme and the width never reaches one. And the lower endpoint clears $0.95$ only when
$\theta\ge\ln(20)/a=118.33$, that is $\hat g\ge1-4.1\times10^{-52}$, so three calibration units
cannot certify a release at this rule whatever they report. At a Uniform parent the same two
evaluations need the sample-variance law of \cite{paper1} in place of $F_{\nu_t}$ and the pivot
quantiles of $s/\sigma$ in place of $c_p$. Both are available in closed form at $n=3$, but the
pivot quantiles are roots of that law rather than elementary functions of $p$, so the Uniform case
closes one step less completely than \eqref{eq:n33law}.

\subsection{The mean channel of the wedge}\label{app:meanchannel}

The one-dimensional evaluator \eqref{eq:g1d} of \S\ref{sec:closed}, checked and compared with the chi-square law.

Against the FFT engine at $N=1024$ it agrees to
$10^{-6}$ for $\sigma\le2$ and to $9.3\times10^{-4}$ at $\sigma=8$, the engine's grid error at
large spread. Two distinct approximations must be kept apart here. The first freezes $\Xbar$ at $\mu$, which gives an
effective limit $[L-D(\mu)]_+$; the positive part matters once $\mu$ lies far enough outside
$[98.5,101.5]$ that $D(\mu)>L$. With this effective limit Proposition~\ref{prop:chi} is exact again.
Its departure from the exact fixed-location spread-only width is the mean channel alone, the
difference of $g-g_{\mathrm{eff}}$ at the two $\sigma$-quantiles, and it is small. The second is what
the comparison below actually computes, and it is a different object: the chi-square width entered at
$\hat q=F_{\nu_t}^{-1}(\hat g)$ with $\hat g$ the \emph{full} wedge plug-in. It therefore also carries
the calibration term relating the two arguments. On IJPX-F12-cont the frozen-mean width is $0.29259$
and the exact fixed-location spread-only width $0.30123$, a difference of $0.00864$. The
plug-in-matched chi-square width is $0.38255$, a difference of $0.08132$. On the ninety-eight lots of
\S\ref{sec:reallots}, evaluated with \eqref{eq:g1d}, the plug-in-matched law reproduces the exact
$\sigma$-only width to a median of $10^{-6}$ and a maximum of $0.081$. The four lots above $0.05$ all
have means at least four and a half units off target. The sixty-two lots inside the dead zone
($|\bar x-100|\le1.5$) reach at most $0.038$, with a median indistinguishable from zero. Wherever their
residual is not negligible it has the opposite sign, and there the test mean straddles the
dead-zone edge at $\sigma\approx3$--$7$.
Neither comparison bounds the error in the \emph{lower endpoint} of the full posterior, which is
the quantity a release decision consumes; \S\ref{sec:endpoint} reports that error directly.

\subsection{Large-sample forms and upper endpoints}\label{app:asym}

Accuracy of the large-sample form \eqref{eq:probitlaw} of \S\ref{sec:closed}. The leading term is not accurate at the sample sizes practice uses; the
$\nu^{-1/2}$ term is (Table~\ref{tab:asym}).

\begin{table}[htbp]\centering\small
\caption{Maximum error of the large-sample forms against the exact law \eqref{eq:chilaw} over
$z_{\hat g}\in[-3,3]$: the probit term of \eqref{eq:probitlaw} alone, with its $\nu^{-1/2}$
correction, and with both chi-squares replaced by their Wilson--Hilferty cube-root forms.}
\label{tab:asym}
\begin{tabular}{@{}ccccc@{}}
\toprule
$\nu_e$ & $\nu_t$ & probit only & with $\nu^{-1/2}$ term & Wilson--Hilferty form\\
\midrule
$9$ & $9$ & $0.373$ & $0.010$ & $0.008$\\
$29$ & $9$ & $0.118$ & $0.014$ & $0.004$\\
$29$ & $29$ & $0.198$ & $0.004$ & $0.003$\\
$99$ & $99$ & $0.103$ & $0.001$ & $0.001$\\
\bottomrule
\end{tabular}
\end{table}

On the ninety-eight lots the
probit term alone misses by up to $0.369$, twenty-five lots by more than $0.05$; with the
correction the largest miss is $0.023$ and none exceeds $0.05$.

\paragraph{Upper endpoints of Table~\ref{tab:chilaw}.} The main article tabulates only the lower
endpoint, the one a release decision consumes. Table~\ref{tab:chilawhi} completes the interval
from the same proposition at the same grid; the width of any entry is its upper endpoint here minus
its lower endpoint there.

\begin{table}[h]\centering\small
\caption{Upper endpoint of the $95\%$ interval on $\Pacc$ for the pure-spread rule of
Proposition~\ref{prop:chi}, test size $n_t=10$, at the grid of Table~\ref{tab:chilaw}. Entries
above $0.9995$ are shown as $>0.999$.}
\label{tab:chilawhi}
\begin{tabular}{@{}lccccc@{}}
\toprule
$\hat g$ & $m=10$ & $m=20$ & $m=30$ & $m=50$ & $m=100$\\
\midrule
$0.5$ & $0.960$ & $0.892$ & $0.844$ & $0.784$ & $0.712$\\
$0.8$ & $0.998$ & $0.988$ & $0.977$ & $0.959$ & $0.931$\\
$0.9$ & $>0.999$ & $0.997$ & $0.994$ & $0.988$ & $0.975$\\
$0.95$ & $>0.999$ & $0.999$ & $0.998$ & $0.996$ & $0.991$\\
$0.99$ & $>0.999$ & $>0.999$ & $>0.999$ & $>0.999$ & $0.999$\\
$0.999$ & $>0.999$ & $>0.999$ & $>0.999$ & $>0.999$ & $>0.999$\\
$0.9999$ & $>0.999$ & $>0.999$ & $>0.999$ & $>0.999$ & $>0.999$\\
$0.99999$ & $>0.999$ & $>0.999$ & $>0.999$ & $>0.999$ & $>0.999$\\
\bottomrule
\end{tabular}
\end{table}

\section{The three treatments at a common operating point}\label{app:three}

Section~\ref{sec:three} defines the three ways of turning uncertainty about $\theta$ into a
statement about $g$. This appendix carries all three through the same operating point with every
intermediate quantity shown, so that each can be checked by hand and, more to the point, so that what each one \emph{assumes} can be seen doing its work.

\subsection{The common setting}

The rule is USP $\langle905\rangle$ Stage 2: $n=30$, $k=2$, acceptance value
$\mathrm{AV}=|M-\Xbar|+ks$ with $M=\operatorname{clip}(\Xbar,98.5,101.5)$. The parent is Normal.
The limit is calibrated to $L=11.320582$ so that the point estimate is exactly the baseline,
\begin{equation}
\hat\theta=(\hat\mu,\hat\sigma,\hat\rho)=(100,\,5,\,0),\qquad
g(\hat\theta)=0.850000 .
\end{equation}
The parameter law behind a calibration sample of $M_{\mathrm{cal}}=30$ units with summaries
$(\bar x,s)=(100,5)$ is the posterior \eqref{eq:cspost} of \S\ref{sec:905}, restated here in
the coordinates the treatments below use:
\begin{equation}\label{eq:applaw}
v=\sigma^2(1-\rho)=\frac{29\,s^2}{U},\ U\sim\chi^2_{29};\qquad
\mu\mid v,\rho\sim N\!\Big(\bar x,\;v\Big[\tfrac1{30}+\tfrac{\rho}{1-\rho}\Big]\Big);\qquad
\rho\sim\mathrm{U}[0,\,0.30],
\end{equation}
under the reference prior $\pi(\mu,v,\rho)\propto v^{-1}$ on $(\mu,v)$ times the Uniform on
$\rho$. In the coordinates $\theta=(\mu,\sigma,\rho)$ with $\sigma=\sqrt{v/(1-\rho)}$ the
total spread, its moments are $\bar\theta=(100,\,5.591,\,0.150)$,
$\mathrm{sd}(\mu)=2.443$, $\mathrm{sd}(\sigma)=0.818$, $\mathrm{sd}(\rho)=0.086603$ and
$\operatorname{Cov}(\sigma,\rho)=0.02486$, the covariance arising because the total spread
$\sigma=\sqrt v/\sqrt{1-\rho}$ rises with $\rho$ at fixed within-lot variance. Every
treatment below uses \emph{this same law}; they differ only in what they do with it.

\subsection{Treatment 1: global (pushforward quantiles)}\label{app:global}

\emph{Definition.} Evaluate $g$ at many $\theta$ drawn from or gridded over \eqref{eq:applaw}
and read the quantiles of the induced law $g_\#\pi$. Every evaluation is exact, so the only
approximation is in the outer integral over $\theta$.

\emph{Numbers.} Four independent scrambled Sobol' sets of $2^{16}$ points of
\eqref{eq:applaw}, each pushed through the one-dimensional evaluator, give mean $0.53254$,
median $0.54436$ and a $95\%$ interval $[0.03532,\ 0.98309]$, width $0.94777$; the standard
errors over the four scrambles are $6\times10^{-5}$ on the lower and $1.3\times10^{-4}$ on the
upper endpoint, and below $10^{-5}$ on the mean. The same-lot forecast, further units of the lot
already sampled with $\alpha\mid v\sim N(\bar x,v/30)$ and no fresh shared factor, gives mean
$0.71352$ and $[0.09357,\ 0.99869]$, width $0.90512$. A stipulated law of three independent marginals instead
(Normal in $\mu$ with the independence variance, Normal in $\sigma$, Uniform in $\rho$), which is
not a posterior, gives mean $0.56678$ and $[0.08817,\ 0.99231]$: its lower endpoint sits $0.053$
above the posterior's, so the choice of law is not harmless at $M_{\mathrm{cal}}=30$.

\emph{Why the outer integral is replicated.} A tail quantile is determined by the evaluated
atoms in the tail, so a scheme that settles the mean to a fraction of a percent can leave a
$2.5\%$ quantile moving by a tenth of its value. On a $245$-node product grid the mean
settles to $0.6\%$ across four refinements while the lower endpoint moves $13\%$, and a
$1000$-draw Monte Carlo sample carries a standard error of $0.011$ on that endpoint. Replicated randomized quasi-Monte Carlo removes the guesswork: the endpoint's own
error is measured from the scrambles, and every endpoint in the article is quoted with it.

One convention has to be recorded. Equation~\eqref{eq:pushforward} is the quantile of the
continuous law $g_\#\pi$, but a finite sample supplies only a discrete measure on the evaluated
atoms. Every endpoint in this paper is read from that measure by linear interpolation
between adjacent order statistics. This reading is consistent for \eqref{eq:pushforward} and less
biased in the tails than the attained value $\min\{v:\hat F(v)\ge\beta\}$. The two differ by at
most one inter-atom spacing, below $10^{-4}$ at $2^{16}$ points.

\emph{What the assumption is doing.} Nothing is assumed about the shape of $g_\#\pi$; the
asymmetry that defeats Treatment 2 below is simply reported. The cost is $2^{16}$ exact
evaluations per scramble, a few seconds, rather than the seven that Treatment 2 needs for a
three-coordinate gradient.

\subsection{Treatment 2: local (delta method)}\label{app:delta}

\emph{Definition.} Replace $g$ near $\hat\theta$ by its tangent plane,
\begin{equation}\label{eq:applin}
g(\theta)\;\approx\;g(\hat\theta)+\nabla g(\hat\theta)^{\top}(\theta-\hat\theta),
\end{equation}
so that the law of $g(\theta)$ is approximated by the linear image of $\pi$. Its variance is
$\nabla g^{\top}\Sigma_\theta\nabla g$, where $\Sigma_\theta$ is the full covariance of
\eqref{eq:applaw} in the coordinates $(\mu,\sigma,\rho)$. That covariance is not diagonal; its
$(\sigma,\rho)$ entry is $0.02486$.
Treating that image as Normal converts the standard deviation
$\mathrm{sd}=\sqrt{\nabla g^{\top}\Sigma_\theta\nabla g}$ into the symmetric band
$\hat g\pm z\,\mathrm{sd}$, with $z=1.96$ at the $95\%$ level. Three assumptions are stacked
here: $g$ is linear over the support of $\pi$; the resulting law is Normal in shape; and $\pi$
enters only through a center and a covariance. For the $\rho$-coordinate, whose law
$\mathrm{U}[0,0.30]$ has mean $0.15\ne\hat\rho=0$, centering at $\hat\theta$ is itself
already an approximation.

The gradient requires no analytic derivative: each coordinate is a finite difference of the
exact evaluator, central where the parameter can move in both directions,
\begin{equation}\label{eq:appfd}
\frac{\partial g}{\partial\mu}\;\approx\;
\frac{g(\hat\mu+h,\hat\sigma,\hat\rho)-g(\hat\mu-h,\hat\sigma,\hat\rho)}{2h},
\qquad h=0.005,
\end{equation}
and likewise in $\sigma$; one-sided in $\rho$,
$\partial g/\partial\rho\approx[g(\hat\mu,\hat\sigma,h_\rho)-\hat g]/h_\rho$, because
$\hat\rho=0$ sits on the boundary of the parameter's range. The one-sided quotient converges
only like $h_\rho$. The slow convergence comes from the one-sidedness itself, not from a lack of smoothness in $g$.
The latent representation carries $\sqrt\rho$, but the Normal probability depends on $\rho$ only
through positive variances that vary smoothly there. So $g$ is smooth at $\rho=0$, and the
boundary only forces the derivative to be one-sided, taken inside the declared parameter range. The
ladder is accordingly read as finite-difference convergence:
$h_\rho=0.02,\,0.01,\,0.005,\,0.001,\,10^{-4},\,10^{-5}$ reads
$0.2207,\,0.2777,\,0.3129,\,0.3449,\,0.3526,\,0.3534$, with limit $0.35350$ by Richardson
extrapolation. Stopping at $h_\rho=0.02$ would report a derivative more than a third too small, a bias of the
difference quotient and not of the evaluator; the one-dimensional evaluator is exact to $10^{-14}$, so nothing prevents the step from being taken
small. The construction costs six evaluations beyond $\hat g$ itself, plus the ladder.

\emph{Numbers.} At $\hat\theta$,
\begin{equation}\label{eq:appgrad}
\frac{\partial g}{\partial\mu}=0\ \text{(by symmetry)},\qquad
\frac{\partial g}{\partial\sigma}=-0.389061,\qquad
\frac{\partial g}{\partial\rho}=+0.35350 .
\end{equation}
With the full covariance,
\begin{equation}
\begin{split}
\mathrm{sd}^2&=(-0.389061)^2(0.818)^2+(0.353500)^2(0.086603)^2
+2(-0.389061)(0.353500)(0.02486)\\
&=0.101177+0.000937-0.006838,\qquad \mathrm{sd}=0.308668,
\end{split}
\end{equation}
giving $0.850\pm1.96(0.308668)=[0.24502,\ 1.45498]$; dropping the cross term gives
$0.319553$ and $[0.2237,\ 1.4763]$.

\emph{What the assumption is doing.} Three things are visible here that no summary conveys.
First, the upper endpoint is $1.455$: a linear image of a symmetric law is
symmetric, and $g$ is not, because it saturates at $1$.

The obvious repair is to clip the band to $[0,1]$, and it is worth seeing why that is not one.
Unclipped, the band spans $1.210$ against the exact $0.948$: too wide, though a wider band
still misses most of the lower tail. Clipped, it becomes $[0.245,\ 1.000]$, of width $0.755$
--- now \emph{narrower} than the interval it is approximating. Clipping does not repair the band; it reverses the
direction of its error, trading a visible impossibility for an invisible under-coverage. Nor
does it touch the endpoint a decision actually consumes: the lower end stands at $0.245$ against
an exact $0.035$, so the downside risk is understated sevenfold before and after. The excursion
past one is not the defect. It is the only outward sign that the linearization has failed, and
clipping removes the sign while keeping the failure. Second, and more instructive,
$\partial g/\partial\mu$ is \emph{exactly zero}, by symmetry. That is not an error: the rule is
centered, so the operating characteristic is at its maximum in $\mu$ and the dependence there is
second-order. A first-order method therefore assigns location uncertainty \emph{no} contribution
at all, when Appendix~\ref{app:decomp} finds that location in fact contributes $0.159$ of interval
width. The delta method is blind precisely where the curvature lives. Third, the band is
centered at $\hat g$ although the posterior mean is $\bar\theta=(100,\,5.591,\,0.150)$, not
$\hat\theta$. Three centers must be kept apart here. The scale's posterior mean sits above the plug-in and
the correlation's sits above zero, so to first order the offset moves the center by
$\nabla g\cdot(\bar\theta-\hat\theta)=-0.230+0.053=-0.177$. That puts the center at
$0.850-0.177=0.673$ and the band, at the same half-width $0.605$, at $[0.068,\,1.278]$.
Evaluating $g$ at the posterior mean instead gives $g(\bar\theta)=0.614$ and the band
$[0.010,\,1.219]$. Neither is the posterior mean of $g$, which is $\E[g(\theta)]=0.533$; a
first-order shift and the value at the shifted point are different objects, and re-expanding at
a new center would require its own gradient. Under either re-centering the ceiling still exceeds
one, so the expansion point is not the defect; curvature, the boundary and the non-Gaussian law
are. Each failure is one of the Definition's three stacked assumptions, caught in the act.

\subsection{Treatment 3: worst-case (corner) band}\label{app:corner}

\emph{Definition.} Over a confidence set $\Theta_\alpha$, report
$[\min_{\Theta_\alpha}g,\ \max_{\Theta_\alpha}g]$. The set used here is one $95\%$ box per
face $\rho\in\{0,\,0.30\}$: $\mu$ within $\bar x\pm1.96\,s\sqrt{1/30+\rho/(1-\rho)}$ and the
total spread $\sigma$ within the central $95\%$ of its inverse-chi-square law at that $\rho$,
$[98.211,\,101.789]\times[3.982,\,6.722]$ at $\rho=0$ and $[93.340,\,106.660]\times
[4.759,\,8.034]$ at $\rho=0.30$. If $g$ were monotone in each coordinate the extremes would be
attained at corners and eight evaluations would suffice.

\emph{Numbers.} The eight corners give $[0.018436,\ 0.991869]$. A genuine search over the
set, a $21\times21$ grid in $(\sigma,\rho)$ with a bounded one-dimensional optimization in
$\mu$ on each line and a simplex polish of the result, gives $[0.018436,\ 0.999017]$: the
minimum is a corner, $(93.340,\,8.034,\,0.30)$, and the maximum is the interior point
$(100,\,3.982,\,0)$, the target mean on the low-spread face, which the corners miss by
$0.00715$. A search that scans only $\rho$ with $\mu$ and $\sigma$ held at their corners would call the
extremes attained, since it presumes the monotonicity it would need to establish.

\emph{What the assumption is doing, and two traps.} This band is \emph{not} a $95\%$ interval
and must not be read as one. It is the range of $g$ over a set, so it is properly wider at the
bottom than the pushforward's $2.5\%$ quantile, $0.018$ against $0.035$. The set's own
posterior probability is a separate quantity. Consider a box built from the pooled marginal $95\%$
intervals of $\mu$ and $\sigma$, $[94.977,\,105.023]\times[4.244,\,7.434]\times[0,0.3]$.
It carries posterior probability $0.905$ under \eqref{eq:applaw}, close to $0.95^2$ since the two
coordinates are nearly independent. Its $\mu$-interval pools the $\rho$-inflated
tails, so its corners all sit far from target, and its corner maximum $0.616$ misses the true
$0.994$ over that box by $0.38$. Treatment 3 answers ``how bad could it be, anywhere in the
plausible set?''; Treatment 1 answers ``where does the probability actually sit?''.
Section~\ref{sec:corner} records where the corner shortcut attains the true extremes and where
it does not.

\subsection{Side by side}

\begin{table}[h]\centering\small
\caption{The three treatments plus the two link repairs, all at the $\langle905\rangle$ Stage-2
operating point, all from the same $\hat\theta$ and the same posterior \eqref{eq:applaw}. ``Cost'' counts
exact evaluations of $g$.}
\label{tab:appthree}
\begin{tabular}{@{}llccl@{}}
\toprule
treatment & interval & width & cost & answers\\
\midrule
\emph{1. pushforward, new lot} & $[0.0353,\ 0.9831]$ & $0.948$ & $4\times2^{16}$ &
\emph{the deliverable}\\
1. pushforward, same lot & $[0.0936,\ 0.9987]$ & $0.905$ & $4\times2^{16}$ & a different forecast\\
2. delta, raw        & $[0.2450,\ 1.4550]$ & $1.210$ & $7$    & not a probability\\
\quad + logit link   & $[0.0470,\ 0.9985]$ & $0.952$ & $7$    & floor too high by $0.01$\\
\quad + probit link  & $[0.0596,\ 0.9999]$ & $0.940$ & $7$    & floor too high by $0.02$\\
3. corner band       & $[0.0184,\ 0.9919]$ & $0.973$ & $8$    & worst case; misses the maximum\\
3. true range over the box & $[0.0184,\ 0.9990]$ & $0.981$ & search & worst case, not an interval\\
\bottomrule
\end{tabular}
\end{table}

Table~\ref{tab:appthree} is the whole argument of the paper compressed to one operating point.
The reported number is $0.850$. What the data actually entitle anyone to say is
$[0.04,\ 0.98]$. Every shortcut either leaves the unit interval, misplaces the floor, or
answers a different question. The two link repairs are the closest, overstating the floor by
$0.01$ to $0.02$; the raw delta band overstates it by $0.21$ while returning an upper endpoint
above one, which is the error that matters, since a floor quoted too high is a risk statement
in the unsafe direction. The only way to learn which is happening is to compute the posterior
interval once.

\subsection{The corner and delta bands of \S\ref{sec:corner}, in full}\label{app:cornerdetail}

\paragraph{The corner band.} Treatment 3 of \S\ref{sec:three} evaluates $g$ at the corners of a box confidence set. That
shortcut is valid only if $g$ is monotone in each coordinate, and at the $\langle905\rangle$
operating point it is not monotone in the mean. The rule is symmetric about target, so for any
spread the response to $\mu$ peaks at $\mu=100$. Whenever the box contains the target, that
peak is an interior point of the mean interval, so the maximum of $g$ over the box sits on the
low-spread face at the target mean and not at a corner. With the box of Appendix~\ref{app:corner}, one $95\%$ face
per value of $\rho$, the eight corners give $[0.018,\,0.992]$ while a genuine search over the
box gives $[0.018,\,0.999]$: the minimum is attained at a corner, the wide-spread far-mean
corner of the $\rho=0.3$ face, and the maximum is missed by $0.007$, seven hundred times the
evaluator's resolution. A search that holds $\mu$ and $\sigma$ at their corners and scans only $\rho$ cannot detect
this miss, since it presumes the property it would need to establish. The repair is cheap: for a symmetric
Normal wedge evaluate the closest admissible mean to target for the maximum and the farthest
for the minimum on every $(\sigma,\rho)$ face, and for a general parent either prove the
relevant unimodality or search the box. A box built from marginal $95\%$ intervals is, moreover, not a $95\%$ set
(Appendix~\ref{app:corner}). For a Normal parent the corner values have a closed \emph{approximation}:
\eqref{eq:chilaw} with $\hat q$ replaced by
$\hat q_{\mathrm{eff}}(\mu_c)=\nu_t\,[L-D(\mu_c)]_+^2/(k^2s^2)$ at
$\mu_c=\bar x\pm t_{\nu_e,0.975}\,s/\sqrt m$, chi-square in $\sigma$ and Student $t$ in $\mu$.
This is Bergum's construction in the spread coordinate with the mean channel frozen. Against
the exact corner values on the ninety-eight lots it agrees to a median of $3\times10^{-5}$ and
a maximum of $0.068$, with four lots above $0.05$. The residual again comes from the mean
channel at the wide-$\sigma$ corner, so the approximation carries a measured residual and is
not a closed form.  The pushforward remains the right default when the shape of the interval, and
not only its extremes, is wanted; it is also what Lewis and Fan \cite{lewis2016} set against
the E2810 confidence-set construction, and the two answer different questions, a posterior
probability against a worst case over a set whose coverage must be stated separately. On the family's simplest member the failure is of the other kind: on a $k=1$ tent with location and spread well determined and only the dependence parameter uncertain, the response to $\rho$ is non-monotone, the true minimum is interior at $\rho^*\approx0.15$, and the corner band misses it by $0.007$.

\paragraph{The delta band.} Table~\ref{tab:link} separates the two failures of the naive band. The \emph{range} violation
is gross at both points: the upper endpoint is $1.381$ on the exposure study and $1.455$ at
$\langle905\rangle$. Clipping to $[0,1]$ removes the symptom at the top, where it costs
nothing, but leaves the floor untouched. The reported floor stays at $0.245$ against a true
$0.035$ at $\langle905\rangle$, and at $0.319$ against $0.085$ on the exposure study.
Re-centering the band at the posterior mean instead of the plug-in, as the first-order shift $\nabla g\cdot(\bar\theta-\hat\theta)=-0.18$ recommends, does not repair it either. The floor moves to $0.010$, below the truth by $0.026$, and the ceiling stays at $1.22$. The defects are therefore curvature, the boundary and a non-Gaussian parameter law. They are not the expansion point, and not position as such, since the gradient depends on position. The
\emph{accuracy} of the link repair is the surprise. On the exposure study the logit band
reproduces the exact endpoints to a median of $0.004$ and $0.002$, its width within $1\%$ of the
truth at the median; at $\langle905\rangle$ the same construction puts the lower endpoint at
$0.047$ against an exact $0.035$ and the upper at $0.998$ against $0.983$, an overshoot of about
a hundredth at each end, resolved against a standard error below $2\times10^{-4}$. Two traps
manufacture spurious failures: a delta standard deviation built from marginal variances alone
drops the term $2g_\sigma g_\rho\operatorname{Cov}(\sigma,\rho)$, which under
\eqref{eq:cspost} is not zero, and a forward difference in $\rho$ at the boundary converges
only like its step, reporting $0.22$ at $h=0.02$ for a derivative of $0.35$
(Appendix~\ref{app:delta}).

What survives is weaker. The link repairs the range violation by construction; its accuracy is
an observation about two points, one of which has three uncertain coordinates and the non-logistic
pushforward of Figure~\ref{fig:push}. There is no criterion, only two successes, which is a reason
to verify the shortcut once per rule, at the cost of the exact evaluation it was meant to
avoid.

\section{The anchor, drawn}\label{app:anchor}

The analytic anchor (\S\ref{sec:anchor}) is the only operating point in the paper at which the acceptance
probability, its entire operating characteristic, and its amplification are \emph{all} available
in closed form. That makes it the one place where every object the paper computes can be shown
against a known truth rather than asserted, which is what Figure~\ref{fig:anchor} does.

\begin{figure}[tbp]\centering
\includegraphics[width=\linewidth]{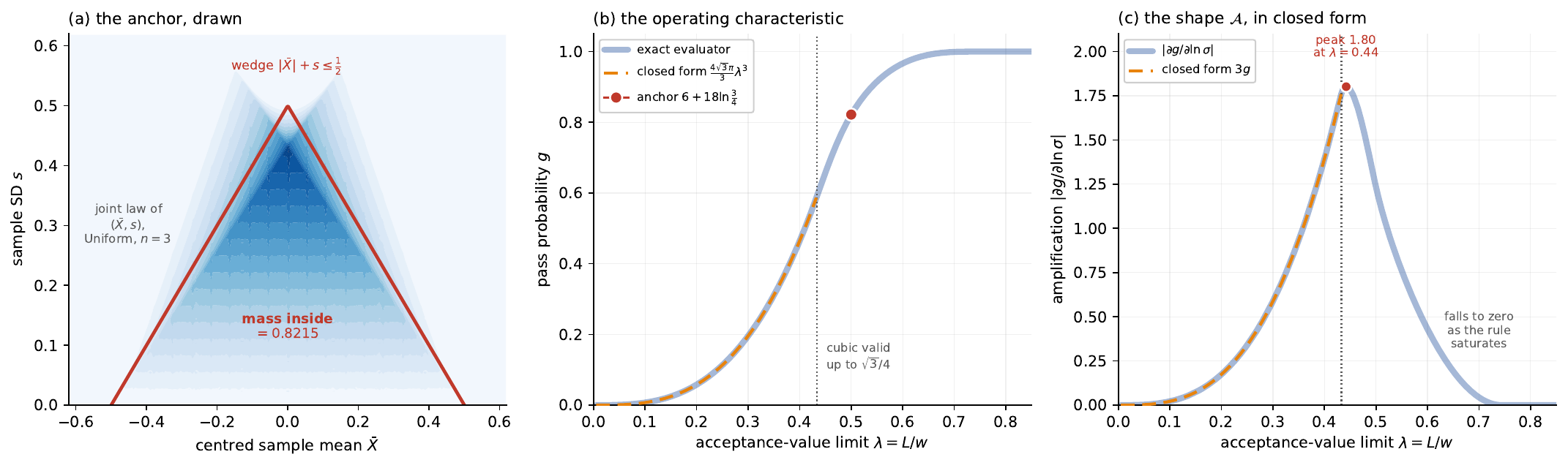}
\caption{The $n=3$ Uniform anchor, where the truth is known. \emph{(a)} The joint law of
$(\Xbar,s)$ for a Uniform parent at $n=3$, with the wedge $|\Xbar|+s\le\frac12$ drawn on top of
it. The acceptance probability is nothing more exotic than the mass inside the triangle. The
support's own boundary is visible behind it: at $n=3$ a large spread forces the mean towards the
center, which is why the density tapers to a point rather than filling a rectangle.
\emph{(b)} The operating characteristic (pass probability against the acceptance-value limit $\lambda=L/w$) and \emph{(c)} the amplification $|\partial g/\partial\ln\sigma|$. In
both, the broad pale band is the exact evaluator and the orange dashed curve is the closed form
($\frac{4\sqrt3\pi}{3}\lambda^{3}$ in (b), $3g$ in (c)), drawn only over
$\lambda\le\sqrt3/4$ where it is valid; the dotted vertical marks that join, and the red dot in
(b) is the anchor $6+18\ln\frac34$. The two agree to plotting accuracy wherever the closed form
applies, and the closed form is visibly \emph{not} the whole story past the join. Panel (c) is
the shape $\mathcal{A}$ of \eqref{eq:ampscale} exhibited in closed form; every other study in the
paper estimates it numerically.}
\label{fig:anchor}
\end{figure}

Four things are worth reading off it. First, panel (a) shows that the pair $(\Xbar,s)$ does not
range over a rectangle: the mean and the spread of a sample constrain one another. The reason is
that each of the three measured values $X_1,X_2,X_3$ must itself lie in the parent's support,
here $[-\frac12,\frac12]$. A sample whose mean sits near one end of that interval must therefore
have all three values crowded near that end, and crowded values cannot be widely dispersed, so a lot cannot be both far off target and highly variable. The most dispersed sample available
splits the three values two-to-one between the ends of the support, $(-\frac12,-\frac12,\frac12)$ or its mirror image, giving
\begin{equation}\label{eq:supporttip}
s_{\max}=\tfrac{1}{\sqrt3}=0.577350\quad\text{at}\quad \Xbar=\mp\tfrac16,
\end{equation}
which is the pair of pale horns at the top of the panel. Those horns are precisely what the rule
throws away: their acceptance value is $\frac16+\frac{1}{\sqrt3}=0.744$, half as much again as
the limit $\frac12$. And directly above target the ceiling is lower still: with $\Xbar=0$ the most dispersed admissible sample is $(-\frac12,0,\frac12)$, whose $s$ is exactly $\frac12$. So at the anchor the wedge's apex sits \emph{exactly} on the support, touching it at one point.
The anchor is not an arbitrary choice of $\lambda$: it is the smallest limit at which an
on-target lot is guaranteed to pass, and raising it only enlarges the accepted set. That the rejected set is a thin sliver of extreme
configurations is why $g=0.82$ rather than something near a half.

Second, panel (b) shows \emph{why} the cubic must stop:
$\frac{4\sqrt3\pi}{3}\lambda^{3}$ tracks the evaluator indistinguishably up to
$\lambda=\sqrt3/4$ and would pass through one shortly after, whereas the true characteristic
turns over and saturates. The anchor at $\lambda=\frac12$ sits just past the join, which is why
its value is the transcendental $6+18\ln\frac34$ rather than a multiple of $\pi$.

Third, panel (c) is the paper's central object with no fitting anywhere in it. On the cubic
range $g\propto\lambda^{3}$, so the elasticity is exactly $3$ and the amplification is exactly
$3g$; the dashed and solid curves coincide to plotting accuracy. Past the join the amplification
peaks and then falls to zero as the rule saturates, which is the same statement as the
humped interval width of \S\ref{sec:reallots}, seen one level down. A reader who wants the
intuition behind ``the interval collapses at both ends'' can get it here, from a curve that is
elementary.

Fourth, the figure doubles as a verification. All three panels come from a single cf/FFT inversion: the joint density of $(\Xbar,s)$ depends on the parent and $n$ only, so sweeping the limit only re-masks one grid. That one inversion reproduces, independently of \S\ref{sec:anchor},
every number quoted there: admitted-cell mass $0.999771$ against a whole-grid mass of $1$ (the
deficit is an inversion and masking effect at the cone boundary, not truncation;
Appendix~\ref{app:engine}); the anchor at $0.821489$ against
the closed-form $0.821723$, a grid error of $2.3\times10^{-4}$; the cubic at $\lambda=\sqrt3/4$
giving $0.5889$ against a closed-form $0.5890$, hence $3g=1.767$; and a peak amplification of
$1.801$ at $\lambda=0.442$.

\subsection{The scale reduction, step by step}\label{app:anchorscale}

Section~\ref{sec:anchor} asserts that for a Uniform parent at $n=3$ the acceptance probability
depends on $\sigma$ only through the ratio $\lambda=L/w$, and that the amplification is exactly
$3G$ on the cubic range. Both statements are a few lines of algebra, given here in full.

\emph{Step 1: the width carries the scale.} A Uniform parent centered on the target $T$ with
support width $w$ is the law of $T+wU$ with $U\sim\mathrm{U}[-\tfrac12,\tfrac12]$. Its variance
is $w^{2}\Var U=w^{2}/12$, so
\begin{equation}
\sigma=\frac{w}{\sqrt{12}},\qquad\text{equivalently}\qquad w=\sigma\sqrt{12}:
\end{equation}
for this family the standard deviation and the support width are the same parameter in
different units, and varying $\sigma$ at fixed mean is varying $w$.

\emph{Step 2: the pair $(\Xbar,s)$ inherits the scale linearly.} Write each observation as
$X_i=T+wU_i$. The sample mean is linear, and the sample standard deviation is
location-invariant and scale-equivariant, so
\begin{equation}
\Xbar-T=w\,\bar U,\qquad s(X_1,X_2,X_3)=w\,s(U_1,U_2,U_3),
\end{equation}
and the joint law of $\big(\bar U,\,s(U)\big)$ contains no parameter at all.

\emph{Step 3: the acceptance event divides through by $w$.} The wedge accepts when
$|\Xbar-T|+ks\le L$, that is, when $w\big(|\bar U|+k\,s(U)\big)\le L$. Dividing by $w$,
\begin{equation}
g(\sigma)=\Prob\Big[\,|\bar U|+k\,s(U)\le\frac{L}{w}\,\Big]=:G(\lambda),
\qquad\lambda=\frac{L}{w}=\frac{L}{\sigma\sqrt{12}}\,.
\end{equation}
Everything $\sigma$ can do to the acceptance probability, it does through the single number
$\lambda$; the function $G$ is a fixed property of the Uniform triple, computed once.

\emph{Step 4: the amplification is an elasticity.} Since
$\ln\lambda=\ln L-\ln\sigma-\ln\sqrt{12}$, a relative change in $\sigma$ is a relative change
in $\lambda$ with the sign flipped: $d\lambda/d\ln\sigma=-\lambda$. By the chain rule,
\begin{equation}
\frac{dG}{d\ln\sigma}=-\lambda\,G'(\lambda),
\qquad\text{hence}\qquad
\Big|\frac{dG}{d\ln\sigma}\Big|=\lambda\,G'(\lambda),
\end{equation}
the sensitivity of $G$ to the log-scale (the elasticity proper would divide by $G$). This is also why $\ln\sigma$ is the natural coordinate
throughout the paper: posterior uncertainty in a scale parameter is multiplicative, and the
amplification that multiplies it is dimension-free.

\emph{Step 5: on the cubic range the elasticity is the exponent.} For $\lambda\le\sqrt3/4$
the exact-evaluation paper \cite{paper1} gives $G=\tfrac{4\sqrt3\pi}{3}\lambda^{3}$, a pure power, and the
elasticity of a power law is its exponent: $\lambda G'/G=3$, which is \eqref{eq:amp3G},
$|dG/d\ln\sigma|=3G$ exactly. At the end of the range the quoted numbers are themselves closed
forms:
\begin{equation}
G\big(\tfrac{\sqrt3}{4}\big)=\frac{4\sqrt3\pi}{3}\cdot\frac{3\sqrt3}{64}
=\frac{3\pi}{16}=0.589049,
\qquad
\Big|\frac{dG}{d\ln\sigma}\Big|=\frac{9\pi}{16}=1.767146 .
\end{equation}

\emph{Step 6: the exact $G$, and how far it is known.} The reduction turns ``an exact
expression for $g(\sigma)$'' into a question about the single function $G$, and the exact-evaluation
paper answers it in closed form over four fifths of the probability range \cite{paper1}. On the
cubic range, substituting $\lambda=L/(\sigma\sqrt{12})$ back gives the exact law in the
original variables,
\begin{equation}\label{eq:gsigma}
g(\sigma)=\frac{4\sqrt3\pi}{3}\Big(\frac{L}{\sigma\sqrt{12}}\Big)^{3}
=\frac{\pi L^{3}}{18\,\sigma^{3}},
\qquad \sigma\ge\frac{2L}{3},
\end{equation}
an inverse-cube law in the standard deviation, from which $|dg/d\ln\sigma|=3g$ is immediate a
second time. Beyond the cubic the closed form continues but stops being compact: on
$\sqrt3/4<\lambda<\tfrac12$ the acceptance region's sector integral splits into three explicit
regimes \cite{paper1}, and $G$ is a polynomial in $\lambda$ together with four evaluations of one
elementary primitive. Written out, with
\begin{equation}\label{eq:band2defs}
e=\tfrac12-\lambda,\quad c=\tfrac12+\lambda,\quad
v_1=\frac{\sqrt3\,e}{2-\sqrt3},\quad v_2=\frac{\sqrt3\,c}{2+\sqrt3},\quad
p(v)=\frac{\sqrt3\,(e+v)}{2v},\quad q(v)=\frac{\sqrt3\,(c-v)}{2v},
\end{equation}
and the primitives of $\big(u\arcsin u+\sqrt{1-u^{2}}\,\big)(2u\pm\sqrt3)^{-4}$,
\begin{align}
\mathcal M_-(u)&=\tfrac1{24}\bigg[\frac{2\sqrt{1-u^{2}}\,(5+4u(u-\sqrt3))-(\sqrt3-6u)\arcsin u}{(\sqrt3-2u)^{3}}\notag\\
&\qquad{}+\sqrt3\log\big|\sqrt3-2u\big|-\sqrt3\log\big|2-\sqrt3u+\sqrt{1-u^{2}}\big|\bigg],\label{eq:Mminus}\\
\mathcal M_+(u)&=\tfrac1{24}\bigg[\frac{-2\sqrt{1-u^{2}}\,(5+4u(u+\sqrt3))-(\sqrt3+6u)\arcsin u}{(\sqrt3+2u)^{3}}\notag\\
&\qquad{}-\sqrt3\log\big|\sqrt3+2u\big|+\sqrt3\log\big|2+\sqrt3u+\sqrt{1-u^{2}}\big|\bigg],\label{eq:Mplus}
\end{align}
the acceptance probability on this band is
\begin{equation}\label{eq:band2}
\begin{aligned}
G(\lambda)={}&8\sqrt3\pi\Big(\tfrac{\lambda v_1^{2}}{2}-\tfrac{v_1^{3}}{3}\Big)
+2\sqrt3\pi\Big(\tfrac{(3-2\lambda)(v_2^{2}-v_1^{2})}{2}+\tfrac{2(v_2^{3}-v_1^{3})}{3}\Big)\\
&-16\sqrt3\pi\Big(\tfrac{\lambda(\lambda^{2}-v_2^{2})}{2}-\tfrac{\lambda^{3}-v_2^{3}}{3}\Big)\\
&+288\sqrt3\,e^{3}\Big[\mathcal M_-\big(p(\lambda)\big)-\mathcal M_-\big(p(v_1)\big)\Big]
-288\sqrt3\,c^{3}\Big[\mathcal M_+\big(q(\lambda)\big)-\mathcal M_+\big(q(v_2)\big)\Big].
\end{aligned}
\end{equation}
Inverse sines, logarithms and radicals only: no dilogarithm and no elliptic period. The absolute
values matter, since $p(v_1)=q(v_2)=1$ and the arguments exceed $\sqrt3/2$ on this band. Evaluated
in $30$-digit arithmetic, \eqref{eq:band2} gives $0.6176932862$ at $\lambda=0.44$, $0.6581024105$
at $0.45$, $0.6967126054$ at $0.46$, $0.7657603842$ at $0.48$ and $0.8192418244$ at $0.499$,
agreeing to ten significant figures with the reference evaluation of \cite{paper1}; it joins the
cubic continuously at $\lambda=\sqrt3/4$, where both equal $3\pi/16$. Together with the cubic this gives $G$ explicitly on $0\le\lambda<\tfrac12$, up
to $G=0.8217$. At $\lambda=\tfrac12$, a
knife-edge at which one geometric cap never binds, the assembly collapses to the anchor
$6+18\ln\tfrac34$. The last band is $\tfrac12\le\lambda<\lambda_{\max}$, where $\lambda_{\max}=\tfrac16+\tfrac1{\sqrt3}=0.744017$ is
the acceptance value \eqref{eq:supporttip} of the most dispersed configuration; past it the
rule can no longer fail and $G\equiv1$. On this band the regime structure is no longer constant, and a
closed form has not been carried out. The exact evaluator takes over there, which is why the
paper treats the closed form as the anchor's explanation rather than its computational engine.
The endpoint of that band is nonetheless known exactly. The maximum is attained at the six
$2$--$1$ vertices of the support cube, where the statistic has a nonvanishing gradient, so the
level set is a simplex and $1-G(\lambda_{\max}-\delta)=(162\sqrt3-270)\,\delta^{3}(1+o(1))$
\cite{paper1}. The cubic vanishing is the reason the amplification $|dG/d\ln\sigma|$ of
\S\ref{sec:anchor} collapses so sharply as the rule saturates, and it is the same vertex
geometry that fixes the ceiling of the sample-variance law in \cite{n3companion}, with the
parent's local density exponents setting the power in each case.

The collapse to one ratio is what makes this study special, and it needs one condition:
the parent is a pure scale family about the target ($\mu$ pinned at $T$), so one
parameter moves one ratio. As soon as $\mu$ is uncertain as well, $g$ depends on the two ratios
$(\mu-T)/w$ and $L/w$ separately, and nothing divides out. One is then back to the general
pushforward of \S\ref{sec:pushforward}, which is precisely the situation of the other two studies.

\section{The compound-symmetry parameter law, derived}\label{app:cslaw}

Each unit of a lot is $X_i=\mu+\sqrt{\rho}\,\sigma Z+\sqrt{1-\rho}\,\sigma\varepsilon_i$ with
$Z$ common to the lot and the $\varepsilon_i$ independent standard Normal, so that
$\operatorname{corr}(X_i,X_j)=\rho$ and the calibration sample of $m$ units is exchangeable,
$X\sim N_m\big(\mu\mathbf 1,\ \sigma^2[(1-\rho)I+\rho J]\big)$. Compound symmetry keeps the two
summaries independent and fixes both their laws exactly:
\begin{equation}\label{eq:cs}
\Xbar\sim N\Big(\mu,\ \sigma^2\Big\{\rho+\frac{1-\rho}{m}\Big\}\Big),
\qquad
\frac{(m-1)S^2}{\sigma^2(1-\rho)}\sim\chi^2_{m-1},
\qquad \Xbar\perp S^2 .
\end{equation}
The first law follows from $\Xbar=\mu+\sqrt{\rho}\,\sigma Z+\sqrt{1-\rho}\,\sigma\bar\varepsilon$,
the second from $X_i-\Xbar=\sqrt{1-\rho}\,\sigma(\varepsilon_i-\bar\varepsilon)$, in which $Z$
cancels. From one lot $\rho$ is not identified: $\Xbar$ is a single draw from a Normal and
carries no information about its own variance. With $\sigma$ rescaled to keep the identified product $\sigma^2(1-\rho)$ equal to
$\hat\sigma^2$, the posterior spread of $\mu$ relative to its independence value
$\hat\sigma/\sqrt m$ is $\sqrt{1+m\rho/(1-\rho)}$. For $m=30$ this is $2.5$ at $\rho=0.15$ and $3.7$ at $\rho=0.30$,
and it is unbounded as $\rho\to1$, so the prior's upper limit on $\rho$ is load-bearing
for the $\mu$ contribution of Appendix~\ref{app:decomp}. With
$\sigma(\rho)=\hat\sigma/\sqrt{1-\rho}$ the within-lot spread of the test sample is
$\sigma(\rho)\sqrt{1-\rho}=\hat\sigma$ for every $\rho$ and dependence enters only through the
between-lot term $\hat\sigma\sqrt{\rho/(1-\rho)}$. Conditional on $Z$ the test sample is
independent Normal with mean $\mu+\hat\sigma\sqrt{\rho/(1-\rho)}\,Z$ and spread $\hat\sigma$, so
the dependent acceptance probability is the independent one averaged over $Z$ with standard
Normal weight, a single one-dimensional integral, which the evaluator performs exactly.

\paragraph{Per-row sampling in the exposure study.} The natural implementation of
Table~\ref{tab:cehd} draws a fresh posterior sample for every row, and it gives identical problems
(rows with the same $n$) different answers. Among the $393$ rows at $n=6$, per-row sampling produces widths from
$0.9015$ to $0.9437$ (standard deviation $0.0071$). This scatter is purely algorithmic, but it reads as
vertical spread at fixed $n$ and swamps any scale-invariance check read off such a table. The
table is therefore computed once per distinct $n$, to a standard error below
$2\times10^{-4}$, with the row counts entering as weights; the numbers move by at most $0.003$.
Scale invariance holds because the algebra says so, and the data can neither support nor refute
it.

\paragraph{Consequences of the posterior \eqref{eq:cspost}.} The last because one lot does not identify $\rho$ at all, so its law is a modeling choice in
the strict sense. Two consequences shape everything downstream. The calibration mean does
\emph{not} average the shared factor away, so the posterior spread of $\mu$ grows with $\rho$,
$2.5$ times the independence value at $\rho=0.15$ and $3.7$ times at $\rho=0.30$ for $m=30$,
and without bound as $\rho\to1$. And $S^2$ sees only the within-lot component, so the data
identify the \emph{product} $\sigma^2(1-\rho)$ and not $\sigma$ and $\rho$ separately. The
within-lot spread of the test sample is governed by $v$ at every $\rho$. Dependence enters
the future test only through the between-lot term $\sqrt{v\rho/(1-\rho)}$, which comes from a fresh shared
factor. The forecast is for a \emph{new lot} from the same process, because that is the question a release rule
puts to a process. A forecast for further units of the lot already sampled is a different
target, reported once in Table~\ref{tab:905} for contrast. The evaluator is the one-dimensional integral of
\S\ref{sec:closed} extended to the shared factor, exact to machine precision
(Appendix~\ref{app:cslaw}).

\section{Three real lots, in detail}\label{app:lots}

Section~\ref{sec:reallots} reports the ninety-eight lots in aggregate: the interval collapses at
both ends of the operating characteristic and opens in the middle. This appendix takes one lot
from each of those three regimes and walks it through, because the aggregate statement is easy
to read as a claim about \emph{how much was known} when it is really a claim about \emph{where
the lot sits}.

Each lot is published as a summary: how many units were assayed, their mean content, their
standard deviation. Taking that summary as the parameter estimate and asking for the probability
that a lot from such a process passes $\langle905\rangle$ Stage 1 gives the point estimate. The
interval asks the follow-up: the process was not observed, it was inferred from a handful of
units, so how far can that probability move across the other processes the same data support?

\begin{table}[h]\centering\footnotesize
\caption{One lot from each regime. $\sigma$ range is the central $95\%$ of the posterior
$\sigma^2\mid\text{data}\sim\text{Inv-}\chi^2(n-1,s^2)$; the two columns after it give the exact
acceptance probability at each end of that range. The last two columns are the per-lot
quantities that Table~\ref{tab:reallots} aggregates. Lot labels name the source and the row:
BJBMS-L$k$ is row $k$ of Table~1 in \cite{vranic2008}; EJPB-S$t$-$r$-$k$ is run $r$ of
supplementary Table~S$t$ in \cite{lyytikainen2024}, $k$ numbering the lots across S7 and S8;
IJPX-F$k$-batch/cont is formulation $k$ of Table~6 in \cite{ijpx2024} on the batch or the
continuous line; IJP537-B$k$ is batch $k$ of Table~4 in \cite{goodwin2018}.}
\label{tab:worked}
\setlength{\tabcolsep}{2.6pt}
\begin{tabular}{@{}lcccccccc@{}}
\toprule
& & & & point & \multicolumn{2}{c}{plausible $\sigma$} & $95\%$ interval & \\
\cmidrule(lr){6-7}
lot & $n$ & $s$ & RSD & est.\ & range & $g$ at each end & on $g$ & width\\
\midrule
BJBMS-L12 & $10$ & $1.779$ & $1.84\%$ & $1.0000$ & $[1.22,\,3.25]$ &
$1.0000,\ 0.9948$ & $[0.9932,\,1.0000]$ & $0.0068$\\
EJPB-S8-2-15 & $10$ & $4.845$ & $4.86\%$ & $0.8865$ & $[3.33,\,8.85]$ &
$0.9997,\ 0.0901$ & $[0.0659,\,0.9995]$ & $0.9336$\\
IJPX-F09-cont & $21$ & $19.600$ & $19.50\%$ & $1.45{\times}10^{-4}$ & $[15.00,\,28.30]$ &
$1.59{\times}10^{-3},\ 4.55{\times}10^{-6}$ & $[3.4{\times}10^{-6},\,1.4{\times}10^{-3}]$ & $0.0014$\\
\bottomrule
\end{tabular}
\end{table}

\paragraph{The clear pass.} BJBMS-L12 runs at $1.84\%$ RSD. Its point estimate is $1.0000$, and
from ten units the true $\sigma$ could be anything in $[1.22,\,3.25]$, a $2.65$-fold range.
At the tight end the lot passes with probability $1.0000$; at the loose end, $0.9948$. Every
process the data admit still passes, so the interval has width $0.0068$ and tells nobody
anything they did not already have.

\paragraph{The clear failure.} IJPX-F09-cont runs at $19.50\%$ RSD, four times the dispersion the
rule tolerates. Its point estimate is $1.45\times10^{-4}$ and stays between $1.59\times10^{-3}$ and
$4.55\times10^{-6}$ across a $[15.00,\,28.30]$ range for $\sigma$. Those three numbers span two
orders of magnitude and all round to $0.0000$, which is why the table above carries them in
scientific notation; on the scale of a release decision the verdict is immune to what was not
known.

\paragraph{The contested lot.} EJPB-S8-2-15 runs at $4.86\%$ RSD and reports
$0.8865$, an $89\%$ chance of passing, which is the kind of number that gets acted
on. It carries \emph{exactly the same} statistical ignorance as the first lot: also $n=10$, also
a $2.65$-fold range for $\sigma$, here $[3.33,\,8.85]$. But now the two ends of that range give
acceptance probabilities of $0.9997$ and $0.0898$. If the true standard deviation is $3.3$ the
lot is a certainty; if it is $8.8$ it fails nine times in ten; and the data cannot distinguish
those cases. The $95\%$ posterior interval is $[0.066,\,1.000]$, which spans ``do not bother'' to ``ship
it''.

The reported $0.8865$ is not a miscalculation. It is the exactly correct acceptance probability
for one particular process, the one whose standard deviation is precisely $4.845$, reported as though that process were known to be the true one. With the exact evaluator of \cite{paper1},
this is now the \emph{only} error left: no computational discrepancy of
unknown sign sits on top of it. The residual $0.94$ of interval width can therefore be attributed
squarely to ten units having been assayed instead of a hundred.

\paragraph{Why equal uncertainty has unequal consequence.} The two $n=10$ lots are the
instructive pair. Their posteriors for $\sigma$ are equally wide in relative terms --- the factor $2.65$ is a function of the sample size alone --- yet one interval is $0.0068$ across and the
other $0.9336$. The difference is not statistical but geometric: it comes from where each lot sits relative to the acceptance region. BJBMS-L12 sits far inside the
acceptance region, where the operating characteristic is flat at one and a factor of two in
$\sigma$ moves nothing. IJPX-F09-cont sits far outside it, where the characteristic is flat at
zero. EJPB-S8-2-15 sits on the shoulder between them, where the characteristic is steep, and
there the same factor of two in $\sigma$ is amplified into almost the whole unit interval. That
amplification is the quantity $|\partial g/\partial\ln\sigma|$ of \S\ref{sec:unify}, and this
table is what it looks like on three real lots.

The same contrast can be read off without any reference to a curve. For a lot with a known mean
the Stage-1 condition $|M-\Xbar|+2.4\,s\le15$ is a ceiling on the sample spread alone,
$s\le(15-|M-\Xbar|)/2.4$. BJBMS-L12 has $\Xbar=96.700$, which the reference-value clip of $\langle905\rangle$ carries to $M=98.500$.
It therefore spends $1.800$ of the budget, leaving a ceiling of $5.500$. Its plausible $\sigma$ reaches
only $3.248$, and a sample would have to come in some seventy percent above its own $\sigma$ to
breach the ceiling. EJPB-S8-2-15 is on target, spends nothing, and has the larger ceiling of
$6.250$ --- but its plausible $\sigma$ runs to $8.845$, so the ceiling lies \emph{inside} the
range and the admissible processes disagree about the outcome. The off-target lot is the safe
one: what matters is not how much of the budget a lot spends but whether what remains separates
it from the values of $\sigma$ its own data still permit.

\subsection{The ninety-eight lots by source}\label{app:bysource}

\begin{table}[h]\centering\small
\caption{The same lots by source, all at $\ntest=10$, so the column $\mest$ is now the evidence
alone and no longer doubles as the test. Two effects that the earlier
$\ntest=\mest$ convention confounded separate here: $\mest$ sets how fast a width shrinks, and
how close the lot sits to the acceptance boundary sets whether there is a width to shrink. The
ten commercial batches are determined at every lot; the process-development runs, at the
smallest $\mest$ and the most contested geometry, are the widest.}
\label{tab:bysource}
\begin{tabular}{@{}lrrcc@{}}
\toprule
source & lots & $\mest$ & median width & largest width\\
\midrule
commercial batches \cite{goodwin2018} & $10$ & $50$--$55$ & $0.000$ & $0.000$\\
routine lots \cite{vranic2008} & $30$ & $10$ & $0.000$ & $0.482$\\
batch vs.\ continuous \cite{ijpx2024} & $40$ & $21$ & $0.000$ & $0.837$\\
process development \cite{lyytikainen2024} & $18$ & $10$ & $0.030$ & $0.962$\\
\bottomrule
\end{tabular}
\end{table}

Table~\ref{tab:bysource} attributes it. All ten commercial batches are degenerate at the
resolution reported, and the wide cases are all at $\mest=10$ or $21$: a lot sampled at $55$
units supports a much tighter posterior on $(\mu,\sigma)$, and with $\ntest$ held at $10$ that
alone closes its interval. But a large $\mest$ does not buy a narrow interval on its own. Two
things set a width: $\mest$ fixes how concentrated the posterior is, and the lot's position
relative to the boundary fixes how much of that posterior the rule can separate. The $\mest$ column therefore orders
widths only when position is held roughly fixed, and across four sources it is not. The
$\mest=10$ and $\mest=21$ sources invert on their largest width, $0.482$ against $0.837$. The same mechanism
operates in \S\ref{sec:cehd}, where a design calibrated to a plug-in value of $0.85$
sits in the uninformative regime at the sample sizes the field uses (median $n=8$). That is a
statement about the calibrated design at those sizes, not a finding that field assessments at
their own limits carry those widths.

\subsection{Escalating the test size on the ninety-eight lots}\label{app:escalate}

Escalating the test size does not resolve the contested lots. The width depends on
two things: the posterior, which the calibration sample fixes, and the acceptance region,
which $\ntest$ and $k$ fix. Escalating the test changes only the region, so all it can do is change
which lots sit near a boundary. Lots carried inside the enlarged region become determined, and
the lots that remain contested are contested against a new boundary. Pushing the same
posteriors through the Stage-2 wedge at $\ntest=30$, $k=2.0$ carries seven further lots into
the clearly-passing band (eighty-three of ninety-eight, against seventy-six), while the median
width among the five still contested rises to $0.986$ from $0.769$. Only the mechanism
generalizes, the finite-$n$ residual of \S\ref{sec:unify} seen on real lots. Two caveats apply. First, the
Stage-2 evaluation is the wedge alone, without the individual-unit clause; at the standard
limit and the Normal operating points checked, that clause removes less than $0.001$ of
acceptance probability. Second, it is a
standalone test of thirty units rather than the second stage of the coupled sequence. That difference is not free
at a calibrated limit: reaching Stage~2 conditions on a Stage-1 failure, which
tilts $(\Xbar,s)$ adversely ($-0.076$ at the calibrated limit of \S\ref{sec:905} against
$-0.001$ at the standard's own limit; Appendix~\ref{app:twostage}). The Stage-2 figures are read as
a rule in the catalog, not as the chance a lot survives $\langle905\rangle$.

\section{The two-stage sequence, worked}\label{app:twostage}

$\langle905\rangle$ is sequential: ten units first ($k=2.4$); on failure, twenty more, with the
criterion applied to all thirty ($k=2.0$). The body treats the stages as standalone rules; this
appendix works the sequence itself, at an illustrative contested lot
$(\hat\mu,\hat\sigma)=(100,\,5.5)$ with a Normal parent, and asks what escalation buys. The
individual-unit clause is omitted throughout, and the parent is independent Normal, $\rho=0$.
Standalone acceptance probabilities are the one-dimensional integral of \S\ref{sec:closed}. The
coupled event (pass Stage~1, or fail it and pass Stage~2 on all thirty) is a three-dimensional
product quadrature over $(\bar x_{10},s_{10}^2,s_{20}^2)$. Within it, the Stage-2 pass set in $\bar x_{20}$
is an interval found by bisection. The quadrature converges to $2\times10^{-5}$ at the point estimate and
carries an error of order $10^{-3}$ at the faster order used inside the pushforwards; a
$4\times10^{5}$-lot simulation agrees to $10^{-4}$. Each pushforward is under the posterior
\eqref{eq:cspost} at $\rho=0$ and is integrated by two scrambled Sobol' sets, $2^{16}$ points for the
standalone rules and $2^{10}$ for the coupled event. Endpoints carry standard errors of at
most $0.011$.

\paragraph{Matched difficulty: escalation does not narrow the interval.} Calibrating each
stage's limit so that both point estimates equal $0.850$ (the body's control) isolates
information from difficulty:

\begin{center}\small
\begin{tabular}{@{}llcc@{}}
\toprule
what you know & what you ask & $95\%$ interval & width\\
\midrule
the ten Stage-1 units & Stage-1 rule ($n=10$) & $[0.05,\ 1.00]$ & $0.95$\\
those ten $+$ twenty more & Stage-2 rule ($n=30$) & $[0.09,\ 1.00]$ & $0.91$\\
those same thirty & Stage-1 rule again & $[0.35,\ 0.99]$ & $0.63$\\
\bottomrule
\end{tabular}
\end{center}

Rows one and two are the escalation: three times the data, four hundredths of narrowing ---
the posterior tightens by $1.8\times$ while the sharper test's amplification grows by
$1.6\times$ (\S\ref{sec:unify}), and the two very nearly cancel. Row three is the control that
locates the information: the same thirty units, aimed at the \emph{unchanged} Stage-1
question, cut the width from $0.95$ to $0.63$. The twenty extra units carry real information; the sequence spends it on
asking a harder question.

\paragraph{The coupled rule is almost its own second stage.} At the real constants ($L=15$ at
both stages) the compound event --- pass Stage 1, or fail it and pass Stage 2 on all thirty ---
has probability $0.9965$ at the point estimate, against $0.9959$ for Stage 2 standalone: the
union is dominated by its more forgiving branch, and Stage 1 functions as an early exit that
saves twenty units when the lot is clearly conforming while almost never changing the final
verdict. \emph{This holds at the standard's constants and should not be carried to a
calibrated one.} The quantity it rests on is the selection effect of conditioning on Stage-1
failure, and that effect grows as the limit tightens: $-0.009$ here, $-0.001$ at the flagship's
spread under $L=15$, but $-0.076$ at the calibrated $L=11.32$ of \S\ref{sec:905}, where a
standalone $0.850$ falls to $0.774$ once the failure that produced the second stage is
conditioned on.

\paragraph{The rescue decision inherits the full width.} Mid-sequence is the one moment a
probability would change an action. Suppose Stage 1 has just failed at
$(\bar x_{10},s_{10})=(100,\,6.4)$, barely over the threshold $15/2.4=6.25$. The probability
that the twenty additional units rescue the lot has point estimate $0.948$, but its
$95\%$ interval is $[0.01,\ 1.00]$, width $0.99$. That interval uses the posterior those ten units support,
which is the same-lot posterior of \S\ref{sec:905} because the twenty units share the lot. The earlier
version's $[0.15,\ 1.00]$ came from independent Normal and folded-Normal draws that were not
that posterior. Conditioning on the failure adds nothing beyond those ten values, since the failure is a
function of them. The verdict hinges on whether the true $\sigma$ sits below the Stage-2 cliff near
$15/2.0=7.5$, and ten units place $\sigma$ at $6.4$ give or take thirty percent --- straddling
it. After all thirty units are in, the coupled interval is $[0.54,\ 1.00]$, width $0.46$:
genuinely narrow, and informing no decision, because none remains. The hump of
\S\ref{sec:reallots} read along the sequence in time: the probability becomes dependable
exactly when nothing depends on it.

\paragraph{What escalation does buy.} Accuracy of the verdict, not certainty about it. Over
what range of true $\sigma$ is each calibrated test's outcome genuinely a lottery (acceptance
probability between $0.95$ and $0.05$)? Stage 1: $\sigma\in[4.9,\,10.6]$, a doubt band of
$5.7$; Stage 2: $[5.1,\,7.9]$, a band of $2.8$ --- half as wide. A lot at $\sigma=8$ slips past
Stage 1 twenty-six percent of the time and past Stage 2 four percent; a lot at $\sigma=5$ is
wrongly failed six percent against three. Escalation is the standard's designers buying a
sharper instrument, and the purchase is delivered --- to the population of lots. What it cannot
deliver, by the cancellation of \S\ref{sec:unify}, is a more predictable verdict on the lot in
hand.

\section{The summary likelihood from the engine}\label{app:nnlik}

Section~\ref{sec:skewed} needs a posterior for a non-Normal parent from published summaries
alone. This appendix records how the evaluator of \cite{paper1} supplies it and what the
construction was checked against.

\paragraph{Construction.} Let $X=\mu+\sigma Z$ with $Z$ standardized to mean zero and unit
variance and its shape fixed. For $n$ units the engine returns the joint density $f_T$ of
$(T_1,T_2)=(\sum Z_i,\sum Z_i^2)$ on a grid. The standardized summaries are
$\bar Z=T_1/n$ and $S_Z^2=(T_2-T_1^2/n)/(n-1)$, so their joint density is
\begin{equation}\label{eq:fzs}
f_n(\bar z,s_z)=2n(n-1)\,s_z\;f_T\big(n\bar z,\ (n-1)s_z^2+n\bar z^2\big),
\end{equation}
and since $(\bar X,S)=(\mu+\sigma\bar Z,\sigma S_Z)$ the likelihood of an observed pair
$(\bar x,s)$ from $m$ units is \eqref{eq:sumlik}. Under $\pi(\mu,\sigma)\propto1/\sigma$ the
posterior density is proportional to $\sigma^{-3}f_m((\bar x-\mu)/\sigma,\,s/\sigma)$; for a
Normal $Z$ this is the product of the inverse-chi-square law for $\sigma^2$ and the conditional
Normal law for $\mu$, which is the check below. The posterior is tabulated on a
$481\times401$ grid in $(\mu,\ln\sigma)$, $\mu$ within $\pm16\,s/\sqrt m$ of $\bar x$ and $\sigma$
within $s\times[0.15,\,8]$, $f_m$ read off the engine grid by bilinear interpolation, and
normalized numerically. The acceptance probability
$g(\mu,\sigma)=\Prob[(\mu+\sigma\bar Z,\sigma S_Z)\in A]$ is the standardized density for the
test size $n_t$ summed over the cells that the transformed region admits, tabulated once per
shape and test size on a common $(\mu,\sigma)$ grid ($\mu$ in steps of $0.125$ on $[85,115]$,
$241$ log-spaced values of $\sigma$ on $[0.5,30]$) and interpolated bilinearly. Grids are
$N=2048$; the characteristic function of a bounded or skewed $Z$ is integrated over its density
with $2000$ Gauss--Jacobi or Gauss--Legendre nodes rather than represented by point masses.
Shapes: Normal; the standardized Beta$(1.5,8)$, skewness $1.057$, support $[-1.40,\,7.48]$
standard deviations; the skew-Normal with shape parameter $2.174$, skewness $0.500$. No
correlation is carried; $\rho=0$ throughout.

\paragraph{Checks.}
\begin{center}\small
\begin{tabular}{@{}p{0.62\linewidth}p{0.34\linewidth}@{}}
\toprule
check & result\\
\midrule
Normal shape: posterior quantiles ($2.5$, $50$, $97.5\%$) of $\sigma$ and $\mu$ from the grid
against the inverse-chi-square and Student-$t$ closed forms at $(100,5)$, $m=30$ and $m=10$ &
maximum discrepancy $1.2\times10^{-3}$\\
Normal shape: Stage-1 interval from the grid against the Sobol' evaluation of
\S\ref{sec:reallots} at four $(\bar x,s,m)$ & $\le1.2\times10^{-3}$ on either endpoint\\
Normal shape: the ninety-eight lots against Table~\ref{tab:reallots} & $\le2.3\times10^{-3}$\\
Interpolated $f_m$ against the closed form, posterior-weighted relative error & $2.5\times10^{-5}$
($m=30$), $7.3\times10^{-5}$ ($m=10$)\\
Beta shape: $g$ at three $(\mu,\sigma)$ against $2\times10^6$ simulated lots & within one
standard error at each\\
Beta shape: $f_m$ integrates to one ($m=10,30$) & $1.000000$\\
Beta shape: $f_{10}$ at three points against a kernel estimate from $10^6$ simulated pairs &
within $1$, $1$ and $4\%$\\
Beta characteristic function, $1000$ against $4000$ quadrature nodes & $g$ unchanged to
$10^{-14}$\\
$N=1024$ against $2048$ (Beta, Stage 1) & $2.7\times10^{-4}$\\
Posterior grid doubled and box widened, two lots & endpoints move $\le1.4\times10^{-3}$\\
Skew-Normal shape: $g$ at two points against simulation & within one standard error\\
\bottomrule
\end{tabular}
\end{center}
Pipeline error on an endpoint is therefore about $2\times10^{-3}$, two orders below the effects
reported in \S\ref{sec:skewed}.

\paragraph{The shape, calibrated.} Two sources publish per-unit values: the eighteen
process-development runs of ten tablets in the supplementary tables of \cite{lyytikainen2024},
eleven passing and seven failing the acceptance value, and the reference contents of the
Copenhagen tablet set \cite{dyrby2002}, three hundred and ten tablets of four types at three
production scales. The object identifiers of the latter carry a batch prefix, and the archive is
thirty-one batches of ten tablets, not four groups of thirty: seven full-scale batches, twelve
pilot batches and twelve laboratory calibration batches. The three calibration batches of a
tablet type sit at means spanning twenty-five to thirty per cent of their average, a designed
concentration ladder, while their within-batch relative standard deviations are $0.7$ to
$2.4\%$; grouping those tablets by type manufactures a $12$ to $13\%$ spread that no batch has.
The pilot batches, whose three batches of a type agree to about $0.1\%$ in mean and whose
within-batch relative standard deviations are $2.0$ to $3.5\%$, are the production replicates
here. The primary set is thirty batches and three hundred units, and the seven full-scale
batches enter as a sensitivity analysis.

The within-batch standardized residuals $r=(x-\bar x)/s$ are pooled. They satisfy
$\sum_i r_i=0$ and $\sum_i r_i^2=n-1$ within each batch, so they are dependent and their
marginal law is not the parent's; a product of parent densities over them is a pseudolikelihood
and a $\chi^2_1$ cutoff on it is not calibrated. The estimator is therefore calibrated as a
whole: for each candidate skewness, four thousand datasets with the actual batch design are
drawn from a skew-Normal of that skewness, the same pooled statistic is formed, and the $2.5$,
$50$ and $97.5$ per cent quantile curves are inverted at the observed value. The resulting
interval has $95\%$ coverage by construction under that model. An independent replicate of the
whole calibration moves the estimate by $0.004$ and the lower limit by $0.004$.
\begin{center}\footnotesize\setlength{\tabcolsep}{4pt}
\begin{tabular}{@{}lccc@{}}
\toprule
& runs, $18\times10$ & pilot, $12\times10$ & pooled, $30$ batches\\
\midrule
raw pooled skewness (excess kurtosis) & $+0.651$ ($+0.71$) & $+0.073$ ($-0.26$) & $+0.420$ ($+0.32$)\\
batch bootstrap, $95\%$ & $[+0.13,\,+1.21]$ & $[-0.34,\,+0.41]$ & $[+0.06,\,+0.79]$\\
null band under Normality, $95\%$; $p$ & $[-0.27,\,+0.27]$; $<10^{-4}$ & $[-0.33,\,+0.33]$; $0.66$ & $[-0.21,\,+0.21]$; $0.0002$\\
calibrated skewness ($95\%$) & $+0.99$ $[+0.70,\,+0.99]$ & --- & $+0.74$ $[+0.38,\,+0.99]$\\
\bottomrule
\end{tabular}
\end{center}
The upper limits marked $0.99$ run into the skew-Normal family's own bound of $0.995$: these
data do not bound the skewness above within that family. Adding the seven full-scale batches,
whose pooled residual skewness is $-0.765$, moves the pooled statistic to $+0.196$ and the
calibrated estimate to $+0.365$ with interval $[+0.02,\,+0.68]$, so the answer depends on the
inclusion rule as much as on the calibration.

A common shape is not tenable. The per-batch skewness values spread with standard deviation
$1.03$, against $0.56$ expected under a common shape at the calibrated estimate, with a $95\%$
simulated band of $[0.41,\,0.71]$; no one of four thousand simulated designs reached the
observed spread. The skew is carried by the seven failing, high-dispersion runs, whose raw
pooled skewness is $+1.86$; the eleven passing runs alone give $-0.12$ and the pilot batches
$+0.07$, the latter inside its null band at $p=0.66$. Per-batch skewness at $n=10$ carries a
standard error of $0.77$; nine of thirty batches fail a Shapiro--Wilk test at $0.05$, with
median $p$ $0.39$. A raw-unit skew-Normal likelihood with batch-specific location and scale and
a common shape puts its profile maximum at the family boundary, which is what an incidental-parameter
problem with sixty nuisance parameters and ten units per batch would do; it is reported as a
cross-check and not as an estimate. What these data support is that the parent is not Normal and
that its shape travels with the process, not a single transferable routine-lot shape. The shapes
carried in \S\ref{sec:skewed} are accordingly declared scenarios spanning the calibrated range:
the skew-Normal at $0.74$, the ends of its interval and the inclusion-sensitivity value
$0.365$; the stress case, Beta$(1.5,8)$ at $1.06$, lies just above that range and below the
failing-run subset.

\paragraph{The lots at the calibrated shapes.} At skewness $0.74$, Stage 1, $L=15$, the $2.5\%$
endpoint sits below the Normal one in eighty-three of ninety-eight lots, by at most $0.130$,
fifteen lots above $0.05$ and eight above $0.10$, median $5\times10^{-5}$; the width moves by up
to the same $0.130$, twenty-two lots above $0.05$ and nine above $0.10$. At $0.38$ the largest
endpoint move is $0.060$ with four lots above $0.05$ and none above $0.10$; at the
inclusion-sensitivity value $0.365$ it is $0.057$ with three; at $0.99$ it is $0.183$ with
twenty-two lots above $0.05$. Three lots change band at $0.74$, as under the stress case. On the
spread sweep at $m=30$ the lower-endpoint shift at $\hat\sigma\approx3.6$ is
$-0.069$ at the calibrated shape, $-0.104$ at $0.99$ and $-0.129$ under the stress case, and the
width shift at the same point is $+0.069$, $+0.104$ and $+0.129$.

\paragraph{The flagship-type point.} At $(\bar x,s,m)=(100,5,30)$ under the Stage-2 rule with the
Normal-calibrated limit and $\rho=0$, the plug-in, posterior mean and $95\%$ interval are
$0.850$, $0.713$, $[0.094,\,0.999]$ for the Normal shape; $0.813$, $0.690$, $[0.083,\,0.998]$
for the Beta shape; $0.837$, $0.704$, $[0.091,\,0.998]$ for the skew-Normal. Recalibrating each
shape to its own $0.850$ gives $[0.104,\,0.999]$ (Beta, $L=11.583$) and $[0.099,\,0.998]$
(skew-Normal, $L=11.404$). The posterior $97.5\%$ quantile of $\sigma$ is $6.72$ (Normal),
$7.22$ (Beta) and $6.86$ (skew-Normal): the skewed likelihood admits a larger true spread behind
the same reported $s$, which is the mechanism of the endpoint shift in \S\ref{sec:skewed}. At
this contested point the shift is small, $0.01$; it is the tight-spread, apparently-passing
region that carries it.

\paragraph{The largest moves on the lots.} Under the Beta shape at Stage 1, $L=15$, the six
largest downward moves of the $2.5\%$ endpoint relative to the Normal posterior, with the lot's
reported summary, sample size and the two endpoints: BJBMS-L20 ($100.84\pm2.34$, $m=10$)
$0.953\to0.733$; BJBMS-L10 ($99.95\pm2.80$, $10$) $0.785\to0.571$; IJPX-F15-cont
($103.0\pm2.8$, $21$) $0.923\to0.711$; IJPX-F13-batch ($100.7\pm3.4$, $21$) $0.849\to0.663$;
IJPX-F03-cont ($103.7\pm2.5$, $21$) $0.959\to0.774$; IJPX-F04-cont ($103.15\pm3.6$, $21$)
$0.592\to0.413$. Every one is a lot whose Normal plug-in sits at or near one. The upper endpoint
moves by at most $0.180$ and the plug-in by at most $0.100$ across the ninety-eight.

\subsection{The parent shape of \S\ref{sec:skewed}: sweep and lot comparisons}\label{app:shape}

The construction, its checks and the shape estimate are above; this subsection gives the full
readings of Figure~\ref{fig:oc} and of the lot comparison that \S\ref{sec:skewed} summarizes.

\begin{figure}[tbp]\centering
\includegraphics[width=0.97\linewidth]{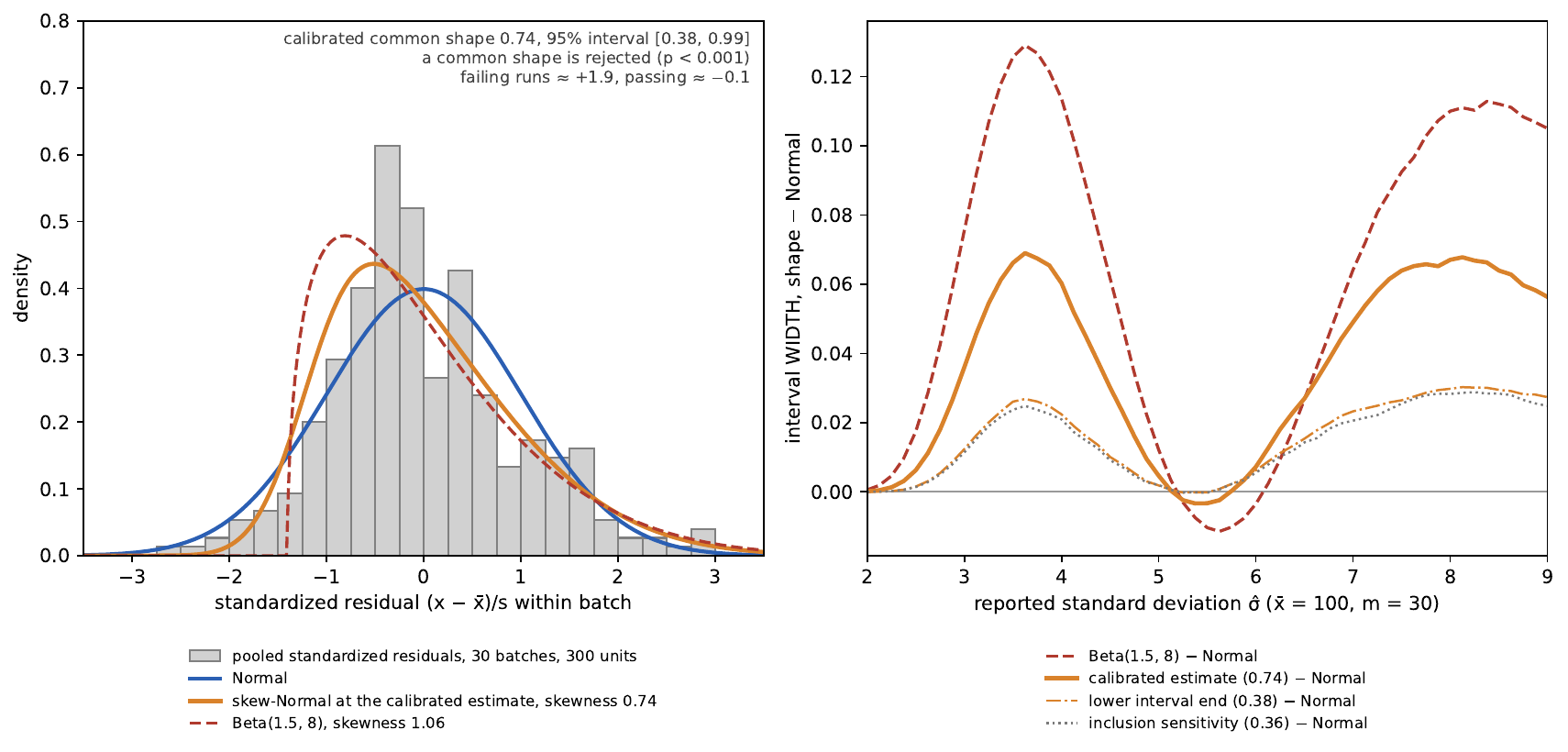}
\caption{The parent shape, calibrated. \emph{Left:} within-batch standardized residuals pooled
over thirty batches and three hundred tablets, with the Normal density, the skew-Normal at the
calibrated common-shape skewness $0.74$, interval $[0.38,\,0.99]$, and the
stress-case Beta$(1.5,8)$. A common shape is itself rejected across batches, so the curve is a
scenario and not a fitted parent. \emph{Right:} the interval \emph{width} along the spread sweep
of Figure~\ref{fig:oc}, each shape minus the Normal: the calibrated shape, the lower end of its
interval, the inclusion-sensitivity value, and the stress case. The width is not shape-free.}
\label{fig:shape}
\end{figure}

Figure~\ref{fig:oc} sweeps the Stage-1 rule at its own limit across the operating
characteristic with $m=30$, the posterior under each parent's own likelihood carried along, for
the Normal and the stress case; Figure~\ref{fig:shape} adds the calibrated shape and the ends of
its interval on the spread sweep. The figures support three readings. First, the \emph{summary} width statistics are nearly
shape-free: $0.62$ and $0.64$ at target for the Normal and the stress case, and maxima of $0.74$
and $0.75$ on the two sweeps. The three-regime pattern of Table~\ref{tab:reallots} therefore survives
the parent. Second, that similarity is a statement about the curves' peaks and not about agreement at a
given operating point: compared \emph{at the same reported summaries}, the width moves by up to
$0.129$ on the spread sweep under the stress case and $0.069$ at the calibrated shape, both at
$\hat\sigma\approx3.6$, against at most $0.023$ and $0.013$ on the location sweep. Matching
maxima do not imply pointwise agreement, and conditioning on each model's own probability band
changes which lots are being compared. Third, the \emph{endpoints} move with the skew, by an amount
that depends on which regime the lot is in. At the calibrated skewness the lower endpoint on
the spread sweep falls at most $0.069$ below the Normal one, $0.104$ at the upper end of the
interval, at $\hat\sigma\approx3.6$; under the stress case it falls $0.129$ there, and the upper
endpoint rises at large $\hat\sigma$. On the mean sweep the differences stay
within $0.03$ with mixed sign. The mechanism is the likelihood. A right-skewed sample spread
has a heavier right tail, so the same reported $s$ admits a larger true $\sigma$: the posterior
$97.5\%$ quantile of $\sigma$ at $m=30$ moves from $6.72$ to $6.97$ at the calibrated shape and
to $7.22$ under the stress case. Stated precisely, the summary posterior sets $\sigma=s/S_Z$, so
the upper tail of $\sigma$ is governed by the \emph{lower} tail of the standardized sample
spread, and it is the joint summary law that moves, not a one-sided tail of $s$ by itself; the
future test's spread is heavier-tailed in the same way.

On the ninety-eight lots the message is the same two-regime one, now read on both the endpoint
and the width; the calibrated shapes are reported above. Under the
stress case, Table~\ref{tab:beta}, both move in one direction and by much more: the endpoint
falls below the Normal one in eighty-one of ninety-eight lots, by more than $0.10$ in sixteen
and by up to $0.22$, at a lot reporting $100.8\pm2.3$ from ten units whose Normal lower endpoint
of $0.95$ becomes $0.73$, and on the same lot the interval widens from $0.047$ to $0.267$, a
factor of $5.7$. Across the ninety-eight, the paired width increase reaches $0.221$ and exceeds
$0.10$ on nineteen lots, with a median of $0.0004$ and a minimum of $-0.120$: the width is not
shape-free, and its move is not one-signed. The effect is concentrated where the Normal plug-in
sits against the ceiling at a tight reported spread: among the clearly-passing lots the number
carrying a width above $0.10$ nearly doubles, from eight to fifteen, the failing band's median
width rises from $0.03$ to $0.15$, and three lots change band, while in the contested band the
median lower endpoint moves by $0.003$ and the median width by $0.03$. The stress case is not
gratuitous: it is the shape the data show where content uniformity is already failing, and a
lot whose plug-in sits against the ceiling has, by that very fact, given no evidence about
which regime it is in. This is the interval counterpart of the exact-evaluation paper's level error. On a disturbed process the two
are of the same order ($0.22$ against $0.10$ for the plug-in on these lots) and coincide in location;
on a routine process both are small. The endpoint shift is
mostly one-signed under right skew, so a margin that depends on the reported spread would repair
much of it. No margin repairs the widening, which a one-sided allowance does not touch. Nor does a margin
repair the fact that the size of both is set by a shape that ten to thirty units do
not identify, and that the shape is observed to change with the process.

\begin{table}[h]\centering\small
\caption{Table~\ref{tab:reallots} recomputed under the stress-case likelihood, Beta$(1.5,8)$,
Stage 1 at $L=15$, the ninety-eight lots banded by each parent's own plug-in value; the Normal
values of Table~\ref{tab:reallots} are repeated for comparison. At the calibrated shape $0.74$
the same table moves in the same direction and about as far: $73$ lots in the top band against
$76$, nine of them carrying a width above $0.10$ against eight, and a bottom-band median width
of $0.099$ against $0.032$.}
\label{tab:beta}
\begin{tabular}{@{}lcccc@{}}
\toprule
plug-in band & lots, Beta / Normal & median width & largest width & lots with width $>0.10$\\
\midrule
$>0.99$ & $73$ / $76$ & $<0.001$ / $<0.001$ & $0.645$ / $0.692$ & $15$ / $8$\\
$0.80$--$0.99$ & $10$ / $7$ & $0.711$ / $0.725$ & $0.933$ / $0.934$ & $10$ / $7$\\
$0.20$--$0.80$ & $5$ / $5$ & $0.836$ / $0.812$ & $0.961$ / $0.962$ & $5$ / $5$\\
$<0.20$ & $10$ / $10$ & $0.150$ / $0.032$ & $0.413$ / $0.404$ & $6$ / $2$\\
\bottomrule
\end{tabular}
\end{table}

\section{The four rules in standardized units}\label{app:geomstd}

Figure~\ref{fig:geom} draws each rule in its own natural units, which keeps the specification
limits recognizable but makes the regions incomparable: $\langle905\rangle$'s calibrated region
spans tens of label-claim percent while capability's spans a few. Figure~\ref{fig:geomstd}
redraws all four with both coordinates divided by $\sigma$ and every constant calibrated so that
all four accept with probability $0.85$ at their own $n$, the paper's usual control, which
fixes difficulty so that what remains is shape.

\begin{figure}[tbp]\centering
\includegraphics[width=\linewidth]{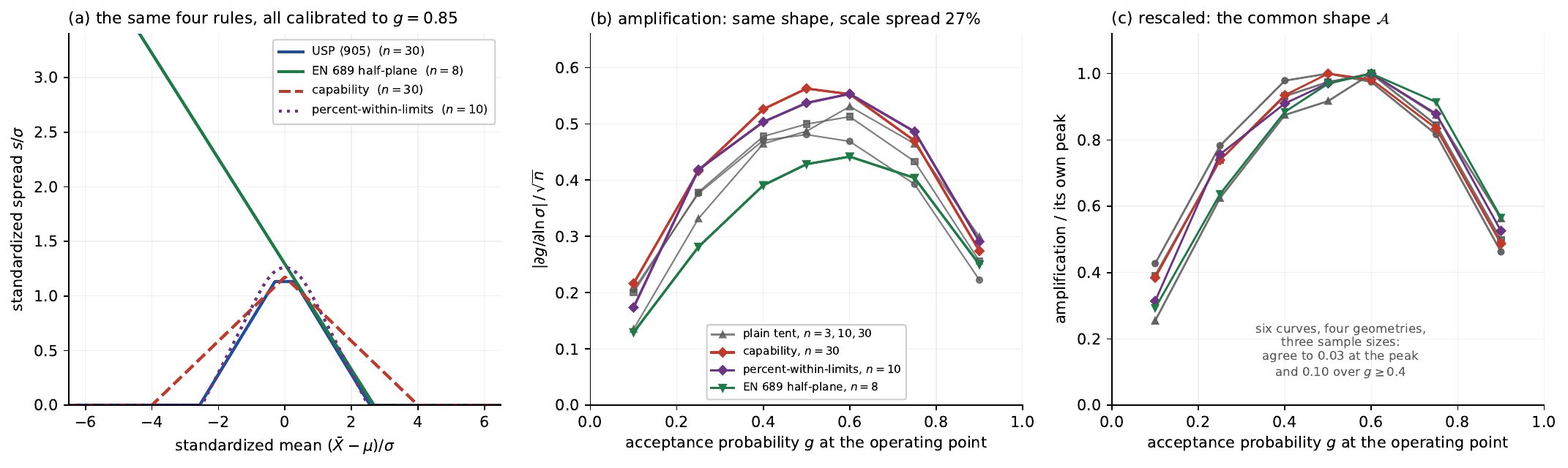}
\caption{\emph{(a)} The four rules of Figure~\ref{fig:geom} in standardized coordinates, each
calibrated to $g=0.85$ at its own $n$. Calibration pulls the three two-sided rules into near
coincidence over the range the sample actually occupies, despite their very different
definitions; the half-plane alone stays distinct, because it is one-sided.
\emph{(b)} The amplification $|\partial g/\partial\ln\sigma|$, divided by $\sqrt n$, against
the operating point $g$. The curves share a shape but not a scale: peaks run from $0.442$
(half-plane) to $0.563$ (capability), a spread of $27\%$.
\emph{(c)} The same curves divided by their own peaks. Six curves, four geometries and three
sample sizes agree to $0.03$ at the peak and to $0.10$ over $g\ge0.4$, degrading to $0.17$ at $g=0.1$, the same low-$g$ weakness that Table~\ref{tab:transfer} finds when the law is used to
predict an unseen geometry. This is the shape $\mathcal{A}$ of \eqref{eq:ampscale}.}
\label{fig:geomstd}
\end{figure}

As a check on the machinery, the half-plane in Figure~\ref{fig:geomstd} is evaluated with the
same cf/FFT joint density used for the polygons, rather than with the one-dimensional formula
\eqref{eq:ghalf} the exposure study uses. At the calibrated $\delta=2.6734$ the two agree to
$3.9\times10^{-5}$, and the calibrated $\delta$ itself reproduces the $\hat\delta=2.67$ quoted
independently in \S\ref{sec:cehd}.

\section{The incumbent tool for the exposure rule}\label{app:utltool}

The tooling in routine use for EN\,689 Annex F decisions evaluates the upper tolerance limit by
Monte Carlo, on the order of $8\times10^5$ simulated samples per computation \cite{utltool},
following the GUM framework \cite{gum2008}. It and the calculation here do not compute the same
quantity: the tool propagates analytical measurement uncertainty, models non-detects, retrieves
an actual exposure limit and returns a verdict, whereas this study deletes non-detects,
calibrates a limit per group, carries no measurement-error component, and returns a posterior
law for a repeated-sampling acceptance probability. What the two share is the inner object, the
probability that a sampled summary falls in an acceptance region, which \cite{paper1} evaluates
without simulation; whether an exact inner evaluator can carry the full measurement-error and
censoring workflow is not settled here.

\section{The width law}\label{app:widthlaw}

The detail behind \S\ref{sec:unify}: the closed constants of \eqref{eq:ampscale} and \eqref{eq:shape}, the conditions under which the cancellation \eqref{eq:cancel} holds, the sample-size experiment on the family's simplest member, and the width surface over both sample sizes.

\subsection{Closed constants of the width law}\label{app:constants}

For the one-sided rule of \S\ref{sec:cehd} both constants are closed, which turns
\eqref{eq:ampscale} from a fitted description into a derived one. With $V=Z+K_n\sqrt n\,S$,
$Z$ standard Normal, $S=(\chi^2_{n-1}/(n-1))^{1/2}$ and $K_n$ the tolerance factor,
$g(\delta)=F_V(\sqrt n\,\delta)$ exactly; the amplification is $q\,f_V(q)$ at the calibration
pin $q=F_V^{-1}(0.85)$, and the asymptotics of $V$ give, with $a=z_{0.95}$, $b=z_{0.85}$,
$c=z_{0.70}$ and $\kappa=1+a^2/2$,
\begin{equation}\label{eq:CD}
C=\frac{a\,\varphi(b)}{\sqrt\kappa}=0.250028,
\qquad
D=\frac{2\varphi(b)\bigl[(6+6a^2+a^4)\,b+3(2+a^2)\,c\bigr]}{3(2+a^2)^2}=0.266992 .
\end{equation}
The next coefficient is likewise closed, $-0.136661$ at order $n^{-1/2}$, and three terms hold
$6\le n\le60$ to between $3\%$ and $0.1\%$.

For the
pure spread rule of Proposition~\ref{prop:chi}, every constant in \eqref{eq:ampscale} and
\eqref{eq:shape} has a closed form. The amplification is $2q\,f_\nu(q)$, with
$q=\nu L'^2/(k^2\sigma^2)$ and $f_\nu$ the chi-square density. It peaks at $q=\nu$ exactly, hence
at $\hat g^*=F_\nu(\nu)=\tfrac12+1/(3\sqrt{\pi\nu})+O(\nu^{-3/2})$. The constants are
$C=1/\sqrt\pi=0.564190$, $D=0$, $E=-2/(3\sqrt\pi)=-0.376126$ and $t_1=\sqrt2/3=0.4714$; this
$t_1$ is the same coefficient that multiplies the width's first correction in
\eqref{eq:probitlaw}. The one-sided rule's constants differ
because its statistic mixes a mean and a spread term; the forms are the same.

The first factor, the amplification, can be read off directly
(Table~\ref{tab:ampscale}). To do so we sweep the limit of a \emph{plain tent}, the unclipped
wedge \eqref{eq:wedge} at $k=1$, under a Normal parent, recalibrated at each $n$ to visit the
same acceptance probabilities.

\begin{table}[h]\centering\small
\caption{Peak amplification against sample size, plain tent, Normal parent. Each entry is the
true maximum over the limit, located by root-finding on the derivative of a one-dimensional
quadrature; no expansion in $n$ is used, which matters at $n=3$.}
\label{tab:ampscale}
\begin{tabular}{@{}rccc@{}}
\toprule
$n$ & peak $|\partial g/\partial\ln\sigma|$ & attained at $g$ & peak$/\sqrt n$\\
\midrule
$3$ & $0.9070$ & $0.603$ & $0.5236$\\
$10$ & $1.6144$ & $0.546$ & $0.5105$\\
$30$ & $2.6541$ & $0.512$ & $0.4846$\\
\bottomrule
\end{tabular}
\end{table}

The amplification grows like $\sqrt n$, and the slow descent of the last column is an additive
term beside the leading one, not a different exponent:
\begin{equation}\label{eq:ampscaleS}
\Big|\frac{\partial g}{\partial\ln\sigma}\Big|
\;=\;C\,\sqrt n\;+\;D\;+\;O(n^{-1/2}),
\qquad\text{so}\qquad
\frac{1}{\sqrt n}\Big|\frac{\partial g}{\partial\ln\sigma}\Big|\;=\;C+\frac{D}{\sqrt n}+\cdots .
\end{equation}
The local log--log slope is then, to the order retained in \eqref{eq:ampscaleS},
\begin{equation}\label{eq:localslope}
\frac{d\log|\partial g/\partial\ln\sigma|}{d\log n}\;=\;\frac{1/2}{1+D/(C\sqrt n)},
\end{equation}
well below $\tfrac12$ across any range a sample-size table can cover. A fitted power law
therefore reports an exponent near $0.44$ with every appearance of precision. An additive term
is invisible to a fitted exponent (Appendix~\ref{app:widthn}).

\subsection{When the cancellation holds, and where it fails}\label{app:cancelholds}

The caveats to \eqref{eq:cancel} of \S\ref{sec:unify}, in full.

Equation~\eqref{eq:cancel} is a statement about
the \emph{rate}, which is why it is written with a proportionality sign. When the estimation and
test sample sizes are comparable, the product of a local derivative and a posterior scale gives
neither the shape nor the constant of a central interval. The exact limiting width of the pure-spread model,
already \eqref{eq:probitlaw}, is
$W_\infty=\Phi(z+ar)-\Phi(z-ar)$ with $z=\Phi^{-1}(\hat g)$, $r=\sqrt{\ntest/\mest}$ and
$a=\Phi^{-1}(0.975)$; the tangent expression $2ar\varphi(z)$ follows only for small $r$. At
$r=1$ they part company in both directions, $0.950$ against $1.564$ at $\hat g=0.5$ and $0.357$
against $0.104$ at $\hat g=0.99$. The displayed coefficient $C/\sqrt2=0.1768$ is the
rate constant of that product, evaluated at $g=0.85$, and is not a calibrated width; multiplying
it by a peak-normalized $\mathcal{A}$ would count the operating-point dependence twice.
Read \eqref{eq:cancel} as a local heuristic for how the width scales, and the exact formulas for
what it is. The qualitative claim survives intact: the
information a larger sample supplies about $\sigma$ is very nearly offset by the leverage
$\sigma$ thereby acquires, so along such a design no sample size escapes this. Appendix~\ref{app:widthn} shows it on the
family's simplest member, a pure-spread rule under a Uniform parent, and why that flat result is
a property of the stipulated regular law for the scale rather than of genuine Uniform
inference.

\paragraph{What the cancellation assumes.} Equation~\eqref{eq:cancel} holds the operating point
$g$ fixed while $n$ grows. Recalibration of the limit is one construction that achieves this,
and it is the one used here; it is sufficient, not necessary. A process sitting on, or
approaching, a \emph{fixed} decision boundary keeps the operating point interior by itself. Take
the pure-spread rule $s_t\le c$ of Proposition~\ref{prop:chi} with an observed $s=c$ and
$\mest=\ntest=n$: the posterior of $\Pacc$ is then $F_{\nu}(X)$ with $X\sim\chi^2_\nu$, which is
exactly uniform, so the central interval is $[0.025,\,0.975]$ and the width $0.95$ at every
sample size, while the plug-in $F_\nu(\nu)$ descends to $\tfrac12$. The same happens with random
data rather than one stipulated summary: with $\sigma$ known, a fixed limit $c$, a true mean at
$c$ and $\mest=\ntest=n$, the width is exactly $\Phi(a-Z)-\Phi(-a-Z)$ with
$Z\sim N(0,1)$ and $a=\Phi^{-1}(0.975)$, a law free of $n$ whose expectation is
$2\Phi(a/\sqrt2)-1=0.834$. So the width need not vanish under a fixed specification.
What the cancellation requires is that the posterior contract at the rate at which the rule
gains discrimination, with the operating point interior; recalibration and a fixed boundary are
two ways to arrange it. Away from a boundary, a coupled design with the limit held fixed does narrow. For a Normal
parent with $\sigma$ known and the limit one $\sigma$ above the calibration mean, the coupled
width runs $0.391$, $0.115$, $0.006$, $0.00001$ at $n=5,10,20,40$, because the test drives
$\mathcal{A}(g)\to0$. That narrowing is saturation rather than information, and the value of
$g$ at which it arrives is what tells the two apart.

A second route reaches the same cancellation with no uncertainty about the scale at all. With
$\sigma$ known and $\mu\mid\text{data}\sim N(\bar x,\sigma^{2}/\mest)$ against a one-sided test
of $\ntest$ units at a recalibrated limit, $g(\mu)=\Phi\bigl(z_0-\sqrt{\ntest/\mest}\,Z\bigr)$
with $Z\sim N(0,1)$, so along $\mest=\ntest$ the width is $\Phi(z_0+1.96)-\Phi(z_0-1.96)$ for
every $n$, an identity rather than a limit, because the standardized error is pivotal. The two
routes agreeing is what entitles \eqref{eq:cancel} to be read as a law about the design: what
cancels is a \emph{rate}, and any coordinate whose posterior contracts at the rule's own rate
cancels in the same way. On the location route the whole $(\mest,\ntest)$ plane is closed and
Appendix~\ref{app:wsurface} records the surface: in $\mest$ the width decreases strictly at
every test size and every limit, so more estimation data narrows the interval when the observed
summaries are held fixed. That is a statement about the surface and not a pathwise guarantee: a
new observation that carries $\bar x$ toward the limit can widen the interval even as the
posterior contracts. Ten observations with mean $-1$ against a limit at zero give width $0.115$;
an eleventh observation equal to $10$ moves the mean to zero and the width to $0.938$. In
$\ntest$ it rises to a ridge at the test size for which the process stands one standard error
from its limit and falls by saturation beyond it.

\paragraph{Where it fails.} The cancellation is a statement about regular families, whose
posteriors contract at $n^{-1/2}$. Where a parameter is estimable faster the evidence outruns
the test and the width does fall: the Uniform parent's support boundary is estimable at rate
$1/\mest$, and for a one-sided mean test along the coupled path the width decays like
$n^{-1/2}$, with $W\sqrt n$ rising to $1.52$ by $n=80$. That computation is location-only with a
known half-width, so it supplies a direction and a mechanism rather than settling
Figure~\ref{fig:widthn}; Open Problem~\ref{op:unifpost} records what remains.

\begin{figure}[tbp]\centering
\includegraphics[width=0.9\linewidth]{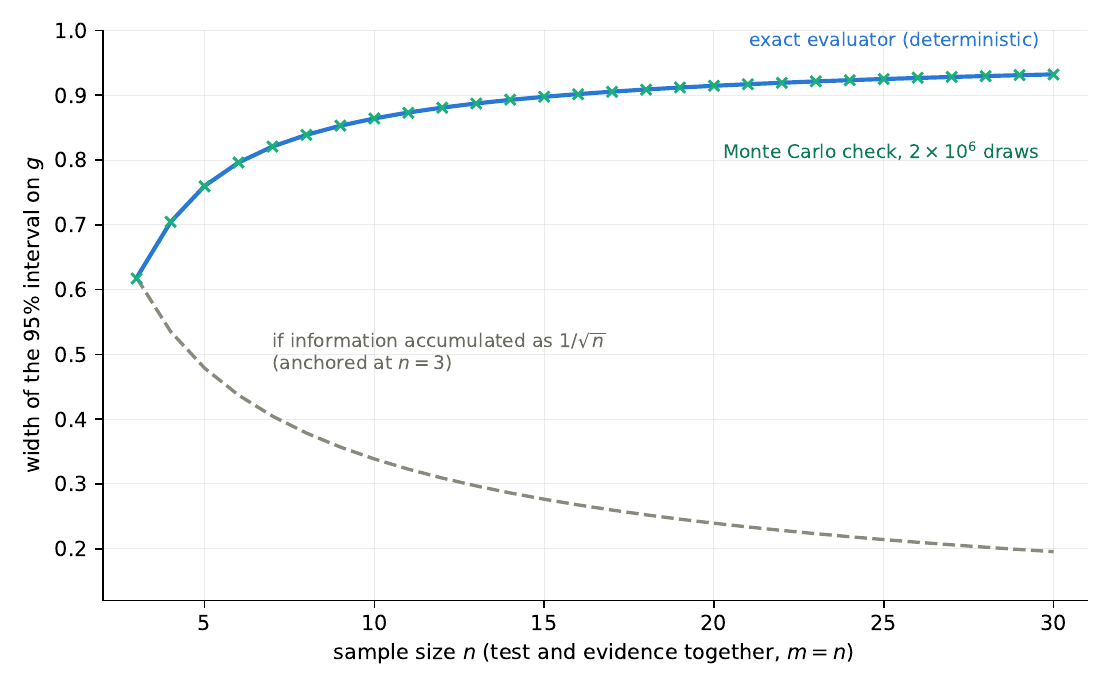}
\caption{A pure-spread rule calibrated to report $g=0.850$ at every $n$, under a stipulated regular law
for the scale, $\sigma\sim N(\hat\sigma,\hat\sigma^{2}/(2n))$ with $\ntest=\mest$. The $95\%$
interval does not narrow as the sample grows, running from $[0.38,\,1.00]$ at $n=3$ to
$[0.07,\,1.00]$ at $n=30$; the upward drift is law-dependent and not claimed. Dashed: the
$n=3$ width carried forward as $1/\sqrt n$, the naive expectation. Crosses: an independent
$2\times10^{6}$-draw Monte Carlo evaluation of the lower endpoint, agreeing within $2.5$
standard errors except at $n=3$, where the discrepancy is $1.0\times10^{-3}$; grid convergence
is in Appendix~\ref{app:widthn}.}\label{fig:widthn}
\end{figure}

\begin{table}[h]\centering\small
\caption{Why sampling harder does not narrow the interval. From $n=6$ to $n=60$ the posterior
spread of $\ln\sigma$ shrinks by a factor $3.4$ while the sensitivity of $g$ to $\ln\sigma$
steepens by a factor $2.7$; the product, which governs the interval width, falls by $20\%$. The
$\mathrm{sd}(\ln\sigma)$ column is the large-sample value $1/\sqrt{2(n-1)}$; the exact posterior
value $\tfrac12\sqrt{\psi_1((n-1)/2)}$ is $0.350$ at $n=6$ and approaches it from above.}
\label{tab:cancel}
\begin{tabular}{@{}rccccc@{}}
\toprule
$n$ & $\sigma/\hat\sigma$, $2.5$--$97.5\%$ & $\hat\delta$ &
$\mathrm{sd}(\ln\sigma)$, approx.\ & $\partial g/\partial\ln\sigma$ & product\\
\midrule
$6$  & $0.624$--$2.453$ & $2.901$ & $0.316$ & $-0.798$ & $0.2525$\\
$8$  & $0.661$--$2.035$ & $2.674$ & $0.267$ & $-0.908$ & $0.2428$\\
$12$ & $0.708$--$1.698$ & $2.441$ & $0.213$ & $-1.083$ & $0.2309$\\
$20$ & $0.760$--$1.461$ & $2.235$ & $0.162$ & $-1.349$ & $0.2188$\\
$30$ & $0.796$--$1.344$ & $2.116$ & $0.131$ & $-1.608$ & $0.2111$\\
$60$ & $0.848$--$1.220$ & $1.969$ & $0.092$ & $-2.185$ & $0.2011$\\
\bottomrule
\end{tabular}
\end{table}

\subsection{Escalation at every sample size, on the family's simplest member}\label{app:widthn}

Appendix~\ref{app:twostage} follows the $\langle905\rangle$ sequence between two sample sizes.
This appendix asks the same question along a whole axis, $n=3$ to $30$, on the simplest rule in
the family, chosen so that almost nothing numerical stands between the question and the answer.
The rule is pure spread: accept when $s^{2}\le c^{2}$. The parent is Uniform with standard
deviation $\sigma$, so that $g(\sigma)=F_n\!\big(c^{2}/(12\sigma^{2})\big)$ with $F_n$ the
sample-variance distribution of $n$ standard Uniform observations. Two structural gifts follow.
First, $g$ is monotone in $\sigma$, so the pushforward interval is $g$ evaluated at two
posterior quantiles and the outer integral is \emph{exact}; every approximation in this
appendix lives in the inner evaluation of $F_n$. Second, one characteristic-function inversion
per $n$ yields the whole function $F_n$ at once, so the calibration and both endpoints come
from a single deterministic computation. The result is Figure~\ref{fig:widthn}, placed with
the law it illustrates in \S\ref{sec:unify}: calibrated to report $0.850$ at every $n$, the
interval widens from $[0.38,\,1.00]$ at $n=3$ to $[0.07,\,1.00]$ at $n=30$.

\paragraph{The law imposed on $\sigma$, and what it is not.} The computation stipulates
$\sigma\sim N(\hat\sigma,\hat\sigma^{2}/(2n))$, i.e.\ a relative standard deviation
$1/\sqrt{2n}$. \emph{This is the regular large-sample law for a Normal scale parameter; it is
not derived from a Uniform likelihood and prior, and it is not a posterior for this
experiment.} The Uniform family has support that moves with its scale and is therefore
nonregular, so inference about $\sigma$ from a Uniform sample is not governed by the regular
$\sqrt n$ theory at all. We report the figure as a controlled experiment under a stipulated
regular uncertainty model, and the paragraph below states what changes when the law is instead
matched to the data-generating process.

\paragraph{The cancellation on one line, and the finite-$n$ correction.} Through
$Y\propto\sigma^{-2}$ a relative excursion $\tau$ in $\sigma$ displaces the rule's argument by
$2\tau$ relatively, and the resolving power of the test is the sampling standard deviation of
$s^{2}$. For a parent with excess kurtosis $\gamma_2$ that standard deviation is
\begin{equation}\label{eq:sdS2}
\mathrm{sd}(S^{2})=\sigma^{2}\sqrt{\frac{\gamma_2}{n}+\frac{2}{n-1}}
\;\xrightarrow[\ n\to\infty\ ]{}\;\sigma^{2}\sqrt{\frac{\gamma_2+2}{n}},
\end{equation}
which for the Uniform parent ($\gamma_2=-6/5$) has the familiar asymptotic form
$\sqrt{0.8/n}\,\sigma^{2}$. \emph{The asymptotic form is not adequate at the small
$n$ this figure starts from}: at $n=3$ it gives $0.516\,\sigma^{2}$ against the exact
$0.775\,\sigma^{2}$, an error of a third, confirmed by $4\times10^{6}$-draw simulation to four
digits at $n=3,10,30$. The excursion measured in units of the test's own resolving power is
therefore
\begin{equation}\label{eq:nfree}
\frac{2\times1.96\,\tau}{\sqrt{\gamma_2/n+2/(n-1)}},
\end{equation}
and everything turns on how $\tau$ is chosen. Under the stipulated law $\tau=1/\sqrt{2n}$ the
ratio is \emph{not} $n$-free: it rises from $2.07$ at $n=3$ through $2.74$ at $n=10$ to $2.97$ at
$n=30$, approaching $2\times1.96/\sqrt{1.6}=3.10$ only in the limit. Under a law matched to the
estimator the sample actually supplies --- $\tau=\mathrm{sd}(S^{2})/(2\sigma^{2})$ by the delta
method --- the same ratio collapses to exactly $1.96$ at every $n$, and the cancellation is not
merely asymptotic but exact.

\paragraph{Consequence for Figure~\ref{fig:widthn}, stated plainly.} The two laws differ by more
than a constant. Simulating $\mathrm{sd}(\ln s)$ directly under the Uniform parent gives
$0.603$ at $n=3$ and $0.0875$ at $n=30$, against $1/\sqrt{2n}=0.408$ and $0.129$: the stipulated
law is $0.68$ times the simulated spread at $n=3$, too narrow, and $1.47$ times it at $n=30$, too
wide. That is a drift by a factor of $2.2$ across the plotted range, and it runs in exactly the direction that would make
the interval appear to widen with $n$. \emph{We therefore do not claim that the interval for this rule widens with $n$.}
What the figure establishes, and what \eqref{eq:nfree} explains, is the weaker and still
substantive statement that the interval fails to narrow at the $1/\sqrt n$ rate the parameter
uncertainty itself obeys, because the test sharpens at very nearly the same rate; whether the
residual drift is upward, flat or mildly downward depends on the law placed on $\sigma$ and is
not settled by this experiment. Recomputing the curve under an estimator-matched law, and under
a Normal parent where the inverse-$\chi^{2}$ posterior is coherent, is required before the
direction of the drift can be asserted, and is listed in Appendix~\ref{app:open}. Two further
qualifications. The stipulated Normal law for the scale is untruncated and places probability
$0.007$ on a negative scale at $n=3$, a defect of the stipulation and not of the rule. And the
genuine posterior is not out of reach. For a known-center Uniform parent $\mathrm{U}[-a,a]$ with
prior $1/a$ and $M=\max_i|X_i|$, it is elementary: $\Prob(a\le A\mid D)=1-(M/A)^{m}$ for $A\ge M$,
contracting at rate $m^{-1}$. Pushing its $2.5\%$ and $97.5\%$ quantiles through the same
calibrated pure-variance rule, with $m=\ntest=n$ and plug-in $\hat a=M$, gives widths $0.72$,
$0.71$, $0.52$, $0.27$ and $0.14$ at $n=3$, $10$, $30$, $100$ and $300$, with leading term
$2\varphi(z_{0.85})\ln39/\sqrt{0.8\,n}=1.91/\sqrt n$. Under genuine Uniform inference the
interval narrows, so the flatness of Figure~\ref{fig:widthn} is a property of the regular law
stipulated for it and does not transfer. The exposure study of
\S\ref{sec:cehd} is not affected: its parent is lognormal, so $\ln$-exposure is Normal and the
regular posterior used there is the correct one.

\paragraph{Details displaced from \S\ref{sec:unify}.}
\emph{Why the exponent is misread.} In Table~\ref{tab:ampscale} the column peak$/\sqrt n$ falls
monotonically, $0.5236$, $0.5105$, $0.4846$; a monotone residual in the quantity one claims to
have removed is the signature of a missing term, and \eqref{eq:ampscale} identifies it as an
additive constant $D$ beside the $\sqrt n$ term. Because the local slope \eqref{eq:localslope}
lies below $\tfrac12$ over any sample-size table, a power-law fit over such a range returns an
exponent near $0.44$ and reports it with every appearance of precision.

\emph{The stipulated law in Figure~\ref{fig:widthn}.} The law $\sigma\sim
N(\hat\sigma,\hat\sigma^{2}/(2n))$ is the Normal-theory large-sample form, not a posterior
derived from the Uniform likelihood, whose family is nonregular; it drifts by a factor of $2.2$
relative to the sampling behavior of the Uniform sample standard deviation across the plotted
range, in the direction that manufactures widening. The flatness is explained by
\eqref{eq:nfree} under either law; the sign of the residual drift is not claimed.

\emph{Grid convergence of Figure~\ref{fig:widthn}.} Over the ladder $N=1024$ to $8192$ the width
moves by at most $9.3\times10^{-5}$ at $n=4$, $10$ and $30$, but by $3.2\times10^{-4}$ at $n=3$,
where successive halvings of the grid spacing reduce the movement by a factor of only $2.05$:
the leftmost point of the figure is the one place in the paper where the evaluator converges at
first order in $N$ rather than faster, and it is accordingly the least settled point plotted.
The effect is $0.05\%$ of the width at that $n$ and does not bear on the shape of the curve.
These figures are ladder spreads rather than differences between the two finest grids; the
latter would read $3.4\times10^{-5}$ and $1.1\times10^{-4}$, understating the movement about
threefold, and at $n=4$ by a factor of sixty, where two consecutive steps nearly cancel.

\emph{Design B against C.} In Table~\ref{tab:designs} the widths of B and C cross near forty
units; that is an accident of this operating point and carries no meaning.

\subsection{The width surface over both sample sizes}\label{app:wsurface}

This appendix records the width over the whole $(\mest,\ntest)$ design, rather than along one path
through it. On the location route the surface has a closed form, given here in
full; \S\ref{sec:unify} keeps the conclusion and Table~\ref{tab:designs}.

The two routes of \S\ref{sec:unify} are lines through a plane, and on the second, location route the
whole plane has a closed form. Keep $\sigma$ known and the mean unknown, and let the
rule carry a \emph{fixed} limit $c$ rather than a recalibrated one, so that the data report a
standardized margin $d=(c-\bar x)/\sigma$. Since $g$ is monotone in the posterior draw, the
central width is exact at every pair of sample sizes:
\begin{equation}\label{eq:wsurface}
W(\mest,\ntest)=\Phi\bigl(z_0+1.96\,r\bigr)-\Phi\bigl(z_0-1.96\,r\bigr),
\qquad r=\sqrt{\ntest/\mest},\quad z_0=d\sqrt{\ntest}.
\end{equation}
Two coordinates govern this surface and not one. The ratio $r$ sets how far the posterior spreads
along the operating characteristic, and $z_0$ sets where on that curve the spread is read; at a
fixed limit $z_0$ is itself carried by the test size, at rate $\sqrt{\ntest}$. The ratio alone
therefore does not price a design. At $d=0.25$ the diagonal $\mest=\ntest$ gives $0.876$, $0.799$,
$0.647$, $0.391$ at $n=10,20,40,80$ --- a design whose ratio never changes and whose width halves.
The width stands still only on a level set of $z_0$, which is exactly what recalibration
constructs, and \eqref{eq:cancel} is the same statement reached through the scale instead of the
location.

Read along its two axes the surface separates the two sample sizes, and they do not behave alike.
In $\mest$ it decreases strictly at every test size and every limit, since
$\partial W/\partial r=1.96\{\varphi(z_0+1.96r)+\varphi(z_0-1.96r)\}>0$ while $r$ falls in
$\mest$: estimation data narrows the interval at fixed observed summaries. In $\ntest$ it is not monotone,
because a larger test moves the spread and the operating point at once. For small $r$ the width is
$2\times1.96\,\varphi(d\sqrt{\ntest})\sqrt{\ntest/\mest}$, whose dependence on the test size is
$u\mathrm{e}^{-u^{2}/2}$ in $u=d\sqrt{\ntest}$, so the ridge and its height are elementary,
\begin{equation}\label{eq:ridge}
\ntest^{*}=d^{-2},\qquad
W^{*}=\frac{2\times1.96\,\varphi(1)}{|d|\sqrt{\mest}}=\frac{0.9485}{|d|\sqrt{\mest}} ,
\end{equation}
to leading order in $r$, the widest interval sitting at the test size for which the process
stands one standard error from its limit. Exactly, with $a=1.96/\sqrt{\mest}$, an interior ridge
exists only for $|d|>a$ and sits at $\ntest^{*}=\log\!\big((|d|+a)/(|d|-a)\big)/(2|d|a)$, which for
$d=0.25$ is $18.7$ at $\mest=160$ against the leading-order $16$. Below $\ntest^{*}$, a larger test buys leverage faster than it buys resolution, and the
interval opens. Above it, the rule saturates and the interval closes for no informational reason.
This is the collapse \S\ref{sec:unify} has already attributed to $\mathcal{A}(g)\to0$. The
elementary form \eqref{eq:ridge} governs once $\mest d^{2}\gtrsim10$, locating the peak height to $0.2\%$ there;
at $\mest d^{2}=2.5$ there is no interior ridge at all, the surface having saturated. The
saturation value depends on which side of the boundary the process sits: under the strict
inequality $|d|<a$ the surface tends to $W=1$ as the test grows, while exactly at $|d|=a$ one of
the two Normal arguments is identically zero and the limit is $W=\tfrac12$. The three tabulated
no-ridge rows satisfy the strict inequality and are unaffected. The leading-order ridge and the
exact finite-$\mest$ ridge need not fall between the same tabulated test sizes.

Table~\ref{tab:wsurface} prints the surface, and the paper's three designs are three curves drawn
on it: the recalibrated coupled path is a level set of $z_0$, the fixed-limit coupled path is the
diagonal that crosses the ridge and descends, and the fixed-test path of
Appendix~\ref{app:twostage} is a vertical line. Only the last is available to a laboratory working
under a published constant, and it is the one axis along which more measurement pays at fixed observed summaries.

\begin{table}[h]\centering\small
\caption{The width surface \eqref{eq:wsurface} at a fixed limit, $d=0.25$, which puts the leading-order ridge
\eqref{eq:ridge} at $\ntest^{*}=16$ ($18.7$ exactly at $\mest=160$). Rows are the sample that estimates the process, columns the
sample the rule prescribes. Reading down a column is always a gain; reading across a row is not,
and the turn appears between $\ntest=20$ and $\ntest=40$ once $\mest$ is large enough to carry an
interior peak. The first three rows have none, $|d|\le1.96/\sqrt{\mest}$ there, and the surface is climbing
toward $W=1$. Checked
against a $2\times10^{6}$-draw Monte Carlo evaluation at seven points, agreeing to
$7\times10^{-4}$, which is the size of the Monte Carlo quantile error itself at that draw count.}
\label{tab:wsurface}
\begin{tabular}{@{}rcccc@{}}
\toprule
& $\ntest=5$ & $\ntest=10$ & $\ntest=20$ & $\ntest=40$\\
\midrule
$\mest=10$ & $0.770$ & $0.876$ & $0.951$ & $0.990$\\
$\mest=20$ & $0.601$ & $0.709$ & $0.799$ & $0.883$\\
$\mest=40$ & $0.448$ & $0.537$ & $0.599$ & $0.647$\\
$\mest=80$ & $0.325$ & $0.392$ & $0.427$ & $0.421$\\
$\mest=160$ & $0.233$ & $0.282$ & $0.300$ & $0.269$\\
$\mest=320$ & $0.166$ & $0.201$ & $0.211$ & $0.176$\\
\bottomrule
\end{tabular}
\end{table}

\section{Where the uncertainty lives, in full}\label{app:decomp}

To attribute the width rather than merely report it, we make one coordinate of $\theta$
uncertain at a time, at the $\langle905\rangle$ operating point of \S\ref{sec:905}, and hold
the others at the plug-in.

\begin{table}[h]\centering\small
\caption{Interval width contributed by each coordinate, $\langle905\rangle$ Stage-2, clipped
acceptance value. Each row is the central $95\%$ of the pushforward with that coordinate alone
uncertain \emph{under its law in the posterior \eqref{eq:cspost}}, the others held at the
plug-in: the inverse-chi-square law for the within-lot spread $w$ at $\mu=100$, $\rho=0$; the
Normal law $N(\bar x,s^2/m)$ for $\mu$ at $w=s$, $\rho=0$; and the prior Uniform$[0,0.3]$ for
$\rho$ at $\mu=100$, $w=s$. The ratios are specific to those three ranges at this operating
point, not intrinsic sensitivities of $g$.}
\label{tab:decomp}
\begin{tabular}{@{}lcc@{}}
\toprule
uncertain coordinate & $95\%$ band & width\\
\midrule
$w$ alone & $[0.114,\ 0.999]$ & $0.885$\\
$\mu$ alone & $[0.691,\ 0.850]$ & $0.159$\\
$\rho$ alone & $[0.556,\ 0.845]$ & $0.289$\\
\bottomrule
\end{tabular}
\end{table}

The practical content of this section is the ordering $w>\rho\approx\mu$. It says where to
spend measurement effort, \emph{given uncertainty of the stated size in each coordinate}. A
one-at-a-time band prices a coordinate together with the range assumed for it, so the factor of
$3.1$ would change if a laboratory knew $w$ better or $\rho$ worse. Nor is the ordering
intrinsic to the coordinates. The $\langle905\rangle$ wedge charges spread at $2.4$ times the
rate it charges location, and a dead zone around target shelters location, so a rule at the far
end of the exact-evaluation paper's axis of effective spread weight would reverse the ordering
on the same posterior. The ordering is a statement about $\langle905\rangle$ and rules that
weight spread as it does, the half-plane of \S\ref{sec:cehd} among them. The $\rho$ row holds the within-lot spread $w$ fixed, because the calibration summaries
identify only $\sigma^2(1-\rho)$; holding the total spread fixed instead ranks dependence last by
two orders of magnitude, on a slice the data do not speak to.
The bands are not an additive decomposition, since quantiles do not add.
Establishing the ordering across operating points and priors is left open.

\subsection{The dependence row and the figure}\label{app:decompdetail}

\begin{figure}[tbp]\centering
\includegraphics[width=0.92\linewidth]{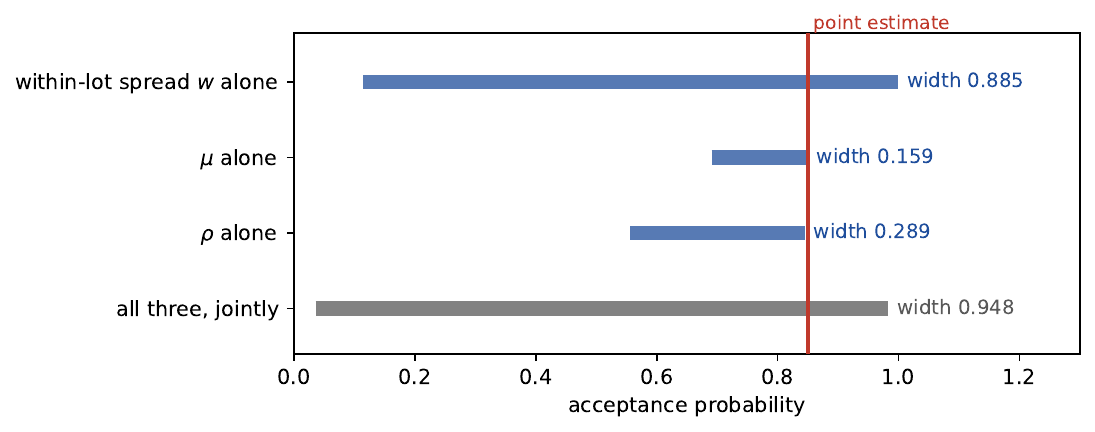}
\caption{Interval width contributed by each coordinate of $\theta$, all against the same point
estimate, each under its law in the posterior \eqref{eq:cspost}. The $\rho$ row holds the
within-lot spread $w$ fixed, because the calibration data identify only $\sigma^2(1-\rho)$ and
the total spread moves with $\rho$; the $w$ band is then $3.1$ times the $\rho$ band. The ratio
prices the stated ranges, not the coordinates in the abstract.}\label{fig:decomp}
\end{figure}

The $\rho$ row is easy to get wrong. Holding the total
spread $\sigma$ at $\hat\sigma$ while $\rho$ moves gives a width of $0.008$ and ranks
dependence last by two orders of magnitude. But that slice is not the one the data speak to.
The calibration summaries fix $\sigma^2(1-\rho)$, so moving $\rho$ at fixed $\sigma$ silently
changes the within-lot spread $w$, which the posterior leaves uncertain. Holding $w$ fixed
instead, dependence contributes as much as location. Then $g$ \emph{falls} with correlation,
monotonically over the stated range ($0.692$ at $\rho=0.15$ and $0.550$ at $0.30$), because a
fresh shared factor in the test lot widens the distribution of its mean. The bands are
not an additive decomposition. Quantiles do not add, so the three widths need not sum to the
joint width $0.948$. The $w$-alone floor, $0.114$, also lies above the joint interval's floor
of $0.035$: under the joint law $\Prob[g\le0.114]=0.088$ rather than $0.025$, because the
$\mu$ and $\rho$ uncertainty pushes mass below the floor that $w$ alone sets.

\section{Open problems}\label{app:open}

\begin{openproblem}[Tail quantiles of a pushforward for a non-Normal parent]\label{op:tails}
For the Normal studies the outer integral is now replicated randomized quasi-Monte Carlo with a
reported standard error, or root finding on a closed posterior, so tail endpoints are quoted to
the precision the mean supports. For a non-Normal parent the evaluator is a Fourier inversion
whose discretization error at a tail endpoint is measured rather than bounded
(Appendix~\ref{app:engine}); a certified quantile error for that case is open.
\end{openproblem}

\begin{openproblem}[A coherent parameter law for a nonregular parent]\label{op:unifpost}
Figure~\ref{fig:widthn} stipulates a regular $\sqrt n$ law for the scale of a Uniform parent,
whose likelihood is nonregular. Under the genuine posterior for a known-center Uniform scale,
which contracts at rate $m^{-1}$, the interval narrows (Appendix~\ref{app:widthn}), so the
flatness of that figure is a property of the stipulated law. What remains open is the
unknown-center Uniform posterior and, more generally, which nonregular families keep the
cancellation and which lose it.
\end{openproblem}

\begin{openproblem}[Left-censored parents]\label{op:censor}
Sixty percent of the personal samples in \cite{cehd} are non-detects. Censoring of the
calibration observations does not by itself make the latent parent discontinuous: a censored
likelihood estimates a continuous parent to which the evaluator applies as it stands. What is
open is the rule computed on censored reported values, a different acceptance problem; the
environmental and occupational literatures have an apparatus for
censoring, and joining the two is the natural next study.
\end{openproblem}

\begin{openproblem}[Real limits rather than calibrated ones]\label{op:limits}
Attaching real permissible exposure limits per substance would convert Table~\ref{tab:cehd}
from a controlled comparison into a statement about the compliance decisions actually made.
\end{openproblem}

\begin{openproblem}[Family uncertainty and the value of escalation]\label{op:family}
All widths here are conditional on the parent family; \S\ref{sec:skewed} estimates a skewness
from three hundred published tablets and sizes the interval effect at it and at a stress case,
but the skew is observed to change with the process, and a lot's own units do not say which
regime it is in. Uncertainty about the family adds a further component. Its size is governed by how far apart
the candidate families are in the rule's own moment currency \cite{paper1}, and it decays only
past the $n$ at which that feature becomes estimable. Escalation may therefore have a second
regime: its benefit is canceled within a family (\S\ref{sec:unify}), but it is valuable where
it selects among families. Budgeting samples against that threshold is a sequential design problem left open.
\end{openproblem}


\begin{thebibliography}{99}\small

\ifblind
\bibitem{paper1} Author (2026). Manuscript on exact acceptance probabilities.
\emph{Unpublished manuscript}; blinded for review, not peer reviewed at the time of writing.

\bibitem{n3companion} Author (2026). Manuscript on the $n=3$ sample variance.
\emph{Unpublished manuscript}; blinded for review, not peer reviewed at the time of writing.

\bibitem{chu2025} Author(s) (2025). Details withheld for anonymous review.

\bibitem{morrison2021} Author(s) (2021). Details withheld for anonymous review.
\else
\bibitem{paper1} Osan, R. (2026). Phasing out Monte Carlo: exact acceptance probabilities for
mean--spread release rules. \emph{arXiv preprint} arXiv:2609.21577.

\bibitem{n3companion} Osan, R., Chu, K., \& Yu, R. (2026). Closed forms and open obstructions:
the sample variance of three observations. \emph{arXiv preprint} arXiv:2609.25333.

\bibitem{chu2025} Chu, K.\,T., Osan, R., Tin, N., Bhattacharya, S., \& Lai, C. (2025).
Simulation-based content uniformity engineering for drug development and manufacturing.
\emph{International Journal of Pharmaceutics} \textbf{684}, 126132.

\bibitem{morrison2021} Morrison, H., Osan, R., Horstman, E., Lee, E., Ritchie, S., Payne, P.,
Scott, M.\,E., Geier, M.\,J., \& Wang, X. (2021). Correlation of drug substance bulk properties
to predict and troubleshoot the formulation of drug products: the API camera. \emph{Organic
Process Research \& Development} \textbf{25}, 1036--1046.
\fi

\bibitem{goodwin2018} Goodwin, D.\,J., et al.\ (2018). Real time release testing of tablet
content and content uniformity. \emph{International Journal of Pharmaceutics} \textbf{537},
183--192.

\bibitem{vranic2008} Vrani\'c, E., \& Uzunovi\'c, A. (2008). Study of the applicability of
content uniformity and dissolution variation test on ropinirole hydrochloride tablets.
\emph{Bosnian Journal of Basic Medical Sciences} \textbf{8}(2), 193--200.

\bibitem{lyytikainen2024} Lyytik\"ainen, J., Stasiak, P., Kubelka, T., Bogaerts, I., Wanek, A.,
Stynen, B., Holman, J., Ketolainen, J., Ervasti, T., \& Korhonen, O. (2024). Continuous direct
compression of a commercially batch-manufactured tablet formulation with two different
processing lines. \emph{European Journal of Pharmaceutics and Biopharmaceutics} \textbf{199},
114278 (supplementary Tables S7--S8).

\bibitem{ijpx2024} Bekaert, B., Janssen, P.\,H.\,M., Fathollahi, S., Vanderroost, D., Roelofs,
T., Dickhoff, B.\,H.\,J., Vervaet, C., \& Vanhoorne, V. (2024). Batch vs.\ continuous direct
compression: a comparison of material processability and final tablet quality.
\emph{International Journal of Pharmaceutics: X} \textbf{7}, 100226 (Table 6).

\bibitem{dyrby2002} Dyrby, M., Engelsen, S.\,B., N\o rgaard, L., Bruhn, M., \& Lundsberg-Nielsen, L.
(2002). Chemometric quantitation of the active substance in a pharmaceutical tablet using
near-infrared (NIR) transmittance and NIR FT-Raman spectra. \emph{Applied Spectroscopy}
\textbf{56}(5), 579--585 (reference content of the 310-tablet NIR set).

\bibitem{leblond2014} LeBlond, D., \& Mockus, L. (2014). The posterior probability of passing a
compendial standard, Part 1: Uniformity of dosage units. \emph{Statistics in Biopharmaceutical
Research} \textbf{6}(3), 270--286.

\bibitem{lewis2016} Lewis, R.\,A., \& Fan, A. (2016). Improved acceptance limits for ASTM
standard E2810. \emph{Statistics in Biopharmaceutical Research} \textbf{8}(1), 40--48.

\bibitem{cramer1946} Cram\'er, H. (1946). \emph{Mathematical Methods of Statistics}. Princeton
University Press.

\bibitem{serfling1980} Serfling, R.\,J. (1980). \emph{Approximation Theorems of Mathematical
Statistics}. Wiley.

\bibitem{bergum1990} Bergum, J.\,S. (1990). Constructing acceptance limits for multiple stage
tests. \emph{Drug Development and Industrial Pharmacy} \textbf{16}(14), 2153--2166.

\bibitem{bergum2015} Bergum, J., Parks, T., Prescott, J., et al.\ (2015). Assessment of blend
and content uniformity: technical discussion of sampling plans and application of ASTM
E2709/E2810. \emph{Journal of Pharmaceutical Innovation} \textbf{10}, 84--97.

\bibitem{jensen2006} Jensen, W.\,A., Jones-Farmer, L.\,A., Champ, C.\,W., \& Woodall, W.\,H. (2006).
Effects of parameter estimation on control chart properties: a literature review.
\emph{Journal of Quality Technology} \textbf{38}(4), 349--364.

\bibitem{epprecht2015} Epprecht, E.\,K., Loureiro, L.\,D., \& Chakraborti, S. (2015). Effect of
the amount of Phase I data on the Phase II performance of $S^2$ and $S$ control charts.
\emph{Journal of Quality Technology} \textbf{47}(2), 139--155.

\bibitem{faraz2015} Faraz, A., Woodall, W.\,H., \& Heuchenne, C. (2015). Guaranteed conditional
performance of the $S^2$ control chart with estimated parameters. \emph{International Journal
of Production Research} \textbf{53}(14), 4405--4413.

\bibitem{does2020} Does, R.\,J.\,M.\,M., Goedhart, R., \& Woodall, W.\,H. (2020). On the design of
control charts with guaranteed conditional performance under estimated parameters.
\emph{Quality and Reliability Engineering International} \textbf{36}(8), 2610--2620.

\bibitem{astm2709} ASTM International (2019). \emph{E2709-19: Standard Practice for Demonstrating
Capability to Comply with an Acceptance Procedure}. West Conshohocken, PA.

\bibitem{astm2810} ASTM International (2017). \emph{E2810-11 (Reapproved 2017): Standard Practice for Demonstrating
Capability to Comply with the Test for Uniformity of Dosage Units}. West Conshohocken, PA.

\bibitem{ohagan2005} O'Hagan, A., Stevens, J.\,W., \& Campbell, M.\,J. (2005). Assurance in
clinical trial design. \emph{Pharmaceutical Statistics} \textbf{4}(3), 187--201.

\bibitem{lavoue2019} Lavou\'e, J., Joseph, L., Knott, P., Davies, H., Labr\`eche, F., Clerc, F.,
Mater, G., \& Kirkham, T. (2019). Expostats: a Bayesian toolkit to aid the interpretation of
occupational exposure measurements. \emph{Annals of Work Exposures and Health} \textbf{63}(3),
267--279.

\bibitem{wsdot517} Muench, S.\,T., \& Mahoney, J.\,P. (2001). \emph{A Quantification and
Evaluation of WSDOT's Hot Mix Asphalt Concrete Statistical Acceptance Specification}. Washington
State Department of Transportation, Research Report WA-RD 517.1.

\bibitem{gum2008} Joint Committee for Guides in Metrology (2008). \emph{Evaluation of
measurement data --- Supplement 1 to the ``Guide to the expression of uncertainty in
measurement'' --- Propagation of distributions using a Monte Carlo method}. JCGM 101:2008;
also published as ISO/IEC Guide 98-3:2008/Suppl.1.

\bibitem{jiang2026a} Jiang, F., \& Yang, L. (2026). Finite-sample decision instability in
threshold-based process capability approval. arXiv:2603.11315.

\bibitem{jiang2026b} Jiang, F., \& Yang, L. (2026). Nonlinear amplification of finite-sample
uncertainty in capability-based decisions. arXiv:2605.06843.

\bibitem{usp905} United States Pharmacopeial Convention. \emph{USP General Chapter
$\langle905\rangle$ Uniformity of Dosage Units}. USP--NF, Rockville, MD (accessed 2026).

\bibitem{cehd} Occupational Safety and Health Administration. \emph{Chemical Exposure Health
Data}. U.S. Department of Labor. \url{https://www.osha.gov/opengov/health-samples}. Per-year sample
files for 2015, 2017, 2018 and 2019, retrieved 24 September 2026.

\bibitem{en689} European Committee for Standardization (2018). \emph{EN 689:2018, Workplace exposure: measurement of exposure by inhalation to chemical agents. Strategy for
testing compliance with occupational exposure limit values}, Annex F.

\bibitem{annweh2022} D'Errico, A., Houba, R., \& Kromhout, H. (2022). Is the new EN\,689 a better
standard to test compliance with occupational exposure limits in the workplace?
\emph{Annals of Work Exposures and Health} \textbf{66}(3), 412--415.

\bibitem{utltool} DOHSBase Online. \emph{EN\,689 compliance assessment: testing TWA measurements
against OELs with UTL statistics}. Computational methodology of R.\ Emonds, T.\ Scheffers and
P.\ van Balen for the Monte-Carlo treatment of measurement uncertainty.
\url{https://www.dohsbase.nl/en/knowledge-base/en-689-utl-compliance/}.

\bibitem{aiha} American Industrial Hygiene Association (2015). \emph{A Strategy for Assessing
and Managing Occupational Exposures}, 4th ed.

\end{thebibliography}
\end{document}